\pdfoutput=1
\RequirePackage{ifpdf}
\ifpdf 
\documentclass[pdftex]{sigma}
\else
\documentclass{sigma}
\fi

\numberwithin{equation}{section}

\newtheorem{Theorem}{Theorem}[section]
\newtheorem{Corollary}[Theorem]{Corollary}
\newtheorem{Lemma}[Theorem]{Lemma}
\newtheorem{Proposition}[Theorem]{Proposition}
 { \theoremstyle{definition}
\newtheorem{Definition}[Theorem]{Definition}

\newtheorem{Remark}[Theorem]{Remark} }

\usepackage{tikz}

\def\O{\Omega}

\def\Rscr{\mathcal R}
\def\Lscr{\mathcal L}
\def\Mscr{\mathcal M}

\def \Dscr{\mathcal D}

\def\wt{\widetilde}

\def\wh{\widehat}

\def\res{\mathop{\mathrm{res}}\limits_}

\def\le{\left}
\def\ri{\right}

\def\m{\mathop}

\def \s{\mathfrak s}
\def\ov{\overline}
\def\1{{\bf 1}}

\def \pa{\partial}
\def\C{{\mathbb C}}
\def\R{{\mathbb R}}
\def\N{{\mathbb N}}
\def\Z{{\mathbb Z}}

\def\mod{{\rm mod}}
\def\gt{\hat\gamma}
\def \hf{\frac{1}{2}}
\def \hg{\hat\gamma}
\def\part{\partial}
\def \qt{\frac{1}{4}}
\def\vr{\varrho}
\def \pa{\partial}

\def \ra{\rightarrow}
\def\C{{\mathbb C}}
\def \D{\Delta}

\def\a{\alpha}
\def\b{\beta}
\def\d{\delta}
\def\g{\gamma}
\def\k{\varkappa}
\def\l{\lambda}
\def\m{\mu}
\def\o{\omega}
\def\r{\rho}
\def\s{\sigma}

\def\t{\tau}

\def\e{\varepsilon}
\def\z{\zeta}

\def\L{\Lambda}
\def\R{{\mathbb R}}
\def\N{{\mathbb N}}
\def\h{{\mathbf h}}
\def\Z{{\mathbb Z}}

\def\th{\theta}
\def\Th{\Theta}
\def\star{*}

\begin{document}


\renewcommand{\thefootnote}{$\star$}

\renewcommand{\PaperNumber}{118}

\FirstPageHeading

\ShortArticleName{Orthogonal Polynomials with Complex Varying Quartic Exponential Weight}

\ArticleName{On Asymptotic Regimes of Orthogonal Polynomials\\ with Complex Varying Quartic Exponential Weight\footnote{This paper is a~contribution to the Special Issue on Orthogonal Polynomials, Special Functions and Applications. The full collection is available at \href{http://www.emis.de/journals/SIGMA/OPSFA2015.html}{http://www.emis.de/journals/SIGMA/OPSFA2015.html}}}

\Author{Marco BERTOLA~$^\dag$ and Alexander TOVBIS~$^\ddag$}

\AuthorNameForHeading{M.~Bertola and A.~Tovbis}

\Address{$^\dag$~Department of Mathematics and Statistics, Concordia University,\\
\hphantom{$^\dag$}~1455 de Maisonneuve W., Montr\'eal, Qu\'ebec, Canada H3G 1M8}
\EmailD{\href{mailto:marco.bertola@concordia.ca}{marco.bertola@concordia.ca}}

\Address{$^\ddag$~University of Central Florida, Department of Mathematics,\\
\hphantom{$^\ddag$}~4000 Central Florida Blvd., P.O.~Box 161364, Orlando, FL 32816-1364, USA}
\EmailD{\href{mailto:alexander.tovbis@ucf.edu}{alexander.tovbis@ucf.edu}}
\ArticleDates{Received March 11, 2016, in f\/inal form December 12, 2016; Published online December 27, 2016}

\Abstract{We study the asymptotics of recurrence coef\/f\/icients for monic orthogonal polynomials $\pi_n(z)$ with the quartic exponential weight $\exp [-N (\frac 12 z^2 + \frac 14 t z^4)]$, where $t\in {\mathbb C}$ and $N\in{\mathbb N}$, $N\to \infty$. Our goal is to describe these asymptotic behaviors globally for $t\in {\mathbb C}$ in dif\/ferent regions. We also describe the ``breaking'' curves separating these regions, and discuss their special (critical) points. All these pieces of information combined provide the global asymptotic ``phase portrait'' of the recurrence coef\/f\/icients of $\pi_n(z)$, which was stu\-died numerically in [\textit{Constr. Approx.} \textbf{41} (2015), 529--587, arXiv:1108.0321]. The main goal of the present paper is to provide a rigorous framework for the global asymptotic portrait through the nonlinear steepest descent analysis (with the $g$-function mechanism) of the corresponding Riemann--Hilbert problem (RHP) and the continuation in the parameter space principle. The latter allows to extend the nonlinear steepest descent analysis from some parts of the complex $t$-plane to all noncritical values of~$t$. We also provide explicit solutions for recurrence coef\/f\/icients in terms of the Riemann theta functions. The leading order behaviour of the recurrence coef\/f\/icients in the full scaling neighbourhoods the critical points (double and triple scaling limits) was obtained in [\textit{Constr. Approx.} \textbf{41} (2015), 529--587, arXiv:1108.0321] and [Asymptotics of complex orthogonal polynomials on the
 cross with varying quartic weight: critical point behaviour and the second
 Painlev\'e transcendents, in preparation].}

\Keywords{double scaling limit of orthogonal polynomials; asymptotics of recurrence coef\/f\/i\-cients; method of Riemann--Hilbert problem; nonlinear steepest descent analysis}

\Classification{33D45; 33E17; 15B52}

\renewcommand{\thefootnote}{\arabic{footnote}}
\setcounter{footnote}{0}

\vspace{-2mm}

{\small \tableofcontents}

\section{Introduction and main results}\label{sectintro}

In this paper we continue (see \cite{BT3}) our study of monic polynomials $\pi_n(z)$, orthogonal with respect to the quartic exponential weight
$e^{-N f(z,t)}$, where $f(z,t)=\hf z^2 + \qt t z^4$, $t\in \C$ and $N\in\N$, and the contour of integration consists of four rays in the directions of the exponential decay of~$e^{-N f(z,t)}$. The exponent $f(z,t)$ is also known as the potential.

In \cite{BT3} we obtained local (in $t$) asymptotics of the recurrence coef\/f\/icients $\a_n(N,t),\b_n(N,t)$, associated with $\pi_n(z)$, in the limit $n=N\ra\infty$, in full neighborhoods of critical points $t_0=-\frac 1{12}$ and $t_1=\frac 1{15}$ through certain Painlev\'{e}~I transcendents (note that we will not encounter $t_1$ in the particular settings considered in this paper). Those results were a signif\/icant development of \cite{ArnoDu, FIK}, where these asymptotics were studied only ``away from the poles'' of the Painlev\'{e}~I transcendents, that is, on certain subsets of full neighborhoods. The (only) remaining critical point $t_2=\frac 14$ for $f(z,t)=\hf z^2 + \qt t z^4$ was studied in \cite{BT4}, where the Painlev\'{e}~II transcendents played the same role as Painlev\'{e}~I transcendents in the case of~$t_0$,~$t_1$ (see \cite{BleherIts} for analysis of~$t_2$ in the case of real orthogonal polynomials on~$\R$). In order to understand the role of the critical points~$t_j$, $j=0,1,2$, in the global ``asymptotic phase portrait'', in~\cite{BT3} we studied numerically the regions of dif\/ferent asymptotic regimes (in the limit $n=N\ra\infty$) in the complex $t$-plane (in fact, on the Riemann surface of $t^\qt$), see, for example, Fig.~\ref{Generic}. We found there that there are six possible topologically dif\/ferent cases of global ``asymptotic phase portraits'', which are determined by the constant weights~$\varrho_j$, see~\eqref{orthog}, associated with each ray.

The objective of the current paper is to give a rigorous foundation for the numerical results, obtained in~\cite{BT3}. We will restrict our attention only to one of the above mentioned six cases, called Generic case, which will be def\/ined below. Our goal is to rigorously justify the asymptotic phase portrait for the Generic case (case {\bf (a1)} in~\cite{BT3}) using the {\it principle of continuation in the parameter space}, developed in~\cite{BertoBoutroux} and~\cite{TV1}. As it was stated in \cite{BertoBoutroux}, this technique will work for any $k$-th degree polynomial $f(z,\vec t)$ that depends on $k$ complex parameters $(t_1,\dots,t_k)=\vec t$ that are the coef\/f\/icients of $f$ (it is obvious that the free term in $f$ is irrelevant). Because of that consideration, a number of statements in this paper will be formulated and proved for general polynomial potentials $f(z;\vec t)$. We also obtain explicit expressions for the recurrence coef\/f\/icients and for the pseudonorms of $\pi_n(z)$ in terms of the Riemann theta functions. For the benef\/it of the reader, we repeat here some facts from \cite{BT3} that are required for the exposition.

As $z\ra \infty$, the weight function is exponentially decaying in four sectors $S_j$ of width $\pi/4$, centered around the rays $\O_j=\big\{z\colon \arg z=-\frac {\arg t}4+ \frac {\pi (1-j)} 2\big\}$, $j=0,1,2,3$. Since the argument of $t$ enters in this formula divided by four, we shall understand $\arg(t) \in [-\pi,7\pi)$, so that, ef\/fectively, $t$ lives on a four-sheeted Riemann surface. We consider the polynomials $\pi_n(z)$ to be integrated on the ``cross'' formed by the rays $\O_j$, where the rays $\O_{{0,1}}$ are oriented outwards (away from the origin) and the rays $\O_{{2,3}}$ -- inwards. The corresponding bilinear form is
\begin{gather}\label{orthog}
\langle p,q\rangle_{\vec\vr}= \sum_{j={0}}^{{3}}\varrho_j\int_{\O_j}p(z)q(z)e^{-N f(z,t)}dz,\qquad f(z,t):=\hf z^2 + \qt t z^4,
\end{gather}
where $\vec\vr=\{\varrho_{{0}},\varrho_1,\varrho_2,\varrho_3\}$ is a vector of f\/ixed complex numbers chosen to satisfy the so-called ``traf\/f\/ic condition''\footnote{The naming is suggested by the fact that the sums of the $\varrho_j$ associated with the incoming and with the outgoing rays are the same.}
 \begin{gather} \label{traffic}
\varrho_1+\varrho_{{0}}=\varrho_3+\varrho_{{2}}.
\end{gather}
Moreover, since multiplying all the $\varrho_j$'s by a common nonzero constant does not af\/fect the families of orthogonal polynomials, these parameters are only def\/ined modulo multiplication by $\C^*=\C\setminus\{0\}$. Interpreting the coef\/f\/icients $\vr_j$ as the traf\/f\/ics along the rays $\O_{j}$, $j=0,1,2,3$, we observe that the condition (\ref{traffic}) implies that there is no source/sink of traf\/f\/ic at the origin. Equivalently, this condition means that all the contours of integration can be viewed as starting and ending at $z=\infty$.

\begin{figure}[t!]\centering
\resizebox{0.4\textwidth}{!}{\input{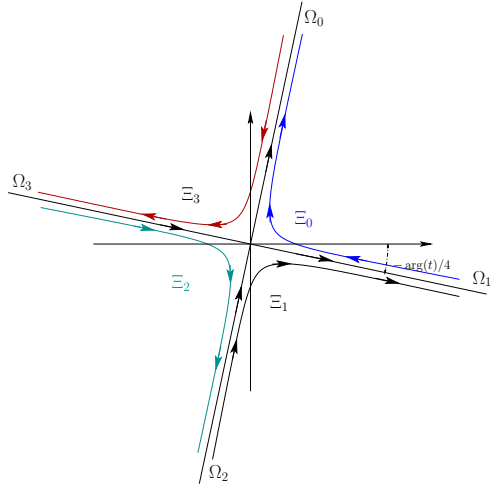}}

\caption{The contours of integration and the asymptotic directions. The contour $\Xi_{{0}}$ is homologically
equal to $-\Xi_1-\Xi_2-\Xi_3$ and, thus, can be excluded from the def\/inition of the pairing~(\ref{orthogA}).}\label{Figcontours}
\end{figure}
Indeed \cite{BertolaMo}, the bilinear form (\ref{orthog}) can be represented as
\begin{gather}\label{orthogA}
\langle p,q\rangle_{\vec\nu}= \sum_{k=1}^3 \nu_k \int_{\Xi_k} p(z)q(z)e^{-N f(z,t)}dz,\\
\nu_1 = \varrho_1 ,\qquad \nu_2 = \varrho_3-\varrho_{0} ,\qquad \nu_3 = - \varrho_{{0}},\nonumber
\end{gather}
where $\Xi_j$, $j=1,2,3$, are simple contours emanating from $\infty$ along $\O_{{j+1}} $ and returning to $\infty$ along $\O_{j}$ with the indices taken modulo~$4$, see Fig.~\ref{Figcontours}. The weights $\nu_j$ on the corresponding contours $\Xi_j$, $j=1,2,3$, will be also called traf\/f\/ics. Contours with zero traf\/f\/ic are not included in the $\O=\cup_{j=1}^3\Xi_j$ for the bilinear form~\eqref{orthogA}.

The generic conf\/iguration of the traf\/f\/ic that we consider in this paper is determined by
\begin{gather}\label{gener-eq}
\varrho_1+\varrho_{{0}}=\varrho_3+\varrho_{{2}}\neq 0 \qquad {\rm and} \qquad \varrho_j\neq 0 \qquad \text{for any} \quad j=0,1,2,3.
\end{gather}
(The ``asymptotic phase portrait'' of the generic case is shown on Fig.~\ref{Generic}.) However, our methods are equally applicable to any of the f\/ive remaining cases. Because of~\eqref{gener-eq}, without any loss of generality we will assume that
\begin{gather}\label{normal-traff}
\varrho_1+\varrho_{{0}}=\varrho_3+\varrho_{{2}}=1.
\end{gather}

\begin{figure}[t!] \centering
\includegraphics[width=0.8\textwidth]{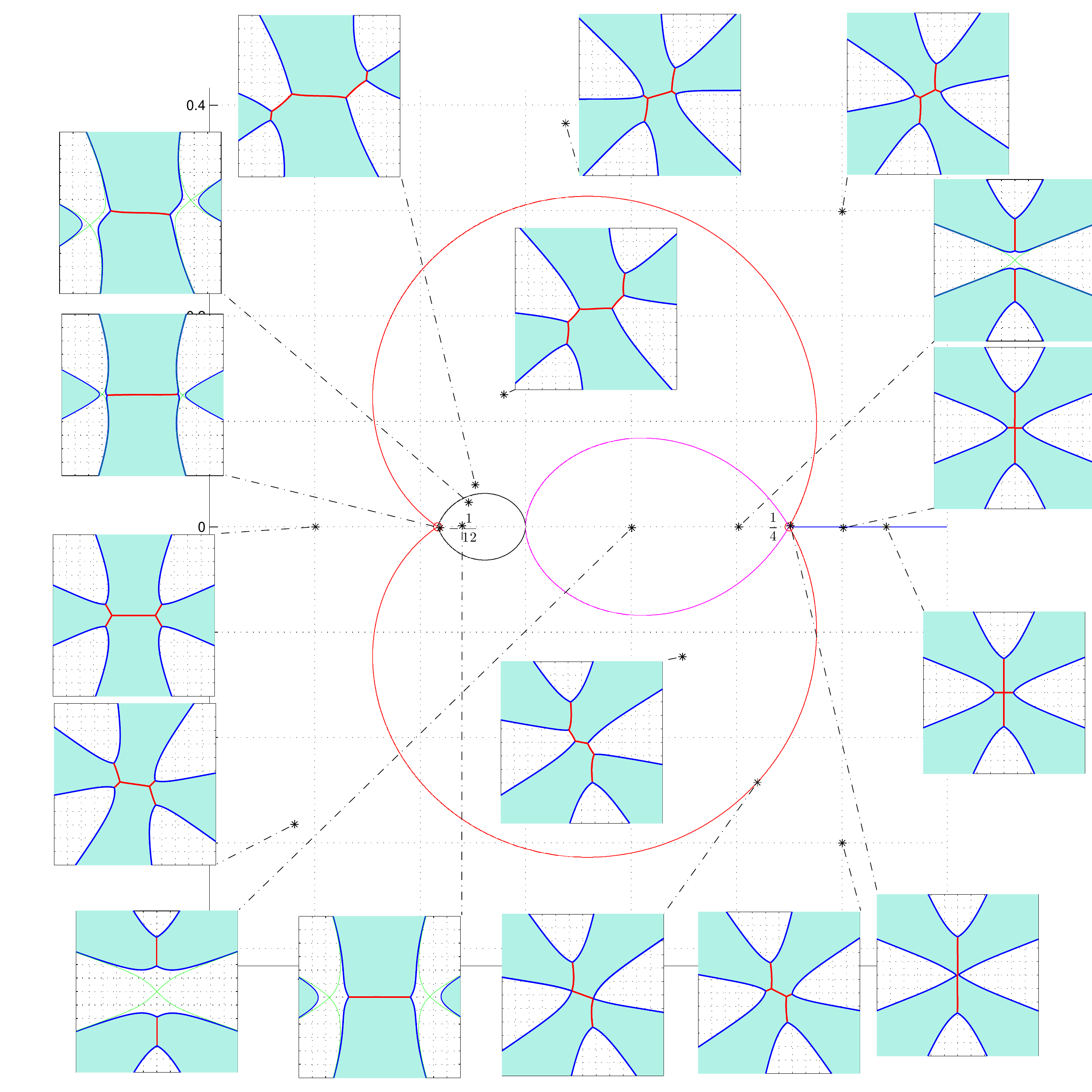}
\caption{Asymptotic regions on the complex $t$ plane in the case of generic traf\/f\/ic conf\/iguration. The smaller contour plots, attached to a particular value of $t$ (shown by the $*$ and connected with the plot by a~dashed line), show zero level curves $\Re h(z;t)=0$ in the complex $z$-plane. On these plots, the $\Re h(z;t)>0$ regions are shown by green (darker) colour whereas the $\Re h(z;t)<0$ regions are shown by white (lighter) colour; the solid red (dark) lines with darker colour on both sides show the main arcs, whereas the solid red (dark) lines with dif\/ferent colors on each side indicate bounded complementary arcs (actual complementary arcs, not shown here, can be obtained by deforming of\/f the corresponding red (dark) line into the lighter region while keeping their endpoints f\/ixed). Four unbounded complementary arcs, also not shown on the panels, are located in the lighter regions. Each of them connects one of the branchpoints with $z=\infty$. The asymptotic directions between the neighboring unbounded complementary arcs dif\/fer by the angle~$\pi/2$. Complementary arcs for some of the panels are shown on Fig.~\ref{comppic} below. The region inside the curve joining $\frac 1 4$ to~$0$ is of genus~$1$; inside the curve that joins $t=-\frac 1 {12}$ and $t=0$ it is of genus~0. Everywhere else the genus is~$2$, except for the degeneration to genus~$0$ occurring on the curve that joins $t=-\frac 1 {12}$ and $t=\frac 14$, and to genus $1$ on the ray $[\frac 1 4, \infty)$. Local behavior near $t=t_0=-\frac 1 {12}$ is expressed through Painlev\'e~I transcendents whereas local behavior near $t=t_2=\frac 14$ is expressed through Painlev\'e~II transcendents.}\label{Generic}
\end{figure}

To obtain the asymptotics of the recurrence coef\/f\/icients (and of the pseudonorm), associated with the polynomials $\pi_n$, we use the nonlinear steepest descent analysis of the associated Riemann--Hilbert problem (RHP), see Section~\ref{RHPsect}. Our goal is to show that this method is applicable at any $t\in \C^*\setminus\{t_0,t_2\}$. (In this paper, we restrict our attention to the f\/irst sheet of the Riemann surface of $t^\qt$ only; the remaining sheets can be analyzed similarly.) As it was mentioned earlier, our approach is based on the method of continuation in the (external) parameter space: once the validity of the nonlinear steepest descent method is established for some (regular) point $t=t_*$, say, for some $t_*\in(-\frac 1{12},0)$, it will remain in place as $t$
is continuously deformed in $ \C^*\setminus\{t_0,t_2\}$.

We now introduce the recurrence coef\/f\/icients and the pseudonorms. The orthogonality condition for the monic polynomials $\pi_n(z)$ can be written as
\begin{gather}\label{orthog1}
\langle \pi_n, z^k\rangle_{\vec \nu}=\h_n \delta_{nk},\qquad k=0,1,2,\dots, n, \qquad \vec\nu=(\nu_1,\nu_2,\nu_3),
\end{gather}
where $\delta_{ij}$ denotes the Kronecker delta. The coef\/f\/icient $\h_n$ can also be written as $\h_n=\langle \pi_n,\pi_n\rangle_{\vec \nu}$ and hence is the equivalent of the ``square norm'' of $\pi_n$ (in general, it is a complex number). We call $\h_n$ the pseudonorms. The existence of orthogonal polynomials $\pi_n(z)$ is not {\it a priori} clear. However, if three consecutive monic polynomials exists, then they are related by a three-term recurrence relation
\begin{gather}\label{3term}
\pi_{n+1}(z)=(z-\b_n)\pi_n(z)-\a_n\pi_{n-1}(z),
\end{gather}
where $\a_n$, $\b_n$ are called recurrence coef\/f\/icients. These recurrence coef\/f\/icients depend on~$t$,~$N$.

If the bilinear pairing is invariant under the map $z\mapsto -z$ then it follows immediately that the orthogonal polynomials are even or odd according to their degree and thus $\b_n=0$, $\forall\, n \in \N$. The remaining recurrence coef\/f\/icients $\a_n$ satisfy
\begin{gather*}
\a_n[1+t(\a_{n+1}+\a_n+\a_{n-1})]=\frac{n}{N},
\end{gather*}
which is known in literature as the string equation or the Freud equation \cite{Freud}. We are interested in the asymptotic limit of $\a_n$, $\b_n$ as $N\ra\infty$ and $\frac{n}{N}=1$, so we will use notations $\a_n=\a_n(t)$, $\b_n=\b_n(t)$.

In the case of $t\in(-\frac 1{12},0)$, the asymptotics of $\a_n$, $\b_n$ was obtained in \cite{ArnoDu} as
\begin{gather*}
\a_n(t)=\frac{\sqrt{1+12t}-1}{6t}+O\big(n^{-1}\big)
\end{gather*}
and $\b_n$ decaying exponentially as $n\ra\infty$. In fact, Theorem~1.1 from~\cite{ArnoDu} states that there exists some $n_0=n_0(t)$, such that $\a_n(t)$, $\b_n(t)$ exist for all $n\geq n_0$ and have the above mentioned asymptotics.

\begin{figure}[t!]\centering
\begin{tikzpicture}
\node(0,0) {\includegraphics[width=3cm]{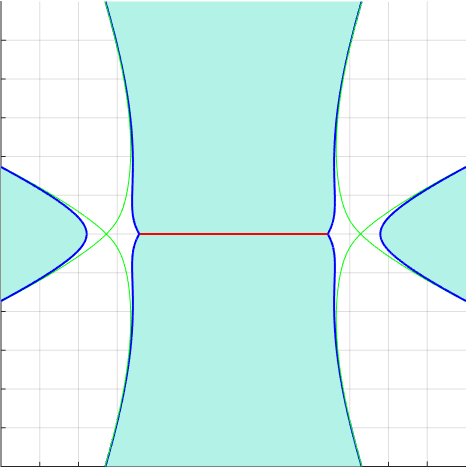}};
\draw (0.6,0) to [out=20,in=-110] (1.5,1.5);
\draw (0.6,0) to [out=-20,in=110] (1.5,-1.5);
\draw (-0.6,0) to [out=160,in=-70] (-1.5,1.5);
\draw (-0.6,0) to [out=-160,in=70] (-1.5,-1.5);
\end{tikzpicture}
\begin{tikzpicture}
\node(0,0) {\includegraphics[width=3cm]{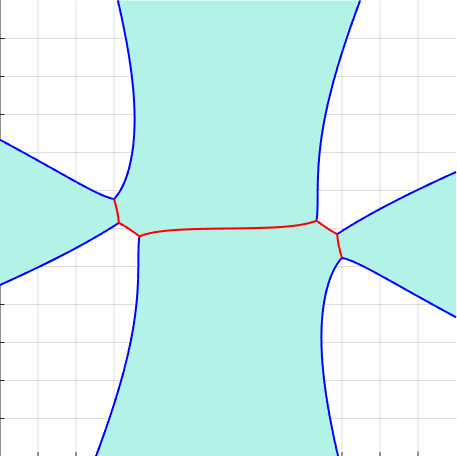}};
\draw (2.3312*0.25, 0.2081*0.25) to [out=50,in=110] (2.8744*0.25, -0.1475*0.25);
\draw (2.3312*0.25, 0.2081*0.25) to [out=60,in=-110] (1.5,1.5);
\draw (-2.9992*0.25, 0.7728*0.25) to [out=130,in=-70] (-1.5,1.5);
\draw (2.9992*0.25, -0.7728*0.25) to [out=-70,in=120] (1.5,-1.5);
\draw (-2.3312*0.25, - 0.2081*0.25) to [out=-100,in=-90] ( -2.8744*0.25, 0.1475*0.25);
\draw (-2.3312*0.25, - 0.2081*0.25) to [out=-100,in=70] (-1.5,-1.5 );
\end{tikzpicture}
\begin{tikzpicture}
\node(0,0) {\includegraphics[width=3cm]{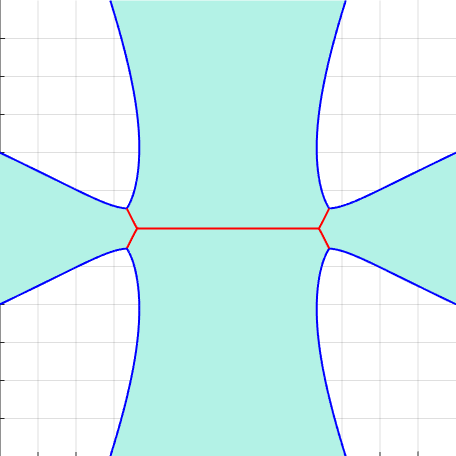}};
\draw (2.6635*0.25 , 0.5277*0.25) to [out=90,in=-110] (1.5,1.5);
\draw (2.6635*0.25, - 0.5277*0.25) to [out=-50,in=110] (1.5,-1.5);
\draw (-2.6635*0.25 , 0.5277*0.25) to [out=140,in=-70] (-1.5,1.5);
\draw (-2.6635*0.25, - 0.5277*0.25) to [out=-100,in=70] (-1.5,-1.5);
\end{tikzpicture}
\begin{tikzpicture}
\node(0,0) {\includegraphics[width=3cm]{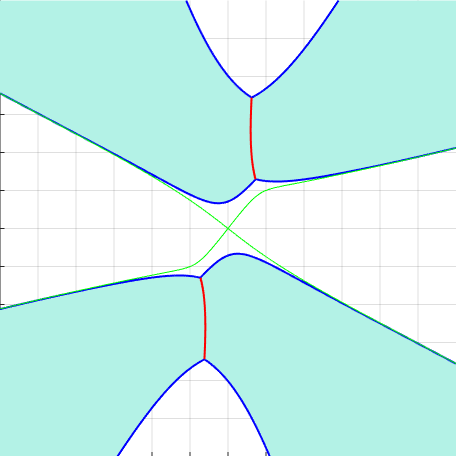}};
\draw (0.7291*0.25, 1.2958*0.25) to [out=-90,in=110] (-0.7291*0.25, -1.2958*0.25);
\draw (0.6210*0.25, 3.4442*0.25) to [out=60,in=-110] (0,1.5);
\draw (-0.6210*0.25, -3.4442*0.25) to [out=-80,in=90] (0,-1.5);
\draw (0.7291*0.25, 1.2958*0.25) to [out=-110,in=0] (-1.5,0);
\draw (-0.7291*0.25, -1.2958*0.25) to [out=60,in=180] (1.5,0);
\end{tikzpicture}
\caption{Four panels from Fig.~\ref{Generic} with complementary arcs shown as solid black curves and main arcs shown as solid red (dark) lines, surrounded by green (darker) regions. The panels correspond to the genus $0$ (left), $ 1$ (right), and $2$ (center) regions, the latter is represented by two panels as there are two topologically dif\/ferent conf\/igurations of main and complementary arcs for genus two. Bounded complementary arcs connect branchpoints with each other whereas the unbounded complementary arcs connect some branchpoints with $z=\infty$ along the four asymptotic directions in the corresponding lighter regions. The actual position of the contours is irrelevant (within the lighter regions).}\label{comppic}
\end{figure}
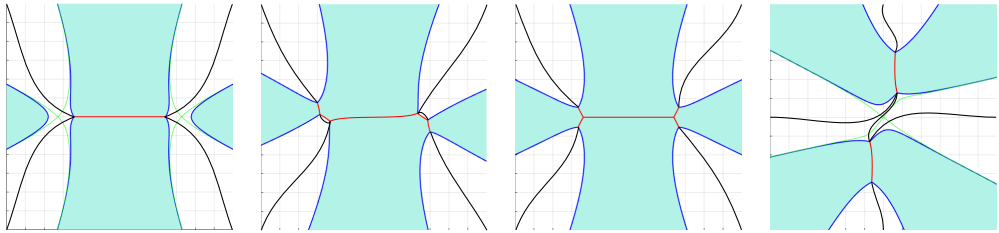

We now outline the content of the paper. The Riemann--Hilbert problem (RHP) representation of the orthogonal polynomials, introduced in \cite{FIK}, is brief\/ly described in Section~\ref{RHPsect}. In the rest of Section \ref{solRHPsect}, we use the nonlinear steepest descent method of Deift--Zhou and the $g$-function mechanism to f\/ind the large $N$ asymptotics of the corresponding RHP. Most of the statements/results of this section are valid for general polynomial potentials $f(z)=f(z;\vec t)$. In Section~\ref{sec-reqg} we def\/ine the $g$ function and the associated $h$ function (also known as the modif\/ied external f\/ield) as solutions to certain scalar RHPs, whose jump contours are yet to be def\/ined. In addition, $\Re g$ must satisfy a set of inequalities, known as sign conditions (sign requirements). Def\/ining the jump contours for $g$ starts with def\/ining their end points, which are also the branchpoints of the hyperelliptic Riemann surface $\Rscr=\Rscr(\vec t)$, associated with $g$. Various forms of modulation equations that def\/ine these branchpoints $\lambda_j$ are discussed in Section~\ref{subsect-modeq}. As the genus of $\Rscr(\vec t)$ can change when we deform $\vec t$, we can talk about regions of dif\/ferent genera in the space of parameters $\vec t$. Schematic description of the steepest descent method and explicit formulae for the recurrence coef\/f\/icients and the pseudonorms in the genus $L$ region, $L\in\N$, can be found in Section~\ref{sect-stdesc}. Some basic facts about hyperelliptic Riemann surfaces and Riemann theta functions are also provided there.

In Section~\ref{modeqsubsec} we calculate the $g$ function in the regions of dif\/ferent genera. (Our approach of calculating $g$ follows \cite{KMM, TVZ1}, since the branchcuts of $g$ are in $\C$ and the usual variational approach used for the equilibrium problem on the real line is not applicable.) This calculation is still formal, as we postpone the proof of the sign requirements until Section~\ref{sec-high-gen}. In Section~\ref{ghcalc}, following \cite{BT3}, we solve the modulation equations for the branchpoints and calculate the $g$ and $h$ functions in the genus zero region for our quartic potential $f(z,t)$. In Section~\ref{sec-genusL}, we construct the $g$ and $h$ functions for higher genera regions in the determinantal form and derive new forms of the modulation equations. In particular, we proved Theorem~\ref{equiv-mod}, stating that the modulation equations for a branchpoint $\l$ is equivalent to $\frac{\part g}{\part\l}\equiv 0$ for all $z\in \C$ and is equivalent to the fact that real jump constants (that we will refer to as phases) in the RHP for~$g$, see~\eqref{jumpgcomp} and~\eqref{jumpgmain}, are all independent of $\l$. This statement is valid for a larger class of RHPs, see, for example,~\cite{ET}, as applied to the semiclassical limit of the focusing Nonlinear Schr\"{o}dinger equation (NLS).

In Section~\ref{sec-high-gen} we prove the existence of $g$ function for any $t\in \C^*\setminus\{t_0,t_2\}$ for our quartic potential $f(z;t)$ in the Generic case, that is, we prove that there is a formal $g$ function satisfying the sign conditions. We also derive equations for the breaking curves on complex $t$ plane, separating regions of dif\/ferent genera, and def\/ine critical (breaking) points. Following~\cite{BT3}, in Section~\ref{sect-exist} we remind the basic steps in proving the existence of $g$-function in the genus zero region, which contains, in particular, the interval $(-\frac 1{12}, 0)$. That implies that the asymptotics of the recurrence coef\/f\/icients is proven on the genus zero region. In Section~\ref{sect-symm} we derived the breaking curve equations for a~general polynomial potential $f(z;\vec t)$ and study properties of breaking curves.

One of the central points of this paper, discussed in Section~\ref{sect-cont-princ}, is a~proof that once the asymptotics of the recurrence coef\/f\/icients is established in some region of the $t$-plane (or even just at one regular point $t$), the {\it continuation principle for Boutroux deformations} would establish the existence of the asymptotics of recurrence coef\/f\/icients for any $t\in\C^*\setminus \{t_0,t_2\}$. That is to say that if there exists the nonlinear steepest descent solution to the RHP with some regular value of~$t$, then the nonlinear steepest descent solution to the RHP exists for any regular value of $t\in\C^*$; moreover, it can be obtained by continuing the solution along any smooth curve (in the parameter space), connecting the original and the f\/inal values of $t$ that avoids $t=0$ as well as critical values of $t=t_0$ and $t=t_2$. Equivalently, one can say that if the $g$-function satisfying all the sign conditions (see Section~\ref{sec-reqg}) is established in some region of the $t$-plane, it can be continued to any regular value of $t\in \C^*$ while preserving all the above-mentioned sign conditions. In the process of such a continuation in the $t$-plane, the genus of the underlying Riemann surface~$\Rscr(t)$ will undergo the appropriate changes. The continuation principle for Boutroux deformations was (independently) introduced in~\cite{BertoBoutroux} in the context of general polynomial potential $f(z,\vec t)$ and in \cite{TV1} in the context of the semiclassical NLS, where $f=f(z;x,t)$ is a Schwarz-symmetrical function in $z$ that depends on real parameters~$x$,~$t$ (space and time). Following \cite{BertoBoutroux, TV1}, we outline the proof of the continuation principle for general potentials $f(z,\vec t)$. The essentially novel part in this proof, presented separately in Section~\ref{sect-bound}, is the proof that the branchpoints $\lambda_j$ of $\Rscr(\vec t)$ are bounded under Boutroux deformations on any compact subset of $\vec t$ that preserves the degree of the polynomial~$f(z;\vec t)$. This particular issue was not entirely addressed in~\cite{BertoBoutroux}. The problem discussed in Section~\ref{sect-cont-princ} is closely related to the problem of existence of a curve with S-property for the polynomial external f\/ield that, using the approach of~\cite{Rah}, was recently solved in~\cite{ArnoSilva}.

In Section \ref{sect-br-cr-symm} we present explicit formulae for all the breaking curves for our quartic poten\-tial~$f(z,t)$ and prove that they separate regions of particular genera, as shown on Fig.~\ref{Generic}. In short, we prove the topology of regions on the ``asymptotic phase portrait'' from Fig.~\ref{Generic}. Explicit formulae for the leading order behavior of the recurrence coef\/f\/icients $\a_n(t)$ and the pseudonorms~${\bf h}_n(t)$ in the
regions of genera one and two are given in Section~\ref{sect-hi-gen-sol}.

\section{Steepest descent analysis of the RHP (\ref{RHPY})} \label{solRHPsect}

In this section we f\/irst (Section~\ref{RHPsect}) express the recurrence coef\/f\/icients $\a_n(t)$ and the pseudo\-norms $\h_n(t)$ through the solution of a matrix RHP $Y$, and then (Sections~\ref{sec-reqg} and~\ref{subsect-modeq}) describe the nonlinear steepest descent approach to f\/ind the asymptotics of $Y$ as $n=N\ra \infty$. The leading order asymptotics for recurrence coef\/f\/icients and pseudonorms can be found in Section~\ref{sect-stdesc}. The results of this section can be extended for general potentials $f(z;\vec t)$ in an obvious way. Our exposition in Section~\ref{RHPsect} and in some other parts of this section follow~\cite{BT3}.

\subsection{The RHP for recurrence coef\/f\/icients}\label{RHPsect}

It is well known \cite{Deift} that the existence of the above-mentioned orthogonal polynomials $\pi_n(z)$ is equivalent to the existence of the solution to the following RHP~(\ref{RHPY}). More precisely, relation between the RHP (\ref{RHPY}) and the orthogonal polynomials $\pi_n(z)$ is given by the following proposition~\cite{ArnoDu}, which has the standard proof (see~\cite{Deift}).

\begin{Proposition}\label{proparno1}
Define $\nu \colon \Omega \to \C$ $($where $\O=\bigcup_{j=1}^3 \Xi_j)$ by $\nu(z)=\nu_j$ when $ z\in \Xi_j$. Then the solution of the following RHP problem
\begin{gather}
 Y(z) \quad \text{is analytic in} \ \ \C\setminus \O ,\nonumber\\
Y_+(z)=Y_-(z) \left(
\begin{matrix}
1 & \nu(z){e}^{-N f(z,t)}\\
0& 1
\end{matrix}\right), \qquad z\in\O,\label{RHPY} \\
Y(z)=\big(\1+O\big(z^{-1}\big)\big) \left(
\begin{matrix}
z^n & 0\\
0& z^{-n}
\end{matrix}
\right), \qquad z\ra\infty, \nonumber
\end{gather}
exists $($and it is unique$)$ if and only if there exist a monic polynomial $p(z)$ of degree $n$ and a~polynomial~$q(z)$ of degree $\leq n-1$ such that
\begin{gather*}
\langle p(z), z^k\rangle_{\vec\nu}=0,\qquad \text{for all} \quad k=0,1,2,\dots, n-1, \\
\langle q(z), z^k\rangle_{\vec\nu}=0,\qquad \text{for all} \quad k=0,1,2,\dots, n-2,\qquad \text{and} \qquad \langle q(z),
z^{n-1}\rangle_{\vec\nu}=-2\pi i,
\end{gather*}
where $\vec\nu=(\nu_1,\nu_2,\nu_3)$. In that case the solution to the RHP \eqref{RHPY} is given by
\begin{gather*}
Y(z)=\left(
\begin{matrix}
p(z) & C_\O\big[p(z)\nu(z){ e}^{-N f(z,t)}\big] \\
q(z) & C_\O\big[q(z)\nu(z){ e}^{-N f(z,t)}\big]
\end{matrix}
\right), \qquad z\in \C\setminus \O,
\end{gather*}
where
\begin{gather*}
C_\O[\phi]=\frac{1}{2\pi i}\int_\O\frac{\phi(\z) d\z}{\z-z}
\end{gather*}
is the Cauchy transform of $\phi(z)$.
\end{Proposition}

\begin{Remark}\label{workingcont}
We remind that the contour $\O$ in \eqref{RHPY} should coincide with the contour of integration in the bilinear form~\eqref{orthogA}, so if any traf\/f\/ic $\nu_j=0$ the corresponding contour $\Xi_j$ should not be included in~$\O$.
\end{Remark}

If $\pi_{n+1}(z)$, $\pi_{n}(z)$, $\pi_{n-1}(z)$ are monic orthogonal polynomials then they satisfy the three term recurrence relation~\eqref{3term} for certain recurrence coef\/f\/icients $\a_n$, $\b_n$ \cite{Chihara, Szego}. The following well known statements (see, for example, \cite{Deift,ArnoDu, FIK}) show the connection between the RHP~(\ref{RHPY}), the orthogonal polynomials~$\pi_n(z)$ and their recurrence coef\/f\/icients.

\begin{Proposition}\label{proparno3}
Let $Y^{(n)}(z)$ denote the solution of the RHP \eqref{RHPY}. If we write
\begin{gather}\label{assYz}
Y^{(n)}(z)=\left(\1+\frac{Y^{(n)}_1}{z} + \frac{Y^{(n)}_2}{z^2} +O(z^{-3})\right)
\left(\begin{matrix}
z^n & 0\\
0& z^{-n}
\end{matrix}
\right), \qquad z\ra\infty,
\end{gather}
then the pseudo-norms $\h_n$, appearing in~\eqref{orthog1}, and the recurrence coefficients~$\a_n$, appearing in~\eqref{3term}, are given by
\begin{gather}\label{expressalpbet}
\h_n = - {2i\pi} \big(Y_1^{(n)}\big)_{12} ,\qquad
\a_n=\big(Y^{(n)}_1\big)_{12}\cdot \big(Y^{(n)}_1\big)_{21},\qquad \b_n=\frac{\big(Y^{(n)}_2\big)_{12}}{\big(Y^{(n)}_1\big)_{12}}-\big(Y^{(n)}_1\big)_{22}.
\end{gather}
\end{Proposition}

\begin{Proposition}\label{proparno4}
Suppose the RHP~\eqref{RHPY} has solution $Y^{(n)}(z)$. Let~\eqref{assYz} be its expansion at~$\infty$ and let $\a_n$, $\b_n$ be given by~\eqref{expressalpbet}. If $\a_n\neq 0$, then the monic orthogonal polynomials $\pi_{n+1}(z)$, $\pi_{n}(z)$, $\pi_{n-1}(z)$ exist and satisfy the three term recurrence relation~\eqref{3term}.
\end{Proposition}

The Deift--Zhou steepest descent analysis is a powerful tool in f\/inding the asymptotic behavior of a~matrix RHP with respect to a given asymptotic parameter. In the case of the RHP \eqref{RHPY}, the asymptotic parameter is the large parameter $N=n$. The remaining parameters of the RHP will be called external parameters. In the case of the RHP~\eqref{RHPY}, the external parameters are~$t$,~$\vec \nu$, however, we will consider~$\vec \nu$ to be f\/ixed in the analysis below. The steepest descent analysis, as applied to the RHP~\eqref{RHPY}, will be outlined here. As customary, the RHP undergoes a sequence of modif\/ications into equivalent RHPs until it can be ef\/fectively solved (so-called model RHP) while keeping the error terms under control.

$\bullet$ \ One starts with the RHP (\ref{RHPY}) for $Y$ and seeks an auxiliary scalar function $g=g(z;t)$, called the $g$-function, which is analytic in $\C$ except for a collection $\Sigma$ of appropriate contours to be described in Section~\ref{sec-reqg} and which behaves like $\ln z + \mathcal O(z^{-1})$ near $z=\infty$. The contour $\O$ of the RHP (\ref{RHPY}) can be deformed because the RHP (\ref{RHPY}) has an analytic in $z$ jump matrix. If two or more parts of $\O$ are deformed into a single oriented arc, the traf\/f\/ic on this arc is a signed (according to the orientation) sum of the traf\/f\/ics on the parts of $\O$ that formed the arc. The parts, if any, of the deformed $\O$ carrying zero traf\/f\/ic (no integration upon) should be removed from the contour~$\O$.

$\bullet$ Then we introduce a new matrix
\begin{gather}
T(z):={\rm e}^{-N\ell\frac {\s_3} 2}Y(z){\rm e}^{-N (g(z,t)-\frac \ell2){\s_3} },\label{defT}
\end{gather}
where the constant $\ell\in\C$ is to be def\/ined. Direct calculations show that $T(z)$ solves the RHP
\begin{gather}
 T(z) \ \text{is analytic in} \ \C\setminus \O, \nonumber\\
T_+(z)=T_-(z) \left(
\begin{matrix}
{\rm e}^{-\frac N2 \le(h_+-h_-\ri)} & \nu(z){ e}^{\frac N2 \le(h_+ + h_-\ri)}\\
0& {\rm e}^{\frac N2 \le(h_+-h_-\ri)}
\end{matrix}
\right), z\in\O,\label{RHPT}\\
T(z)=\big(\1+O\big(z^{-1}\big)\big), \qquad z\ra\infty,\nonumber
\end{gather}
where $h(z,t):= 2g(z,t)-f(z,t) - \ell$.

$\bullet$ At this point the Deift--Zhou method can proceed provided that the function $g$, the constant $\ell$ and the collection of contours $\Omega$ into which we have deformed the original contour of integration fulf\/ill a~rather long collection of equalities and~-- most importantly~-- {\it inequalities} (sign conditions) that we brief\/ly describe in Section~\ref{sec-reqg} (see, for example, \cite{BertolaMo,TVZ1} for more details). If all these requirements are fulf\/illed, the leading order asymptotics for the problem (the solution to the model problem) can be obtained in terms of Riemann theta functions on a suitable hyperelliptic Riemann surface of a positive genus (with the case of zero genus not requiring the Riemann theta function).

We construct the $g$-function (or, more precisely, the corresponding $h$-function, see~\eqref{RHPT}), not by using the variational principle, as often done for orthogonal polynomials on~$\R$, but by satisfying the collection of the above-mentioned equalities. These equalities can be considered as jump conditions in the corresponding scalar RHP for~$g$ (or for~$h$). If the obtained $g$-function does not satisfy the inequality requirements, we need to change the genus of the scalar RHP for~$g$. The continuation principle, see Section~\ref{sect-cont-princ}, enables us to trace the correct genus of~$g$ through the complex $t$-plane.

\subsection[General requirements on the $g$-function]{General requirements on the $\boldsymbol{g}$-function}\label{sec-reqg}

The f\/inal conf\/iguration of the contour $\Omega$ can be partitioned into two subsets of oriented arcs that we shall denote by $\mathfrak M$ and term {\it main arcs} (bands), and by $\mathfrak C$ and term {\it complementary arcs} (gaps), $\Omega=\mathfrak M\cup \mathfrak C$, where all the main arcs are bounded and have a nonzero traf\/f\/ic. {The f\/inal conf\/iguration of $\Omega$ must contain all the contours where $g$ is not analytic}. This partitioning def\/ines the hyperelliptic Riemann surface $\mathfrak R( {t})$, whose branchcuts are the main arcs $\g_{m,j}$, where $\mathfrak M=\cup_{j=0}^L\g_{m,j}$. It is subordinated to the list of requirements for the function $g$ (or for the corresponding function~$h$), associated with the Riemann surface $\mathfrak R( {t})$ of genus $L$, given in the following two subsections. For example, on Fig.~\ref{Generic}, main arcs on each panel are the (red) arcs separating two blue (darker) regions. The genus of the problem is the number of the main arcs reduced by one. A~f\/inite complementary arc (these arcs are not usually shown on the panels) is a bounded curve within a white (lighter) region (or on its boundary), connecting the corresponding endpoints of the main arcs, see, for example, the panel in the low left corner. On the two panels in the middle of Fig.~\ref{Generic}, f\/inite complementary arcs are the (red) arcs, separating the blue (darker) and the white (lighter) regions. On some panels of positive genus (e.g., panels~2,~3 from the bottom on the left edge) there are no f\/inite complementary arcs.

\subsubsection[Equality requirements and formula for $g$]{Equality requirements and formula for $\boldsymbol{g}$}\label{eqreqs}
\begin{enumerate}\itemsep=0pt
\item[1)] $g(z)$ (to shorten notations, we will often drop the $ t$ variable) is analytic in $\C\setminus \Sigma$, where $\Sigma=\le(\mathfrak M\cup\mathfrak C\ri)$, and has the asymptotic behaviour
\begin{gather}\label{assg}
g(z) = \ln z + \mathcal O\big(z^{-1}\big) ,\qquad z\ra\infty;
\end{gather}
\item[2)] $g(z)$ is analytic along all the {\it unbounded} complementary arcs except for exactly one {\it unbounded} complementary arc which we will denote by $\gamma_{0}$ (we assume $\g_0$ is oriented away from inf\/inity), where
\begin{gather}\label{jumpg-gam_0}
g_+(z)-g_-(z) = 2i\pi ,\qquad z\in \gamma_{0}\subset\mathfrak C	
\end{gather}
(note that the function ${\rm e}^{N g(z)}$ from~(\ref{defT}) is analytic across the {unbounded} complementary arc since $N=n\in \N$ by def\/inition). In fact, $\g_0$ could be any simple contour connecting $\infty$ with a point on $\mathfrak M$ that has no other intersections with $\O$;
\item[3)] on each {\it bounded} complementary arc $\gamma_{c,j}$, $j=1,2,\dots,L$, the function $g(z)$ has a constant jump
\begin{gather}\label{jumpgcomp}
 g_+(z) - g_-(z) = 2\pi i \eta_{j} ,\qquad \eta_{j}\in \R ,\qquad z\in \gamma_{c,j}\subset\mathfrak C
\end{gather}
(note that there are no bounded complementary arcs in the genus zero case $L=0$);
\item[4)] across each main arc $\gamma_{m,j}$, $j=0,1,2,\dots,L$, the function $g(z)$ has a jump
\begin{gather}\label{jumpgmain}
 g_+(z) + g_-(z)= f(z) + \ell + 2\pi i \varpi_j,\qquad \varpi_{j}\in \R ,\qquad z\in \gamma_{m,j} \subset \mathfrak M,
\end{gather}
where $\varpi_0=0$. We stress that the constant $\ell$ is the same for all the main arcs;
\item[5)] the limiting values $g_\pm(z)$ along all the contours belong to $L^2_{\rm loc}$.
\end{enumerate}

\begin{Remark}\label{rem-whiskers}
The assumption that all the main and complementary arcs are simple smooth curves that can intersect each other only at the endpoints would have simplif\/ied the following discussion. In reality, though, main (complementary) arcs could intersect each other in a f\/inite number of interior points forming, for example, crosses, ``Y''-shaped contours (whiskers), etc., see, for example, various panels on Fig.~\ref{Generic}. If an endpoint of a main arc $\g_{m,j-1}$ is in the interior of the neighboring main arc $\g_{m,j}$ for some $j=1,\dots,L$, then there is no complementary arc $\g_{c,j}$ that connects these two main arc. Thus, the number of the complementary arcs in~\eqref{jumpgcomp} can be less than~$L$. However, in this case, the constant $\varpi_j$ from~\eqref{jumpgmain} on $\g_{m,j}$ will take dif\/ferent values on the dif\/ferent parts of this arc (separated by the intersection point). This jump condition will be equivalent to the situation where the constant $\varpi_j$ is the same on all of the $\g_{m,j}$, but there is another jump~\eqref{jumpgcomp} coming from the arc $\g_{c,j}$, which is the part of $\g_{m,j}$ that connects one of the endpoints with the point of intersection. Thus, we can assume that there are always~$L$ arcs $\g_{c,j}$ in~\eqref{jumpgcomp}, with the understanding that only proper complementary arcs have to satisfy the sign requirements from Section~\ref{signreqs}.
\end{Remark}

\begin{Remark}\label{rem-RHPg}
Assuming that the contours constituting $\mathfrak M$ and $\mathfrak C$, as well as all the constants $\ell$, $\eta_j$, $\varpi_j$, $j=1,2,\dots,L$, are known, the conditions 1-5 form an RHP for the function~$g(z)$. If a branch of~$\ln z$ in~\eqref{assg} is f\/ixed, that RHP would have a~unique solution (constructed below), provided that the constants can be chosen so that the f\/irst requirement~\eqref{assg} is satisf\/ied.
\end{Remark}

We start the construction of $g(z)$ by removing the jump on the unbounded contour $\g_0$. Let $\lambda_{2j}$, $\lambda_{2j+1}$ be the beginning and the end points of the (oriented) main arc~$\g_{m,j}$, $j=0,\dots,L$. We def\/ine the radical
\begin{gather}\label{R}
R(z)=\sqrt{\prod_{j=0}^{2L+1}(z-\lambda_j)},
\end{gather}
where the branchcuts of $R(z)$ are $\g_{m,j}$, $j=0,1,\dots,L$, and the branch of $R$ is def\/ined by $\lim_{z\ra\infty}\frac{R(z)}{z^{L+1}}=1$. The hyperelliptic Riemann surface $\mathfrak{ R}(t)$ is the Riemann surface of~$R(z)$. Let $\widetilde \o=\frac{\widetilde P(\z)d\z}{R(\z)}$ be a normalized meromorphic dif\/ferential of the third kind with two simple poles at~$\infty_\pm$ with residues~$\pm 1$ respectively (i.e., $\widetilde P(z)$ is monic of degree~$L$). It is well known~\cite{FK} that the above conditions def\/ine $\tilde P(z)$ uniquely. The def\/inition of ``normalized'' means that
\begin{gather}\label{uP}
\int_{\g_{m,j}}\frac{\widetilde P(\z)d\z}{R_+(\z)}=0 \qquad \text{for all} \quad j=1,\dots,L.
\end{gather}

Let
\begin{gather}\label{u-L}
u(z)=\int_{\lambda_0}^z\frac{\widetilde P(\z)}{R(\z)}d\z,
\end{gather}
where $\lambda_0$ is the beginning of the main arc $\g_{m,0}$. It will be also convenient to choose~$\lambda_0$ as the (f\/inite) endpoint of the contour~$\g_0$. Then~$u(z)$ is analytic in $\C\setminus(\mathfrak M\cup\g_0)$, has analytic limiting values on $\mathfrak M\cup\g_0$ (except, possibly, the branchpoints) and satisf\/ies
\begin{gather*}
u(\lambda_0)=0,\qquad u_+(z)-u_-(z) = 2i\pi,\quad z\in \gamma_{0},\qquad u_+(z)+u_-(z) = 0,\quad z\in \gamma_{m,0},
\end{gather*}
and $u(z)-\ln z$ is analytic at $z=\infty$. We will be looking for $g(z)$ in the form $g(z)=u(z)+v(z)$. Then $v(z)$ must be analytic in $\bar\C$ except for the main arcs~$\mathfrak{M}$ and bounded complementary arcs $\cup_{j=1}^L\g_{c,j}$, where, according to~\eqref{jumpgmain},~\eqref{jumpg-gam_0}, it satisf\/ies jump conditions
\begin{gather*}
 v_+(z) + v_-(z)= f(z) -[u_+(z)+u_-(z)] +\ell + 2\pi i \varpi_j,\qquad
 v_+(z) - v_-(z) = 2\pi i \eta_{j}
\end{gather*}
on the corresponding contours. By Sokhotski--Plemelj formula,
\begin{gather*}
v(z)={{R(z)}\over{2\pi i}}\Bigg[ \int_{\mathfrak M} {{f(\z)-[u_+(\z)+u_-(\z)]+\ell}\over{(\z-z)R(\z)_+}}d\z \\
 \hphantom{v(z)=}{}
 + \sum_{j=1}^L \le(\int_{\g_{m,j}} {{2\pi i\varpi_j d\z}\over{(\z-z)R(\z)_+}}
 + \int_{\g_{c,j}} {{2\pi i\eta_j d\z}\over{(\z-z)R(\z)}}\ri) \Bigg].
\end{gather*}
The solution $g(z)$ to the RHP from Remark \ref{rem-RHPg} is then given by
 \begin{gather}
 g(z) =u(z)+{{R(z)}\over{2\pi i}}\Bigg[ \int_{\mathfrak M}{{f(\z)-[u_+(\z)+u_-(\z)]+\ell}\over{(\z-z)R(\z)_+}}d\z
 \nonumber\\
\hphantom{ g(z) =}{} +\sum_{j=1}^L\le( \int_{\g_{m,j}}{{2\pi i\varpi_j d\z}\over{(\z-z)R(\z)_+}}
+\int_{\g_{c,j}}{{2\pi i\eta_j d\z}\over{(\z-z)R(\z)}}\ri) \Bigg]\nonumber\\
\hphantom{ g(z)}{}=u(z)+{{R(z)}\over{4\pi i} } \Bigg[ \oint_{\hat {\mathfrak M} } {{f(\z)-2u(\z)+\ell}\over{(\z-z)R(\z)}}d\z\nonumber\\
\hphantom{ g(z) =}{}
+\sum_{j=1}^L\le( \oint_{\hg_{m,j}}{{2\pi i\varpi_j d\z}\over{(\z-z)R(\z)}}
+\oint_{\hg_{c,j}}{{2\pi i\eta_j d\z}\over{(\z-z)R(\z)}}\ri) \Bigg],\label{gform2}
\end{gather}
where $\hg_{m,j}$, $\hg_{c,j}$ is a negatively oriented loop around $\g_{m,j}$, $\g_{c,j}$ respectively, with the latter traversing both sheets of the Riemann surface $\mathfrak R(t)$, $j=1,\dots,L$, and $\hat{\mathfrak M}$ denotes a~negatively oriented loop around~$\mathfrak M$ that (because of the $u(z)$ term) passes through $\lambda_0$. Here~$z$ is chosen outside all the loops.

\begin{Remark}\label{rem-analjump}
Note that the jumps of $g$ are analytic (in $z$) functions; thus, the interior of any main or complementary arc can be smoothly deformed without af\/fecting the value of~$g(z)$, as long as the deformed arc does not pass through $z$. Moreover, the limiting values~$g_\pm(z)$ are analytic functions on the contour $\O$ with the exception of the branchpoints of the Riemann surface~$\mathfrak R(t)$ (endpoints of the main and complementary arcs).
\end{Remark}

\begin{Remark}\label{rem-RHPh}
Similarly to $g$, the function
\begin{gather}\label{h-def}
h=2g-f-\ell
\end{gather}
can be considered as the (unique) solution of the scalar RHP with the jumps
\begin{gather}
h_+(z)-h_-(z)= 4\pi i \eta_j,\quad z\in\g_{c,j},\qquad h_+(z) + h_-(z) =4\pi i\varpi_j,\quad z\in\g_{m,j}, \nonumber\\
h_+(z)-h_-(z) = 4i\pi,\quad z\in \gamma_0,	\label{jumph}
\end{gather}
across all the main and all the bounded complementary arcs and with the asymptotic behavior
\begin{gather}\label{asshgen}
h(z) = -f(z)-\ell +2\ln z + \mathcal O\big(z^{-1}\big) ,\qquad z\ra\infty.
\end{gather}
\end{Remark}

It follows immediately from (\ref{jumph}) that $\Re h(z)$ is continuous across the complementary arcs~$\mathfrak C$. Also, \eqref{gform2} implies that
\begin{gather}
h(z)={{R(z)}\over{2\pi i}}\Bigg[ \oint_{\hat{\mathfrak M}}{{f(\z)+\ell-2u(\z)}\over{(\z-z)R(\z)}}d\z\nonumber\\
\hphantom{h(z)=}{}
+\sum_{j=1}^L \le(\oint_{\hg_{m,j}}{{2\pi i\varpi_j d\z}\over{(\z-z)R(\z)}}
+\oint_{\hg_{c,j}}{{2\pi i\eta_j d\z}\over{(\z-z)R(\z)}} \ri)\Bigg],\label{hform}
\end{gather}
where $z$ is inside the loop $\hat{\mathfrak M}$ but outside all other loops. Notice that $\ell$ can be removed from~\eqref{hform} since $\oint_{\hat{\mathfrak M}}{{d\z}\over{(\z-z)R(\z)}}=0$ when $z$ is inside~$\hat{\mathfrak M}$.

\subsubsection[Inequality (sign) requirements (or sign distribution requirements) for $h$ and the modulation equation]{Inequality (sign) requirements (or sign distribution requirements)\\ for $\boldsymbol{h}$ and the modulation equation}\label{signreqs}

The following sign distribution requirements must be satisf\/ied:
\begin{enumerate}\itemsep=0pt
\item[1)] in the interior of every bounded complementary arc $\g_{c,j}\subset \mathfrak C$, $j=1,\dots,L$, as well as in the interior of
every unbounded complementary arc, we have $\Re h(z)< 0$;
\item[2)] if $z_0$ is an interior point of a main arc $\gamma_{m,j}\subset \mathfrak M$, $j=0,1,\dots,L$ then on both sides $\g_{m,j}$ in close proximity of $z_0$ we have $\Re h(z)> 0$.
\end{enumerate}

\begin{Remark}\label{rem-signReh}
The second requirement together with \eqref{jumph} imply that $\Re h(z)\equiv 0$ on $ \mathfrak M$. Thus, the jump conditions \eqref{jumph} imply that {\it $\Re h(z)$ is continuous everywhere in~$\C$}.
\end{Remark}

\begin{Remark}\label{rem-except}
We shall call the above mentioned case {\it regular}, with the same connotation as in~\cite{DKMVZ}. Violation of the above mentioned strict inequalities is allowed in exceptional (non-regular) cases at no more than a~f\/inite number of interior points of the corresponding main and/or complementary arcs $\O= \mathfrak C\cup \mathfrak M$. Let $z_*\in\O$ denote one of such points. Then $\Re h(z_*)=0$. Indeed, this is true when $z_*\in \mathfrak M$, see Remark~\ref{rem-signReh}; otherwise, when $z_*\in \mathfrak C$, this is implied by the continuity of $\Re h(z)$. Moreover, the above requirements imply that there are at least four zero level curves of $\Re h(z)=0$, emanating from $z=z_*$ (pinching).
\end{Remark}

\subsection{Modulation equations }\label{subsect-modeq}

The sign distribution requirements of Section \ref{signreqs} imply that $\Re h(z)$ should attain both positive and negative values in a vicinity of any branchpoint $\l$, where a main and a complementary arcs meet. On the other hand, if an endpoint $\l$ of a main arc is an interior point of another main arc, then there are at least three zero level curves of $\Re h(z)$ emanating from~$\l$. Combined with the requirement $g_\pm \in L^2_{\rm loc}(\O)$, we see that
\begin{gather}
\Re h(z) = \mathcal O (z-\lambda_k)^\frac 32 ,\qquad z\to\lambda_k,\label{modeq}
\end{gather}
at every branchpoint $\lambda_k$, $k=0,1,\dots,2L+1$. In other words, $\Re h(z)\ra 0$ at least at the order $(z-\lambda_k)^\frac 32 $ (or faster) as $z\ra\lambda_k$.

Consider equation \eqref{hform} in the light of the latter requirement. Moving $z$ to a vicinity of a~branchpoint $\lambda_k$, $k=0,\dots, 2L+1$, one has to cross some loops $\hg_{m,j}$, $\hg_{c,l}$, which yield terms $2\pi i\o_j$ and $\pm 2\pi i \eta_l$ respectively, $j,l=1,\dots, L$. These terms do not af\/fect $\Re h(z)$. It now follows from~\eqref{hform} that $\Re h(z) = \mathcal O (z-\lambda_k)^\frac 12$, $k=0,1,\dots,2L+1$. Then, the conditions \eqref{modeq} can be equivalently restated as
\begin{gather*}
 \oint_{\hat{\mathfrak M}}{{f(\z)+\ell}\over{(\z-\lambda_k)R(\z)_+}}d\z
+\sum_{j=0}^L \oint_{\hg_{m,j}}{{2\pi i\varpi_j d\z}\over{(\z-\lambda_k)R(\z)_+}}\nonumber\\
\qquad{} +\sum_{j=1}^L\oint_{\hg_{c,j}}{{2\pi i\eta_j d\z}\over{(\z-\lambda_k)R(\z)}}+ \int_{\widetilde \g_0}{{2\pi i d\z}\over{(\z-\lambda_k)R(\z)}}=0,\label{modeq-long}
\end{gather*}
where $k=0,1,\dots,2L+1$ and $\widetilde \g_0$ is a positively oriented loop around~$\g_0$. Assuming that all the constants $\ell$, $\eta_j$, $\o_j$, $j=1,2,\dots,L$, are known (see Section~\ref{sec-genusL}), the system~\eqref{modeq-long} is a system of $2L+2$ equations~\eqref{modeq} for the $2L+2$ unknown branchpoints~$\lambda_k$. It is known as the system of {\it modulation equations} that governs the location of the endpoints $\lambda_k$ provided the genus $L$ of the corresponding hyperelliptic Riemann surface $\mathfrak R(t)$ is given.

The modulation equations (\ref{modeq}) guarantee that there are {\it at least} three zero level curves of~$\Re h$ emanating from each branch-point $\lambda_k$, $k=0,1,\dots,2L+1$. They do not guarantee, however, that the sign requirements for~$h$ are satisf\/ied. The latter requires a particular choice of the genus~$L$ of $\mathfrak R(t)$ depending on~$t$.

The modulation equations \eqref{modeq} imply that solutions of the RHPs for~$g$ and for~$h$ commute with dif\/ferentiation, since the limiting values of $g_\pm$, $h_\pm$ and $g'_\pm$, $h'_\pm$ along $\O$ belong to $L^2_{\rm loc}(\O)$. Thus, the modulation equations~\eqref{modeq} can be equivalently stated as
\begin{gather*} 
h'(z) = \mathcal O\big(\sqrt{z-\lambda_k}\big) ,\qquad z\to\lambda_k, \qquad k=0,1,\dots,2L+1
\end{gather*}
(note that $h'(z)$ is not continuous on the main arcs). The function $h'$ can be obtained as the solution of the RHP for~$h'(z)$ with the following jump conditions and the asymptotics at inf\/inity
\begin{gather}
h'_+(z) + h'_-(z) =0,\qquad
 z\in\g_{m,j}, \qquad j=0,1,\dots,L,\nonumber\\ 
h'(z) = -f'(z) +\frac 2z + \mathcal O\big(z^{-2}\big) ,\qquad z\ra\infty.\label{jumpassh'}
\end{gather}
In view of \eqref{h-def}, the RHP for $g'(z)$ (the jump conditions and asymptotics at inf\/inity) is given by
\begin{gather*}
g'_+(z) + g'_-(z) =f'(z),\qquad z\in\g_{m,j}, \qquad j=0,1,\dots,L, \nonumber\\ 
g'(z) = \frac 1z + \mathcal O\big(z^{-2}\big) ,\qquad z\ra\infty.\label{jumpassg'}
\end{gather*}
Then solution to \eqref{jumpassg'} is given by
\begin{gather}\label{g'gen}
g'(z)={{R(z)}\over{4\pi i}} \oint_{\hat{\mathfrak M}}{{f'(\z)}\over{(\z-z)R(\z)}}d\z,
\end{gather}
where $z$ is outside the contour of integration. Using the residue theorem, we see that
\begin{gather}\label{h'gen}
h'(z)={{R(z)}\over{2\pi i}} \oint_{\hat{\mathfrak M}}{{f'(\z)}\over{(\z-z)R(\z)}}d\z=R(z)M(z),
\end{gather}
where $z$ is now {\it inside} the contour $\hat{\mathfrak M}$. Equation~\eqref{h'gen} shows that $M(z)$ is a polynomial of degree $\deg f' -L-1$. If~$R(z)$ (that is, the Riemann surface $\mathfrak R(t)$) is given, then $M(z)$ is def\/ined by the asymptotics in~\eqref{jumpassh'}. This argument, in particular, shows that for a quartic polynomial~$f$ the genus $L$ of $\mathfrak R(t)$ cannot exceed~$2$. In fact, our Riemann surface $\mathfrak R(t)$ is the nodal curve
\begin{gather}\label{y-curve}
y^2(z)=h'^2(z)=M^2(z)\prod_{j=0}^{2L+1}(z-\lambda_j)=P(z),
\end{gather}
where $P(z)$ is a polynomial and $\deg P= 2\deg f -2$. In case of a general potential $f(z;\vec t)$, the genus $L$ is bounded by
\begin{gather}\label{max-gen}
L\leq \deg f -2.
\end{gather}

The modulation equations now can be again recast as the requirements that: i)~the meromorphic dif\/ferential~$h'(z)dz$ on the hyperelliptic surface $\mathfrak R(t)$ has the asymptotics as required in~\eqref{jumpassh'}, and ii)~the Riemann surface~$\Rscr(t)$ satisf\/ies the {\it Boutroux condition} for all $t\in\C^*$, that is, all the periods of~$h'(z)dz$ are purely imaginary (for all $t\in\C$). The corresponding equations can be written as
\begin{gather}\label{modeq-g'}
{1\over{4\pi i}} \oint_{\hat{\mathfrak M}}{{\z^kf'(\z)d\z}\over{R(\z)}}=-\d_{k,L+1},\qquad
\Re \oint_{\hg_{m,j}}h'(\z)d\z=0, \qquad \Re \oint_{\hg_{c,j}}h'(\z)d\z=0,
\end{gather}
where $k=0,1,\dots,L+1$, $j=1,\dots,L$ and $\d_{k,l}$ is the Kronecker delta. The f\/irst $L+2$ equations in~\eqref{modeq-g'}, also known as moment conditions, express the fact that the large~$z$ asymptotics of~$g'(z)$ is given by~\eqref{jumpassg'} (requirement~i)). The remaining equations~\eqref{modeq-g'}, also known as integral conditions~\cite{TVZ1}, express the requirement~ii).

Equivalence of modulation equations \eqref{modeq} and \eqref{modeq-g'} follows from the fact that the solution of the RHP for~$g$,~$h$ commutes with the dif\/ferentiation and from the existence and uniqueness of the solutions of the RHPs for $g$, $h$, $g'$, $h'$.

\subsection{Schematic conclusion of the steepest descent analysis}\label{sect-stdesc}

In this subsection we brief\/ly describe the construction of the leading order asymptotics in the case of a positive genus. The corresponding construction for the genus zero case was given in~\cite{BT3}. For convenience of the reader, we repeat here some common parts (for any genus).

Fix some $t\in \C^*\setminus\{t_0,t_2\}$. Let $\vec\l=(\lambda_0,\dots,\lambda_{2L+1})$ be a solution to the modulation equations with some $L\in \N\cup\{0\}$. Deform the contour $\O$ in such a way that $\lambda_k\in\O$ for all $k=0,\dots,2L+1$ and partition it into a collection of the main and complementary arcs as discussed in Section~\ref{sec-reqg}. If the $g$-function \eqref{gform2} satisf\/ies all the requirements of Section~\ref{sec-reqg} (equalities and inequalities), then we say that $t$ belongs to the genus $L$ region of $\C$ and proceed with the construction of the leading order asymptotics of the recurrence coef\/f\/icients, as outlined below. The choice of the correct genus for a given $t\in \C^*\setminus\{t_0,t_2\}$ will be discussed in Section~\ref{sect-cont-princ}.

The remaining steps in the steepest descent analysis involve inserting additional contours, the {\it lenses}, which enclose each main arc and lie entirely within the $-\Re h<0$ region, the so-called sea. (On the panels in Fig.~\ref{Generic}, these regions (the sea) are shown in blue (darker) color.) One then re-def\/ines $T(z)$ from the RHP~\eqref{RHPT} within the regions between the main arc and its corresponding lens by using the factorization
\begin{equation}\label{mainfact}
\begin{pmatrix} a&d\\0&a^{-1}\end{pmatrix}=
\begin{pmatrix} 1&0\\a^{-1}d^{-1}&1\end{pmatrix}
\begin{pmatrix} 0&d\\-d^{-1}&0\end{pmatrix}
\begin{pmatrix} 1&0\\ad^{-1}&1\end{pmatrix}
\end{equation}
of the jump matrices of $T$ so that
\begin{gather*}
T_+(z)= T_-(z)
\le[
\begin{matrix}
1 & 0\\
\nu_m^{-1}{\rm e}^{-Nh_-} & 1
\end{matrix}
\ri]
\le[
\begin{matrix}
 0& \nu_m\\
-\nu_m^{-1} & 0
\end{matrix}
\ri]
\le[
\begin{matrix}
1 & 0\\
\nu_m^{-1}{\rm e}^{-Nh_+} & 1
\end{matrix}
\ri],
\end{gather*}
where $\nu_m$ is the (constant!) value of $\nu(z)$ on a given main arc $\g_m$ under consideration. (For simplicity, we use here notation $\g_m$ instead of $\g_{m,j}$ for a given particular main arc.) Therefore, def\/ining $\wh T(z)$ as $T$ outside of the lenses and by
\begin{gather*}
\wh T(z):= T(z)
\le[
\begin{matrix}
1 & 0\\
\mp \nu_m^{-1}{\rm e}^{-Nh} & 1
\end{matrix}
\ri]
\end{gather*}
in the regions within the lenses and adjacent to the $\pm$ sides of $\gamma_m$ one obtains a~new RHP for~$\hat T(z)$. In this RHP, the jump matrices on the lenses and on the complementary arcs turn out to be exponentially close to the identity in any~$L^p$, ($1\leq p<+\infty$) as $N\ra\infty$ due to the inequality requirements of Section~\ref{signreqs}, see, for example,~\cite{BertolaMo} for details. (Because of the asymptotics~\eqref{asshgen}, where $f(z)$ is a~polynomial, we can always assume that $\sup_{z\in\O}\Re h(z)<0$ for all suf\/f\/iciently large $z\in\O$; that establishes the above statement for any unbounded complementary arc.) The problem corresponding to the remaining jumps can be solved exactly. The fact that the neglected jump matrices do not approach the identity matrix $\1$ in $L^\infty$ is addressed by the construction of appropriate local solutions of the RHP called parametrices~\cite{BertolaMo, DKMVZ}. The type of local RHP depends on the behavior of $h(z)$ near the endpoints. For regular values of $t$ (away from the breaking curves) we use the standard Airy parametrix at any branchpoint~$\lambda_j$ where a main arc meets a complementary arc. This Airy parametrix should be appropriately modif\/ied~\cite{BertolaMo} if
at least three main arcs come together at~$\lambda_j$. Finally, if $t$ is a regular point on the breaking curve and if $\l$ is one of the interior points of $\mathfrak{M}\cup \mathfrak{C}$, where the sign requirement of Section~\ref{signreqs} is violated, the local parametrix near~$\lambda_j$ is of the same type as in~\cite{BT1}. Consider, for simplicity, a~regular point~$t$. In the f\/inal steps of the approximation we f\/ix suf\/f\/iciently small disks~$\mathbb D_j$,
around the endpoints $\lambda_j$, $j=0,\dots,L+1$, and def\/ine a suitable approximate solution
\begin{gather*}
\Phi(z):= \begin{cases}
\Phi_{\rm ext}(z) & \text{for $z$ outside $\cup_j\mathbb D_{j}$},\\
\Phi_{\rm ext}(z) \mathcal P_j(z) & \text{inside $\mathbb D_j$},
\end{cases}
\end{gather*}
such that the error matrix $\mathcal E(z):= \wh T(z) \Phi^{-1}(z)$ solves a small-norm Riemann--Hilbert problem (as $N\to \infty$) and thus can be~-- in principle~-- completely solved in Neumann series. Here by~$\mathcal P_{j}(z)$ we denote the parametrices near the endpoints $\lambda_{j}$ respectively, $j=0,\dots,L+1$. The RHP for $\Phi_{\rm ext}=\Psi_0$ (``model solution'' or ``exterior parametrix'') will be discussed below, see~\eqref{RHPPsi_0-hi}.

Once we have achieved a suitable approximation for $\wh T(z)$, the recurrence coef\/f\/icients for the orthogonal polynomials can and will be recovered via Proposition~\ref{proparno3}:
\begin{gather*}
\h_n = -{2i\pi}{\rm e}^{N\ell} \le(T_1\ri)_{12} ,\qquad\a_n = \le(T_1\ri)_{12} \le(T_1\ri)_{21} ,\qquad
\b_n = \frac {\le(T_2\ri)_{12}}{\le(T_1\ri)_{12}} - \le(T_1\ri)_{22}.
\end{gather*}
Here we used \eqref{defT} and the fact that $\wh T(z)$ near $\infty$ equals $T(z)$ (since we are in the exterior region) and has expansion
\begin{gather*}
\wh T(z) = T(z) = \1 + \frac {T_1}z + \frac {T_2} {z^2} + \cdots ,\qquad z\to\infty .
\end{gather*}
The matrix entry $T_{1,2}$ can be obtained from the corresponding expansion of $\Phi_{\rm ext}(z)$ near inf\/inity, to within the error determined by the error matrix $\mathcal E=\wh T\Phi^{-1}$. In the case of a regular~$t$, the standard error analysis (which we do not report here) shows that $\mathcal E$ introduces an error of order~$\mathcal O(N^{-1})$ that is uniform on compact subsets of regular~$t$.

The higher genus case (i.e., more than one main arc) with real potentials and on the real line was f\/irst fully treated in~\cite{DKMVZ} and was later extended to the complex plane and complex potentials in~\cite{BertolaMo}. Consider f\/irst the case of interlacing main and complementary arcs, i.e., the region inside the two Schwarz-symmetrical curves connecting the points $t_0=-\frac 1 {12}$ and $t_2=\frac 1 {4}$ on Fig.~\ref{Generic} (we, obviously, exclude the genus zero region bounded by contours connecting the points $t_0=-\frac 1 {12}$ and $t=0$).

\begin{Remark} Let $\rho_j$ denote the traf\/f\/ic on the main arc $\g_{m,j}$, $j=0,1,\dots,L$. Without any loss of generality, see Remark~\ref{workingcont}, we assume that $\rho_0=1$. Alternatively, we can replace~$\rho_j$ by~$\rho_j/\rho_0$ in the formulae below.
\end{Remark}

Let
\begin{gather*}
\s_1= \begin{bmatrix}
0 & 1\\
1&0
\end{bmatrix},\qquad
\s_2= \begin{bmatrix}
0 & -i\\
i&0
\end{bmatrix},\qquad
\s_3= \begin{bmatrix}
1 & 0\\
0&-1
\end{bmatrix}
\end{gather*}
denote the Pauli matrices. Following the standard arguments, we obtain the following RHP (model problem)
\begin{gather}
\begin{array}{@{}ll}
 \Psi_0(z) & \text{is analytic in} \ \C\setminus \wh\O, \\
 \Psi_{0+}(z)=\Psi_{0-}(z)e^{(2\pi i \varpi_j N+\ln\rho_j)\s_3} i\s_2& \text{on} \ \g_{m,j}, \quad j=0,1,\dots, L, \\
 \Psi_{0+}(z)=\Psi_{0-}(z)e^{-2\pi i \eta_j N\s_3}& \text{on} \ \g_{c,j}, \quad j=1,\dots, L,\\
\Psi_0 (z)=\1+\mathcal O(z^{-1})& {\rm as} \ z\ra\infty, \\
\Psi_0(z) = \mathcal O\big((z-\lambda_k)^{-\frac 14}\big) , & z\to\lambda_{k}, \quad k=0,1,\dots, 2L+1
\end{array}
\label{RHPPsi_0-hi}
\end{gather}
for the``model solution'' or ``exterior parametrix'' matrix $\Psi_0(z)$. Here $\wh \O$ is the collection of all main and all bounded complementary arcs, the real constants $ \varpi_j$, $\eta_j$ are given by~\eqref{jumph} and the constants $\rho_j\in \C^*$ are equal to the traf\/f\/ic on the oriented main arc $\g_{m,j}$, obtained by deformation of the original contours into~$\O$. In the case with whiskers, that is, when three main arcs come together, some jumps $\g_{c,j}$ are created according to Remark~\ref{rem-whiskers}. In this case, the diagonal jump matrix on $\g_{c,j}$ in~\eqref{RHPPsi_0-hi} will have additional~$\ln \rho$ terms in the exponential.

\subsection{Solution of the model problem}\label{sect-sol-mod}

To solve the model RHP \eqref{RHPPsi_0-hi} in terms of Riemann theta-functions we follow the approach of~\cite{DIZ}, see also~\cite{TVZ1} (for an alternative approach in the case of real main arcs see~\cite{ArnoMo}). First, we use the transformation $\widetilde \Psi_0(z)=e^{-N\widetilde g(\infty)\s_3}\Psi_0(z)e^{N\widetilde g(z)\s_3}$ in order to remove the jumps $e^{-2\pi i \eta_j N\s_3}$ of~$\Psi_0(z)$ on the complementary arcs. The function $\widetilde g$ is analytic in $ \C\setminus \wh\O$ and have constant jumps on~$\wh\O$. If~$V$ denotes the jump matrix for~$\Psi_0(z)$, then the jump matrix $\widetilde V$ for $\widetilde \Psi_0(z)$ is
\begin{gather} \label{new-model-jump}
\widetilde V=e^{-N\widetilde g_-(z)\s_3}Ve^{N\widetilde g_+(z)\s_3}.
\end{gather}
To remove the jumps on $\g_{c,j}$, we require that
\begin{gather*}
\widetilde g_+-\widetilde g_-=2\pi i \eta_j \quad \text{on} \ \g_{c,j} \qquad {\rm and} \qquad \widetilde g_+ +\widetilde g_-=2\pi i \k_j\quad \text{on} \ \g_{m,j}, \quad j=1,\dots, L,
\end{gather*}
where the complex constants $\k_j$ are def\/ined by the requirement that $\widetilde g$ is analytic at $z=\infty$. Then $\widetilde \Psi_0(z)$ solves the RHP
\begin{gather}
\begin{array}{@{}ll}
\widetilde \Psi_0(z) & \text{is analytic in} \ \C\setminus \wh\O, \\
\widetilde \Psi_{0+}(z)=\widetilde \Psi_{0-}(z)e^{(2\pi i N\widetilde \varpi_j+\ln\rho_j)\s_3} i\s_2& \text{on} \ \g_{m,j}, \quad j=0,1,\dots, L, \\
\widetilde \Psi_0 (z)=\1+O\big(z^{-1}\big)& \text{as} \ z\ra\infty, \\
\widetilde \Psi_0(z) = \mathcal O(z-\lambda_k)^{-\frac 14},& z\to\lambda_{k}, \quad k=0,1,\dots, 2L+1,
\end{array}\label{RHPtPsi_0}
\end{gather}
where $\widetilde \varpi_j=\varpi_j+\k_j$, $j=1,\dots, L$, and $\widetilde \varpi_0=0$.

The solution to the RHP \eqref{RHPtPsi_0} is known \cite{DIZ,TVZ1}. To present it, we need to introduce the Riemann theta function associated with the hyperelliptic Riemann surface $\Rscr(t)$. Let $A_j=\gt_{m,j}$, $j=1,\dots,L$ be A-cycles of~$\Rscr$, whereas the (standard) set of B-cycles $B_j$ consists of the properly oriented loops passing through the branchcuts $\g_{m,0}$ and $\g_{m,j}$. The basis $\vec\o=(\o_1,\dots,\o_L)$ of normalized holomorphic dif\/ferentials, dual to the A-cycles, is def\/ined by
\begin{gather}\label{ao}
\int_{\gt_{m,j} }\o_k=\d_{jk},\qquad j,k=1,\dots, L,
\end{gather}
where $\d_{jk}$ denotes the Kronecker symbol. Each of these dif\/ferentials $\o_j$ has the form $\o_j=\frac{P_j(z)}{R(z)}dz$, where $P_j(z)$ is a~polynomial of degree less than~$L$. The Riemann period matrix is def\/ined
\begin{gather}\label{rpm}
\t=(\t_{kj})=\left( \int_{B_j}\o_k\right), \qquad k,j=1,\dots, L.
\end{gather}
$\t$ is known to be symmetric and have positive def\/inite imaginary part. The Riemann theta function is def\/ined as
\begin{gather*}
\Th(s)=\sum_{l\in\Z^{2N}}e^{2\pi i (l,s)+\pi i (l,\t l)}, \qquad s\in \C^{2N}.
\end{gather*}
It satisf\/ies
\begin{gather}
\text{(a)} \ \ \Th(s)=\Th(-s),\qquad
\mbox{(b)} \ \ \Th(s+e_j)=\Th(s),\nonumber\\
\mbox{(c)} \ \ \Th(s+\t_j)=e^{- 2\pi i s_j-\pi i \t_{jj}}\Th(s),\label{propth}
\end{gather}
where $e_j$ is the $j$-th column of $\1_{L\times L}$ and $\t_j=\t e_j$. The set $\Lambda=\Z^{L} + \t \Z^{L} $ is called {\it the period lattice} and $\mathbb J = \C^L \; {\rm mod}\, \Lambda$ the {\it Jacobian}. Def\/ine the Abel map by
\begin{gather*}
\mathfrak u(z)=\int_{\lambda_0}^z\vec\o,
\end{gather*}
and the vector of functions
\begin{gather}\label{mfun}
\Mscr(z,d)\equiv (\Mscr_1,\Mscr_2)=\left(\frac{\Th(\mathfrak u(z)-{\widehat {W} }+d)}{\Th(\mathfrak u(z)+d)},\frac{\Th(-\mathfrak u(z)-{\widehat {W} }+d)}{\Th(-\mathfrak u(z)+d)}\right),
\end{gather}
where
\begin{gather}\label{hatW}
\widehat {W} =\big(\widehat {W} _1,\dots,\widehat {W} _L\big)=\left ( N\widetilde \varpi_1 +\frac{\ln\rho_1}{2\pi i},\dots, N\widetilde \varpi_L +\frac{\ln\rho_L}{2\pi i}\right) ^t
\end{gather}
and $d\in\C^{L}$ is a vector that will be chosen appropriately later on. Note that
\begin{gather}
\mathfrak u_++\mathfrak u_-=-\t_j \quad {\rm on} \ \g_{m,j}, \quad j=0,1,\dots,L, \quad {\rm where} \ \t_0=0, \qquad {\rm and} \nonumber\\
\mathfrak u_+= \mathfrak u_- - \sum_{k=j}^Le_k \quad {\rm on} \ \g_{c,j}, \quad j=1,\dots,L,\label{jumpu}
\end{gather}
and $\mathfrak u(\lambda_0)=0$,
\begin{gather}
 \mathfrak u(\lambda_{2j+1})=-\hf\left(\t_j + \sum_{k=j+1}^L e_k\right), \qquad j=0,1,\dots,L,\nonumber\\
\mathfrak u(\lambda_{2j})=-\hf\left(\t_j + \sum_{k=j}^L e_k\right),\qquad j=1,\dots,L.\label{valu}
\end{gather}
Although $\mathfrak u(z)$ is multivalued, $\Mscr(z,d)$ is single-valued and meromorphic on~$\C\setminus\mathfrak{M}$. Moreover, according to \eqref{jumpu} and the third equation in~\eqref{propth}, $\Mscr$ satisf\/ies
\begin{gather*}
\Mscr_+=\Mscr_- \left(
\begin{matrix}
0&e^{{2\pi i}\widehat {W} _j}\\e^{{-2\pi i}\widehat {W} _j} &0
\end{matrix}\right)
\end{gather*}
on $\g_{m,j}$, $j=0,1,\dots,L$, where $\widehat {W} _0=0$.

\begin{Theorem}[{\cite[p.~308]{FK}}]\label{generalTheta}
Let ${\bf f}\in \C^L$ be arbitrary, and denote by $\mathfrak u(p)$ the Abel map $($extended to the whole Riemann surface $\Rscr)$. The (multi-valued) function $\Theta( \mathfrak u(z) - {\bf f})$ on the Riemann surface either vanishes identically or vanishes at $L$ points ${p}_1,\dots, {p}_L$ $($counted with multiplicity$)$. In the latter case we have
\begin{gather}\label{zero-Th}
{\bf f} =\sum_{j=1}^{L} \mathfrak u(p_j) + \mathcal K,
\end{gather}
where $\mathcal K = \sum\limits_{j=1}^L \mathfrak u(\lambda_{2j})$.
\end{Theorem}

The sign of $ \mathcal K$ in \eqref{zero-Th} is irrelevant since $2\mathcal K \in \Lambda$. An immediate consequence of Theorem~\ref{generalTheta} is the following statement.
\begin{Corollary} \label{generalThetadiv}
The function $\Theta$ vanishes at ${\bf e}\in \Lambda$ if and only if there exist $L-1$ points $q_1,\dots, q_{L-1}$ on $\Rscr$ such that
\begin{gather*}
{\bf e} =\sum_{j=1}^{L-1} \mathfrak u(q_j) + \mathcal K.
\end{gather*}
\end{Corollary}

\begin{Definition}\label{def-thetadiv}
The {\it Theta divisor} is the locus ${\bf e}\in{\mathbb J}$ such that $\Theta({\bf e})=0$. It will be denoted by the symbol~$(\Theta)$.
\end{Definition}

Let
\begin{gather}\label{r_of_z}
r(z)=\left(\prod_{j=0}^L\frac{z-\lambda_{2j+1}}{z-\lambda_{2j}} \right)^\qt,
\end{gather}
where the branch cuts of $r(z)$ are along the main arcs $\mathfrak M$ and $r(z)$ is normalized by $\lim\limits_{z\ra\infty}r(z)=1$. One can also construct the solution to the RHP~\eqref{RHPtPsi_0} in the form of
 \begin{gather}\label{lscr}
\Lscr(z)=\hf\left(
\begin{matrix}
(r(z)+r^{-1}(z))\Mscr_1(z,d)& -i(r(z)-r^{-1}(z))\Mscr_2(z,d) \\i(r(z)-r^{-1}(z))\Mscr_1(z,-d)
&(r(z)+r^{-1}(z))\Mscr_2(z,-d)
\end{matrix}\right),
\end{gather}
where the vector $d\in\C^L$ is to be def\/ined and $\Mscr(z,d)$, $r(z)$ were given by~\eqref{mfun},~\eqref{r_of_z} respectively. Indeed, using the facts that $(r\pm r^{-1})_+=i(r\mp r^{-1})_-$, it is straightforward to check that~$\Lscr(z)$ satisf\/ies the jump conditions listed in~\eqref{RHPtPsi_0}.

The vector $d\in\C^L$ is found by requiring the analyticity of $\Lscr(z)$ in $\C\setminus\mathfrak{M}$; to this end, consider the equation
\begin{gather*}
r^4(z)-1=0 \equiv \frac{\prod\limits_{j=0}^L(z-\lambda_{2j+1})-\prod\limits_{j=0}^L(z-\lambda_{2j})}{R(z)}=0.
\end{gather*}
Note that $r^4(z)-1$ is a meromorphic function on the Riemann surface $\Rscr$ with a zero $z_0=\infty_+$ (inf\/inity on the main sheet) and~$2L$ (counted with multiplicities) zeroes $z_1,\dots,z_L$ counted twice (for each of the two sheets of~$\Rscr$). Zeroes $z_1,\dots,z_L$ are f\/inite provided
\begin{gather*}
 \sum_{j=0}^L\lambda_{2j}\neq \sum_{j=0}^L\lambda_{2j+1}.
\end{gather*}
With a slight abuse of notation, we denote by $z_1,\dots,z_L$ zeroes of $r^2(z)-1=0$ other than $z_0$ that can be located on either sheet of~$\Rscr$. It is straightforward to check that if $z$ is a zero of $r^2(z)-1=0$ then $\wh z$ is a zero of $r^2(z)+1=0$, where~$\wh z$ denotes the projection of~$z$ on the other sheet of the hyperelliptic surface~$\Rscr$. If we now choose
\begin{gather*}
d=-\sum_{j=1}^L\int_{\lambda_{2j}}^{\hat z_j}\vec\o,
\end{gather*}
then, according to Theorem~\ref{generalTheta}, the set of all zeroes of~$\Th(\mathfrak u(z)+d)$ in~$\Rscr$ consists of $\hat z_j$, $j=1,\dots,L$. Indeed,
\begin{gather*}
\Th(\mathfrak u(z)+d)=\Th\left(\mathfrak u(z)+\mathcal{K}-\sum_{j=1}^L \mathfrak u(\lambda_{2j})-\sum_{j=1}^L\int_{\lambda_{2j}}^{\hat z_j}\vec\o\right)= \Th\le(\mathfrak u(z)+\mathcal{K}-\sum_{j=1}^L \mathfrak u(\hat z_j)\ri).
\end{gather*}
Since
\begin{gather*}
\Th(\mathfrak u(z_j)-d)=\Th(-\mathfrak u(\wh z_j)-d)=\Th(\mathfrak u(\wh z_j)+d)=0
\end{gather*}
we obtain that the set of all zeroes of $\Th(-\mathfrak u(z)+d)$ in~$\Rscr$ consists of $ z_j$, $j=1,\dots,L$.

Since all the zeroes $ z_j$, $j=1,\dots,L$, of $\Th(\mathfrak u(z)-d)$ are also zeroes of $r(z)-\frac 1{r(z)}$ and all the zeroes~$\wh z_j$, $j=1,\dots,L$ of $\Th(\mathfrak u(z)+d)$ are also zeroes of $r(z)+\frac 1{r(z)}$, we conclude that~$\Lscr(z)$ does not have poles on $\Rscr$. In particular, $\Lscr(z)$ is analytic in $\C\setminus\mathfrak{M}$. It obviously satisf\/ies the boundary requirements of~\eqref{RHPtPsi_0}
at the branchpoints.

\begin{Lemma} \label{lem-vec-d}
\begin{gather}\label{eq-d}
d=-\mathfrak u(\infty_+) \ \mod \,\Lambda,
\end{gather}
where $\Lambda$ is the period lattice.
\end{Lemma}

\begin{proof}
By construction, $r^2(z)-1$ is a meromorphic function on $\Rscr$ with simple poles at $\lambda_{2j}$, $j=0,1,\dots,L$, and simple zeroes at $z_j$, $j=1,\dots,L$, and at inf\/inity on the main sheet $\infty=\infty_+$. Let~$\Dscr$ denote the divisor of $r^2(z)-1$. Then according to Abel's theorem (see~\cite{FK}), $U(\Dscr)=0$ $\mod\,\Lambda$. Thus,
\begin{gather}\label{d-1}
\mathfrak u(\infty)+\sum_{j=1}^L\mathfrak u(z_j)= \sum_{j=0}^L\mathfrak u(\lambda_{2j}) \ \mod\,\Lambda.
\end{gather}
Then \eqref{eq-d} follows from \eqref{d-1},
\begin{gather*}
\sum_{j=1}^L\mathfrak u(z_j)=\sum_{j=0}^L\mathfrak u(\lambda_{2j})+d
\end{gather*}
and $\mathfrak u(\lambda_0)=0$.
\end{proof}

With Lemma~\ref{lem-vec-d}, we obtain the following lemma.

\begin{Lemma}\label{lem-d-exact}
The matrix $\Lscr(z)$ with $d=\mathfrak u(\infty_+)$ and $\widehat {W} $ as in~\eqref{hatW} is analytic in $\C\setminus\mathfrak{M}$ and satisf\/ies
the same jump conditions as $\widetilde \Psi_0(z)$ from the RHP~\eqref{RHPtPsi_0}.
\end{Lemma}

Since $\Lscr(z)$ satisf\/ies the same jump conditions as $\widetilde \Psi_0(z)$, the determinant $\det\Lscr(z)$ must be analytic in~$\C$. On the other hand, according to~\eqref{lscr} and~\eqref{r_of_z}
\begin{gather}
\det \Lscr(\infty)=\Mscr_1(\infty,d)\Mscr_2(\infty,-d)\nonumber\\
\hphantom{\det \Lscr(\infty)}{} =\frac{\Th\big(\mathfrak u(\infty)-\widehat {W} +d\big)\Th\big(\mathfrak u(\infty)+\widehat {W} +d\big)}{\Th^2(\mathfrak u(\infty)+d)}= \frac{\Th^2\big(\widehat {W} \big)}{\Th^2(0)}.\label{det_L-inf}
\end{gather}
$\Lscr(z)$ is bounded at $z=\infty$, so that
\begin{gather*}
\det \Lscr(z)\equiv \frac{\Th^2\big(\widehat {W} \big)}{\Th^2(0)}.
\end{gather*}
Thus, we obtain
\begin{gather*}
\widetilde \Psi_0(z)=\Lscr^{-1}(\infty)\Lscr(z).
\end{gather*}

Consequently we have proven the following analog of Theorem \ref{the0-mod-alt1}.

\begin{Theorem} \label{the0-mod-alt1}
The RHP~\eqref{RHPtPsi_0}{admits} a solution $\widetilde \Psi_0$ if and only if $\Th\big(\widehat {W} \big) \neq 0$.
\end{Theorem}

\begin{Corollary}\label{the0-mod-alt}
If $\widehat {W} =\mathfrak u(\lambda_{2j})$ $\mod\, \Lambda$ for some $j=1,\dots,L$, that is
\begin{gather}\label{nonex-cond-alt}
N\vec{\widetilde \varpi}+\frac{\vec{\ln\rho}}{2\pi i} = \hf\le(\t_j+\sum_{k=j}^Le_k \ri) \ \mod\, \Lambda,
\end{gather}
for some $j=1,\dots,L$, where $\vec{\widetilde \varpi}=(\widetilde \varpi_1,\dots, \widetilde \varpi_L)$, $\vec{\ln\rho}=(\ln\rho_1,\dots,\ln\rho_L)$, then the RHP \eqref{RHPtPsi_0} does not have a solution. Condition~\eqref{nonex-cond-alt} is necessary and sufficient in the case $L=1$.
\end{Corollary}

\begin{proof}
The corollary follows from Corollary \ref{generalThetadiv} and \eqref{det_L-inf}, \eqref{valu} and \eqref{hatW}.
\end{proof}

\begin{Theorem} \label{the0-mod-alhgen}
Assume that the $h$ function $h(z;t)$, defined by the RHP \eqref{jumph}, \eqref{asshgen} with $L+1$ main arcs, satisfies the sign requirements $($strict inequalities$)$ in some compact region $\mathcal T\subset \C^*$. Then the recurrence coefficients $\a_n=\a_n(t)$ and the pseudonorms
${\bf h}_n={\bf h}_n(t)$ satisfy
\begin{gather}\label{alp_n}
\a_n =\frac 1{16}\le(\sum_{j=0}^{L}(\lambda_{2j}-\lambda_{2j+1})\ri)^2\frac{\Th\big(2\mathfrak u_\infty -\widehat {W} \big)\Th\big(2\mathfrak u_\infty +\widehat {W} \big)\Th^2(0)} {\Th^2\big(\widehat {W} \big)\Th^2(2\mathfrak u_\infty )}+O\big(N^{-1}\big), \\
{\bf h}_n =-\frac\pi 2 \le(\sum_{j=0}^{L}(\lambda_{2j}-\lambda_{2j+1})\ri) \frac{\Th\big(2\mathfrak u_\infty +\widehat {W} \big)\Th(0)}{\Th\big(\widehat {W} \big)\Th (2\mathfrak u_\infty )}e^{N\ell}+O\big(N^{-1}\big),\label{h_n}
\end{gather}
in the limit $N\ra\infty$ uniformly in $t\in \mathcal T$ provided that $\Th\big(\widehat {W} \big)$ is separated from zero. Here $\mathfrak u_\infty=\mathfrak u(\infty_+)$ and $\ell$, $\widehat {W} $ are given by \eqref{ell-impr} and \eqref{hatW}, \eqref{t-varp} respectively.
\end{Theorem}

\begin{proof} Note that
\begin{gather}\label{lamdif}
r(z)-r^{-1}(z)=\frac{1}{2z}\sum_{j=0}^{L}(\lambda_{2j}-\lambda_{2j+1})+O\big(z^{-2}\big)
\end{gather}
as $z\ra\infty$. According to \eqref{lamdif} and \eqref{det_L-inf}, the $O(z^{-1})$ (residue) term in the expansion of $\widetilde \Psi_0$ at inf\/inity has $(1,2)$ and $(2,1)$ terms given by
\begin{gather}\label{of-diag}
-i\sum_{j=0}^{L}(\lambda_{2j}-\lambda_{2j+1})\cdot\frac{\Mscr_2(\infty,d)}{4\Mscr_1(\infty,d)}\qquad \text{and} \qquad
i\sum_{j=0}^{L}(\lambda_{2j}-\lambda_{2j+1})\cdot\frac{\Mscr_1(\infty,-d)}{4\Mscr_2(\infty,-d)}.
\end{gather}
Then, combining \eqref{expressalpbet}, \eqref{defT}, \eqref{new-model-jump}, we obtain
\begin{gather*}
\a_n=\frac 1{16}\le(\sum_{j=0}^{L}(\lambda_{2j}-\lambda_{2j+1})\ri)^2\frac{\Th\big(\mathfrak u_\infty -\widehat {W} -d\big)\Th\big(\mathfrak u_\infty +\widehat {W} -d\big)\Th^2(\mathfrak u_\infty +d)}
{\Th\big(\mathfrak u_\infty -\widehat {W} +d\big)\Th\big(\mathfrak u_\infty +\widehat {W} +d\big)\Th^2(\mathfrak u_\infty -d)}.
\end{gather*}

We now use \eqref{expressalpbet}, Lemma \ref{lem-d-exact} and properties \eqref{propth} of theta functions to obtain \eqref{alp_n}, \eqref{h_n}.
The error estimates follow from assumptions of the theorem.
\end{proof}

\begin{Remark}\label{rem-solgen} Theorem \ref{the0-mod-alhgen} is valid for general polynomial potentials $f(z;\vec{t})$.
\end{Remark}

\section[Formal calculation of the $g$-function and $h$-function]{Formal calculation of the $\boldsymbol{g}$-function and $\boldsymbol{h}$-function} \label{modeqsubsec}
In this section, following~\cite{BT3}, we calculate explicitly $g(z;t)$ and $h(z;t)$ in the genus zero case. We will also derive the determinantal expressions for $g$ and $h$ in the higher genera regions for general polynomial potentials $f=f(z;\vec t)$.

\subsection[Explicit form of $g$ and $h$ functions in the genus zero region]{Explicit form of $\boldsymbol{g}$ and $\boldsymbol{h}$ functions in the genus zero region}\label{ghcalc}

In the genus zero case, $\mathfrak M$ consists of a single main arc $\g_m$ with the endpoints $\lambda_{0,1}=\lambda_{0,1}(t)$. Applying the Sokhotski--Plemelj formula for the RHP \eqref{jumpassg'} and using the analyticity of $f'$, we obtain
\begin{gather}\label{g'forman}
g'(z)={{R(z)}\over{4\pi i}} \oint_{\gt_m}{{f'(\z)}\over{(\z-z)R(\z)_+}}d\z,\qquad R(z)=\sqrt{(z-\lambda_1)(z-\lambda_0)},
\end{gather}
where the branchcut of $R(z)$ is $\g_m$ and the branch of $R$ is def\/ined by $\lim\limits_{z\ra\infty}\frac{R(z)}{z}=1$; the contour~$\gt_m$ encircles the contour $\g_m$ and has counterclockwise orientation, and; $z$ is outside~$\gt_m$. Note that if $z$ is inside $\gt_m$, the right hand side of~\eqref{g'forman} is equal to $g'(z)-\hf f'(z)=\hf h'(z)$. The modulation equations~\eqref{modeq-g'} in the case $L=0$ reduce to the following moment conditions:
\begin{gather}\label{mom}
\oint_{\gt_m}{{f'(\z)}\over{R(\z)_+}}d\z=0\qquad {\rm and}\qquad \oint_{\gt_m}{{\z f'(\z)}\over{R(\z)_+}}d\z=-4i\pi.
\end{gather}

{\it We use the moment conditions \eqref{mom} to define the location of the endpoints $\lambda_{0,1}$.}

Since $f$ is a polynomial, the equations (\ref{mom}) can be solved using the residue theorem on equations~(\ref{mom}), setting $\lambda_0=a-b$, $\lambda_1=a+b$. We obtain the system
\begin{gather}\label{modalg}
a+ta\left(a^2+\frac 32 b^2\right)=0\qquad {\rm and}\qquad a^2+\hf b^2 +t \left(a^4+ 3a^2b^2+\frac 38 b^4\right)=2.
\end{gather}
There are two possibilities: $a=0$ or $a\neq 0$. The Generic conf\/iguration of the traf\/f\/ic, see Section~\ref{sectintro}, considered in the paper, forces the f\/irst ({\it symmetric}) case, see~\cite{BT3}. In this case we obtain solutions to the system~(\ref{modalg}) as
\begin{gather} \label{branchpeq}
a=0, \qquad b^2=-\frac{2}{3t}\big(1\mp \sqrt{1+ 12 t}\big),\\
\label{lambdasym}
\lambda_{0,1}=\mp b=\mp \sqrt{-\frac{2}{3t}\big(1 - \sqrt{1+ 12 t}\big)}.
\end{gather}
The choice of the negative sign in (\ref{branchpeq}) comes from the requirement that $b$ is bounded as $t\ra 0$. Observe that for
$t>t_0=-\frac 1 {12}$ the values $\pm b$ coincide with the branch-points, derived in~\cite{ArnoDu}. The second pair of roots $b=\pm \sqrt{-\frac{2}{3t}(1 + \sqrt{1+ 12 t})}$ are sliding along the real axis from $\pm\infty$ to $\pm b_0=\pm\sqrt{8}$ as real $t$ varies from $0^-$ to $t_0$, and are sliding along the imaginary axis from $\pm i\infty$ to $0$ as real $t$ varies from $0^+$ to $+\infty$. At the point
\begin{gather*}
t=t_0=-\frac{1}{12}
\end{gather*}
the two pairs of roots (\ref{branchpeq}) coincide, creating f\/ive zero level curves of $\Re h(z)$ emanating from the endpoints $\pm b_0$ (see \cite{BT3, ArnoDu}). Once the values of branch-points (endpoints) $\lambda_{0,1}=\lambda_{0,1}(t)$ are determined, one can calculate explicitly $g(z)=g(z;t)$ and $h(z)=h(z;t)$, where $h(z)=2g(z) - f(z)- \ell$. We are f\/irst looking for $h'(z)=-(k+tz^2)\sqrt{z^2-b^2}$, where the endpoints~$\pm b$ of~$\g_m$ are given by~(\ref{branchpeq}). The asymptotics in~(\ref{jumpassh'}) implies $k=1+\hf tb^2$, so that
\begin{gather}\label{h'}
h'(z) =- \le[tz^2 +1+\frac{tb^2}{2}\ri]\big(z^2-b^2\big)^\frac 12 = -\le[t z^2 + \frac {\sqrt{1+12 t}}3 + \frac 2 3\ri] \big(z^2 - b^2\big)^{\frac 1 2}.
\end{gather}
Since the branch-cut of the radical is $[-b,b]$ we conclude that $h'(z)$ is an odd function. Direct calculations yield
\begin{gather}\label{h}
h(z) = 2 \ln\frac{z+\sqrt{z^2-b^2}}{b} - \frac z8 \big(2tz^2 +tb^2 +4\big) \big(z^2-b^2\big)^\frac 12.
\end{gather}
It is clear that $h(b)=0$. Combined with (\ref{h'}), that implies
\begin{gather}\label{modh}
h(z)=O(z-b)^\frac 32
\end{gather}
(or a higher power of $(z-b)$). It is easy to check that at $t_0=-\frac{1}{12}$ the order $O(z-b)^\frac 32$ in (\ref{modh}) can be replaced $O(z-b)^\frac 52$. There is the oriented branch-cut of $h(z)$ along the ray $(-\infty,-b)$, where $h_+(z)-h_-(z)=4\pi i$. Because of this jump, the function $h(z)$ does not have $O(z+b)^\frac 32$ behavior near $z=-b$; however, $\Re h(z)$ does have $O(z+b)^\frac 32$ behavior near $z=-b$.

\begin{Remark}\label{rhphremark}
One can verify directly that $h(z)$ satisf\/ies the following RHP:
\begin{enumerate}\itemsep=0pt
\item[1)] $h(z)$ is
analytic (in $z$) in $\C\setminus\{\g_m\cup (-\infty,\lambda_0)\}$ and
\begin{gather*}
h(z) =-\qt tz^4 - \hf z^2+2\ln z -\ell +O\big(z^{-1}\big) \quad \text{as} \ \ z\ra\infty,
\end{gather*}
where
\begin{gather*} \ell=\ln\frac{b^2}{4}-\frac{b^2}{8}-\hf;
\end{gather*}
\item[2)] $h(z)$ satisf\/ies the jump condition
\begin{gather*}
h_+ + h_-=0 \quad {\rm on} \ \ \g_{m} \qquad {\rm and} \qquad h_+-h_-= 4\pi i \quad {\rm on} \ \ (-\infty,\lambda_0).
\end{gather*}
\end{enumerate}
\end{Remark}

This RHP is the genus zero case ($L=0$) of the general RHP for $h$ from Remark \ref{rem-RHPh}.

\begin{Remark}\label{rhphtxremark}
As it was mentioned above, the solution to the scalar RHP for $h$ commutes with dif\/ferentiation in $z$; on the same basis, it commutes with dif\/ferentiation in $t$ as well. Thus, we obtain the following RHP for $h_t$:
\begin{enumerate}\itemsep=0pt
\item[1)] $h_t(z)$ is analytic (in $z$) in $\bar\C\setminus \g_m$ and
\begin{gather}\label{assht}
h_t(z) = -\qt z^4 - \ell_t +O\big(z^{-1}\big) \quad {\rm as } \ \ z\ra\infty,
\end{gather}
where
\begin{gather*}
\ell_t=\frac{3}{32}b^4=\frac{2+12t-2\sqrt{1+12t}}{24t^2};
\end{gather*}
\item[2)] $h_t(z)$ satisf\/ies the jump condition
\begin{gather*}
h_{t+} + h_{t-}=0 \quad {\rm on} \ \ \g_{m}.
\end{gather*}
\end{enumerate}
This RHP has the unique solution
\begin{gather}\label{hthx}
h_t(z)=- \frac{z}{8}\big(2z^2+b^2\big)\sqrt{z^2-b^2}
\end{gather}
that can be verif\/ied directly. The explicit formula of $h_t$ is useful in controlling the sign of $\Re h(z;t)$ when $t$ is varying.
\end{Remark}

\subsection[The $g$-function and $h$-function and their properties in the genus $L$ case]{The $\boldsymbol{g}$-function and $\boldsymbol{h}$-function and their properties in the genus $\boldsymbol{L}$ case}\label{sec-genusL}

For the rest of this section we will be considering the case of a general polynomial potential
\begin{gather}
f(z) = f(z,\vec t)=\sum_{j=1}^{n} \frac {t_j}{j} z^j \label{potential_n}
\end{gather}
with $\vec t=(t_1,\dots, t_{n})$ being a vector of (complex) coef\/f\/icients of $f$. Then the genus~$L$ of the corresponding Riemann surface $\Rscr=\Rscr(\vec t)$ satisf\/ies~\eqref{max-gen}.

In the case of genus $L>0$, we start construction of $g$, $h$ by def\/ining the constants $ \eta_j,\varpi_j\in\R$, $j=1,2,\dots,L$, and $\ell\in\C$. Using~\eqref{assg},~\eqref{gform2} (the second expression) and the expansion $\frac{1}{\z-z}=-\frac 1z\sum\limits_{k=0}^\infty \frac{\z^k}{z^k}$, we obtain
\begin{gather}
\frac 1{2\pi i}\oint_{\hat{\mathfrak M}}{\z^k {[f(\z)-2u(\z)]}\over{R(\z)}}d\z+ \frac \ell{2\pi i}\oint_{\hat{\mathfrak M}}{\z^k d\z \over R(\z)}\nonumber\\
\qquad{}+\sum_{j=1}^L \varpi_j\oint_{\gt_{m,j}}{{\z^k d\z}\over{R(\z)}} +\sum_{j=1}^L \eta_j\oint_{\gt_{c,j}}{{\z^kd\z}\over{R(\z)}}=-2\d_{k,L}u_*,\label{moments}
\end{gather}
$k=0,1,\dots, L$, where $u_*=\lim\limits_{z\ra\infty} u(z)-\ln z$ and $\d_{k,L}$ is the Kronecker delta. It is clear that the second integral in \eqref{moments} is zero for all $k<L$, so that the f\/irst~$L$ equations in~\eqref{moments} is a~system of linear equations for~$2L$ real unknowns $ \eta_j$, $\varpi_j$. Then from the last equation we obtain
\begin{gather}\label{ell}
\ell=2u_*+\frac 1{2\pi i}\oint_{\hat{\mathfrak M}}{\z^L {[f(\z)-2u(\z)]}\over{R(\z)}}d\z+
\sum_{j=1}^L \varpi_j\oint_{\gt_{m,j}}{{\z^L d\z}\over{R(\z)}} +\sum_{j=1}^L \eta_j\oint_{\gt_{c,j}}{{\z^Ld\z}\over{R(\z)}}.
\end{gather}

Taking complex conjugates of the f\/irst $L$ equations \eqref{moments}, we obtain
\begin{gather}
-\frac 1{2\pi i}\overline{\oint_{\hat{\mathfrak M}}{\z^k {(f(\z)-2u(z))}\over{R(\z)}}d\z}
+\sum_{j=1}^L \varpi_j\overline{\oint_{\gt_{m,j}}{{\z^kd\z}\over{R(\z)}}}
+\sum_{j=1}^L \eta_j\overline{\oint_{\gt_{c,j}}{{\z^kd\z}\over{R(\z)}}}=0,\label{moments-conj}\\
k=0,1,\dots, L-1.\nonumber
\end{gather}

The f\/irst $L$ equations from \eqref{moments} combined with \eqref{moments-conj} form a system of linear equations
\begin{gather}\label{eq-const}
D\left(\begin{matrix}\vec\varpi\\ \vec\eta \end{matrix}\right)=\frac 1{2\pi i}
\left(\begin{matrix} \displaystyle \oint_{\hat{\mathfrak M}}{ {(f(\z)-2u(z))}\over{R(\z)}}\vec \varsigma(\z)d\z, &
\displaystyle \overline{-\oint_{\hat{\mathfrak M}}{ {(f(\z)-2u(z))}\over{R(\z)}}\vec \varsigma(\z)d\z} \end{matrix}\right)^t,
\end{gather}
where $\vec\varpi$, $\vec\eta$ are vectors of constants $\varpi_1,{\dots},\varpi_L$ and $\eta_1,{\dots},\eta_L$ respectively, $\vec \varsigma(\z){=} (1,\z,{\dots},\z^{L-1})^t$ and
\begin{gather}\label{M}
D^t=\left(\begin{matrix} \displaystyle\oint_{\hg_{m,1}}\frac{ d\z}{R(\z)} &
\cdots & \displaystyle\oint_{\hg_{m,1}}\frac{\z^{L-1} d\z}{R(\z)} &\displaystyle\overline{ \oint_{\gt_{m,1}}\frac{ d\z}{R(\z)}}
& \cdots &\displaystyle\overline{ \oint_{\gt_{m,1}}\frac{\z^{L-1} d\z}{R(\z)}} \\
\cdots &\cdots & \cdots& \cdots & \cdots & \cdots \\
\displaystyle \oint_{\gt_{m,L}}\frac{ d\z}{R(\z)} &
\cdots & \displaystyle\oint_{\gt_{m,L}}\frac{\z^{L-1} d\z}{R(\z)} &\displaystyle\overline{ \oint_{\gt_{m,L}}\frac{ d\z}{R(\z)}}
& \cdots &\displaystyle\overline{ \oint_{\gt_{m,L}}\frac{\z^{L-1} d\z}{R(\z)}}
\\
\displaystyle\oint_{\gt_{c,1}}\frac{ d\z}{R(\z)} &
\cdots &\displaystyle \oint_{\gt_{c,1}}\frac{\z^{L-1} d\z}{R(\z)} & \displaystyle\overline{ \oint_{\gt_{c,1}}\frac{ d\z}{R(\z)}}
& \cdots &\displaystyle\overline{ \oint_{\gt_{c,1}}\frac{\z^{L-1} d\z}{R(\z)}}
\\
\cdots &\cdots & \cdots& \cdots & \cdots & \cdots \cr
 \displaystyle\oint_{\gt_{c,L}}\frac{ d\z}{R(\z)} &
\cdots & \displaystyle\oint_{\gt_{c,L}}\frac{\z^{L-1} d\z}{R(\z)} &\displaystyle\overline{ \oint_{\gt_{c,L}}\frac{ d\z}{R(\z)}}
& \cdots &\displaystyle\overline{ \oint_{\gt_{c,L}}\frac{\z^{L-1} d\z}{R(\z)}}
\end{matrix}\right).
\end{gather}
Here $*^t$ denotes the transposition. Note that the contours $\gt_{m,j}$ can be taken as $\a$-cycles of the Riemann surface $\Rscr=\Rscr(\vec t)$. Then the contours $\gt_{c,1}$, $\gt_{c,1}\cup\gt_{c,2}$, \dots, $\cup_{j=1}^L\gt_{c,j}$ with the opposite orientation form the $\b$-cycles. Let us denote
\begin{gather*}
D=\left(\begin{matrix} A& B\cr \bar A& \bar B \end{matrix}\right),
\end{gather*}
where $A$ and $B$ are the corresponding $L\times L$ blocks of $D$. Then, taking proper linear combinations of the rows of $A$ we can replace the integrands $\frac{\z^j}{R(\z)}$ with the normalized holomorphic dif\/ferentials $\frac{p_j(\z)}{R(\z)}$ of the Riemann surface $\Rscr$, where~$p_j(\z)$ are polynomials of degree not exceeding $L-1$. Thus the matrix $A$ is nonsingular and
\begin{gather}\label{DAB}
D=\left(\begin{matrix}
A& 0\cr 0& \bar A
\end{matrix}\right)
\left(\begin{matrix}\1& A^{-1}B\cr \1& {\bar A}^{-1}\bar B
\end{matrix}\right),
\end{gather}
so that
\begin{gather*}
D^t= \left(\begin{matrix}\1& \1\cr \th& \bar \th
\end{matrix}\right)
\left(\begin{matrix}
A^t& 0\cr 0& \bar A^t
\end{matrix}\right)=
\left(\begin{matrix}
\1& 0\cr 0& -H
\end{matrix}\right)
\left(\begin{matrix}\1& \1\cr \t& \bar \t
\end{matrix}\right)
\left(\begin{matrix}
A^t& 0\cr 0& \bar A^t
\end{matrix}\right),
\end{gather*}
where $\th^t=A^{-1}B$, $H$ is a lower triangular matrix with all entries on and below the main diagonal being one. Then $\th=-H\t$, where~$\t$ is the Riemann matrix of periods. As it is well known, $\Im \t$ is a positive def\/inite matrix. Thus
\begin{gather*}
|D|=|D^t|=(2i)^L|A|^2 |\Im \t|,
\end{gather*}
where $|D|$ denotes $\det D$, and we have proved the following lemma.

\begin{Lemma}[\cite{TV1}]\label{iM>0}
If all the branchpoints $\lambda_0,\lambda_1,\dots,\lambda_{2L+1}$ are distinct then $(-i)^L|D|>0$.
\end{Lemma}

Lemma \ref{iM>0} shows that there exists a unique set of real constants $\vec\varpi$, $\vec\eta$ satisfying the f\/irst $L$ equations~\eqref{moments}. Of course, this system of equations could be solved with all $\vec\eta=0$ if we allow~$\vec\varpi$ to have complex components. Denoting the corresponding complex vector by~$\vec{\widetilde \varpi}$, we obtain
\begin{gather}
\vec{\widetilde \varpi}=-\frac {A^{-1}}{2\pi i}\oint_{\hat{\mathfrak M}}{ {[f(\z)-2u(\z)]}\over{R(\z)}}\vec \varsigma(\z)d\z\nonumber\\
\hphantom{\vec{\widetilde \varpi}}{}=
-\frac 1{2\pi i}\oint_{\hat{\mathfrak M}} {[f(\z)-2u(\z)]}\vec\o=-\frac 1{2\pi i}\oint_{\hat{\mathfrak M}} {f(\z)}\vec\o+2\mathfrak{u}_\infty.\label{t-varp}
\end{gather}

According to \eqref{uP}, we can also reduce \eqref{ell} to
\begin{gather}\label{ell-impr}
\ell=2u_*+\frac 1{2\pi i}\oint_{\hat{\mathfrak M}}{ {[f(\z)-2u(\z)]\tilde P(\z)}\over{R(\z)}}d\z=
2u_*+\frac 1{2\pi i}\oint_{\hat{\mathfrak M}}{ {f(\z)\tilde P(\z)}\over{R(\z)}}d\z-2\pi^2,
\end{gather}
since, according to \eqref{u-L},
\begin{gather*}
\oint_{\hat{\mathfrak M}}{ {u(\z)\tilde P(\z)}\over{R(\z)}}d\z=\oint_{\hat{\mathfrak M}}u(\z)u'(\z)d\z=\left.\hf u^2(z)\right|^{z=\lambda_0e^{-2\pi i}}_{z=\lambda_0} =-2\pi^2 .
\end{gather*}

In view of Lemma \ref{iM>0}, the system \eqref{eq-const} has a unique solution given by
\begin{gather*}
\left(\begin{matrix}\vec\varpi\\ \vec\eta \end{matrix}\right)=D^{-1}\left(\begin{matrix}\vec r\\ {\vec {\bar r}} \end{matrix}\right),
\end{gather*}
where the vector $\vec r$ denotes the f\/irst $L$ entries in the right hand side of~\eqref{eq-const}. Moreover, this solution is real since the system~\eqref{eq-const} can be written as
\begin{gather*}
\left(\begin{matrix}
\Re A& \Re B\cr \Im A& \Im B
\end{matrix}\right)
\left(\begin{matrix}\vec\varpi\cr \vec\eta \end{matrix}\right)=
\left(\begin{matrix} \Re \vec r \cr \Im \vec r
 \end{matrix}\right).
\end{gather*}
Using \eqref{eq-const} and \eqref{DAB}, we calculate
\begin{gather}\label{om_eta}
\vec\eta=\big(\Im \big(A^{-1} B\big)\big)^{-1}\Im \big(A^{-1}\vec r\big),\qquad \vec\varpi= \Re(A^{-1} \vec r)- \Re \big(A^{-1} B\big) \big(\Im \big(A^{-1} B\big)\big)^{-1}\Im \big(A^{-1}\vec r\big).\!\!\!\!
\end{gather}
Alternatively, one can say that the vector $-(\vec \varpi^t,\vec\eta^t)$ coincides with the last row of the matrix $K(z)\operatorname{diag}((D^t)^{-1},1)$

Now, according to \eqref{moments}, the expression \eqref{hform} for $h(z)$ can be represented as
\begin{gather}\label{hK}
h(z)=\frac{R(z)}{|D|}K(z),
\end{gather}
where
\begin{gather}\label{K}
K(z)= \frac{1}{2\pi i}\!\left| \begin{array}{@{}c@{\,}c@{\,}c@{\,}c@{\,}c@{\,}c@{\,}c@{\,}c@{\,}c@{\,}c@{\,}c@{\,}c@{}}
 \oint_{\gt_{m,1}}\frac{ d\z}{R(\z)} &
\cdots & \oint_{\gt_{m,1}}\frac{\z^{L-1} d\z}{R(\z)} & \overline{ \oint_{\gt_{m,1}}\frac{ d\z}{R(\z)}}
& \cdots & \overline{ \oint_{\gt_{m,1}}\frac{\z^{L-1} d\z}{R(\z)}}& \oint_{\gt_{m,1}}\frac{d\z}{(\z-z)R(\z)}\\
\cdots &\cdots & \cdots& \cdots & \cdots & \cdots & \cdots\\
 \oint_{\gt_{m,L}}\frac{ d\z}{R(\z)} &
\cdots & \oint_{\gt_{m,L}}\frac{\z^{L-1} d\z}{R(\z)} & \overline{ \oint_{\gt_{m,L}}\frac{ d\z}{R(\z)}}
& \cdots & \overline{ \oint_{\gt_{m,L}}\frac{\z^{L-1} d\z}{R(\z)}}
 &\oint_{\gt_{m,L}}\frac{d\z}{(\z-z)R(\z)}\\
 \oint_{\gt_{c,1}}\frac{ d\z}{R(\z)} &
\cdots & \oint_{\gt_{c,1}}\frac{\z^{L-1} d\z}{R(\z)} & \overline{ \oint_{\gt_{c,1}}\frac{ d\z}{R(\z)}}
& \cdots & \overline{ \oint_{\gt_{c,1}}\frac{\z^{L-1} d\z}{R(\z)}}
 &\oint_{\gt_{c,1}}\frac{d\z}{(\z-z)R(\z)}\\
\cdots &\cdots & \cdots& \cdots & \cdots & \cdots & \cdots \\
 \oint_{\gt_{c,L}}\frac{ d\z}{R(\z)} &
\cdots & \oint_{\gt_{c,L}}\frac{\z^{L-1} d\z}{R(\z)} & \overline{ \oint_{\gt_{c,L}}\frac{ d\z}{R(\z)}}
& \cdots & \overline{ \oint_{\gt_{c,L}}\frac{\z^{L-1} d\z}{R(\z)}}
&\oint_{\gt_{c,L}}\frac{d\z}{(\z-z)R(\z)}\\
\oint_{\hat{\mathfrak M}}\frac{\widetilde f(\z)d\z}{R(\z)} & \cdots
& \oint_{\hat{\mathfrak M}}\frac{\z^{L-1} \widetilde f(\z))d\z}{R(\z)} &
 \overline{ \oint_{\hat{\mathfrak M}}\frac{\widetilde f(\z)d\z}{R(\z)}} & \cdots
& \overline{ \oint_{\hat{\mathfrak M}}\frac{\z^{L-1} \widetilde f(\z))d\z}{R(\z)}} &\
\oint_{\hat{\mathfrak M}}\frac{(\widetilde f(\z)+\ell))d\z}{(\z-z)R(\z)}\cr
\end{array}\right|\!\!\!
\end{gather}
with $\widetilde f(\z)=f(\z)-2u(z)$, and where $z$ is inside the loop $\hat{\mathfrak M}$ but outside all other loops. Indeed, multiplying the f\/irst $L$ rows of $K(z)$ by $2\pi i\varpi_1, 2\pi i\varpi_2,\dots,2\pi i\varpi_L$ respectively, the next~$L$ rows of~$K(z)$ by $2\pi i\eta_1,2\pi i\eta_2,\dots,2\pi i\eta_L$, and adding them to the last row, we obtain~\eqref{hK}. Equation~\eqref{hK} gives $2g(z)$ when $z$ is outside the loop~$\hat{\mathfrak M}$.

It is clear that moving $z$ inside the loops $\gt_{m,l}$, $\gt_{c,k}$ would generate residue terms $2\pi i\varpi_l$ and $\pm 2\pi i\eta_k$ (depending on the direction $z$ crosses the oriented loop $\gt_{c,k}$) in the right hand side of~\eqref{hK}. Combining this fact with~\eqref{modeq-long}, we obtain a new form of modulation equations
\begin{gather}\label{Kmod}
K(\lambda_{j})=0,\qquad j=0,1,\dots,2L+1.
\end{gather}

\begin{Theorem} \label{equiv-mod}
Let $\l$ denote an arbitrary branchpoint $\lambda_j$ in the set of $2L+2$ distinct branchpoints $j=0,1,\dots, 2L+1$.
Then the following statements are equivalent: $1)$~$K(\l)=0$; $2)$~$\frac{\part(\vec \varpi,\vec\eta)^t }{\part\l}=0$; $3)$~$\frac{\part h(z)}{\part \l}\equiv 0$ for all $z\in\C$; $4)$~$\frac{\part g(z)}{\part \l}\equiv 0$ for all $z\in\C$.
\end{Theorem}

\begin{proof}Combining \eqref{hform} with the identity
\begin{gather*}
 \frac{\part}{\part \l} \le[\frac{R(z)}{(\z-z)R(\z)}\ri]=-\frac{R(z)}{2(z-\l)(\z-\l)R(\z)},
\end{gather*}
we obtain
\begin{gather}
\frac{\part h(z)}{\part \l}=-\frac{R(z)K(\l)}{2(z-\l)|D|}+{{R(z)}\over{4\pi i}} \Bigg(
\oint_{\hat{\mathfrak M}}{{\le(\frac{\part \ell}{\part \l}-2\frac{\part u(\z)}{\part \l}\ri)d\z }\over{(\z-z)R(\z)}} \nonumber\\
\hphantom{\frac{\part h(z)}{\part \l}=}{} + \sum_{j=1}^L \left[ \frac{\part \varpi_j}{\part \l}\oint_{\hg_{m,j}}{{2\pi i d\z}\over{(\z-z)R(\z)_+}}
+ \frac{\part \eta_j}{\part \l}\oint_{\hg_{c,j}}{{2\pi id\z}\over{(\z-z)R(\z)}} \right]\Bigg),\label{gen-eq}
\end{gather}
where $z$ is inside $\hat{\mathfrak M}$.

According to \eqref{u-L},
\begin{gather}
\frac{\part u(z)}{\part \l}=\hf\int_{\lambda_0}^z\frac{\tilde P(\z)d\z}{(\z-\l)R(\z)}+ \int_{\lambda_0}^z\frac{\tilde P'(\z)d\z}{R(\z)}
\end{gather}
provided that $\l\neq \lambda_0$. If $\l=\lambda_0$, the integral in~\eqref{u-L} seems to have singularity at $z=\lambda_0$, but it is compensated by the extra term $- \le.\frac{\tilde P(\z)}{R(\z)}\ri|_{\z\ra\lambda_0}$ obtained when dif\/ferentiating by~$\lambda_0$. In any case, $\frac{\part u(\z)}{\part \l}$ is analytic outside the closed contour $\hat{\mathfrak M}$ and at~$\z=\infty$. Thus, the f\/irst integral in~\eqref{gen-eq} vanishes for every~$\l$. According to~\eqref{asshgen},
\begin{gather}\label{assg_lam}
\frac{\part h(z)}{\part \l}=O(1).
\end{gather}
Then, $K(\l)=0$ implies the system of equations
\begin{gather}\label{syst-deta-dlam}
\sum_{j=1}^L \left[ \frac{\part \varpi_j}{\part \l}\oint_{\hg_{m,j}}{{2\pi i\z^k d\z}\over{R(\z)_+}}
+ \frac{\part \eta_j}{\part \l}\oint_{\hg_{c,j}}{{2\pi i\z^kd\z}\over{R(\z)}} \right]= 0,
\qquad k=0,1,\dots,L-1.
\end{gather}
Considering \eqref{syst-deta-dlam} together with $L$ complex conjugated equations, we obtain the system~\eqref{eq-const} with zero right hand side.
Then, according to Lemma~\ref{iM>0}, $\frac{(\vec \varpi,\vec\eta)^t }{\part\l}=0$. Hence, we proved that 1)~implies~2). Similarly, \eqref{assg_lam} combined with $\frac{(\vec \varpi,\vec\eta)^t }{\part\l}=0$ imply~3), that is, 2)~implies~3). Let us assume~3). Then, according to~\eqref{h-def},~\eqref{assg}, $\frac{\part \ell}{\part \l}=0$ and so $\frac{\part g(z)}{\part \l}\equiv 0$ for all $z\in\C$. Thus, 3)~implies~4).

Let us now assume 4). Then, dif\/ferentiating in $\l$ the jump conditions \eqref{jumpg-gam_0}--\eqref{jumpgmain} for $g(z)$, we obtain $\frac{\part \ell}{\part \l}=0$ ($\ell$ is the jump on $\g_{m,0}$) and $\frac{\part(\vec \varpi,\vec\eta)^t }{\part\l}=0$ (the jumps on the remaining main and complementary arcs). Now 1) follows from~\eqref{gen-eq},~\eqref{assg_lam}.
\end{proof}

\begin{Remark}\label{rem-anal-t} According to \eqref{hthx}, in the case of the quartic potential $f(z;t)$, the function~$h(z;t)$ is analytic in~$t$ in the genus zero region when $z$ is on the Riemann surface~$\Rscr(t)$ away from the branchpoints. But, according to~\eqref{hK},~\eqref{K}, in general, in the higher genera regions the function is only smooth in $t$ since the determinant $K$ depends on both~$t$ and $\bar t$.
\end{Remark}

Fix some $j\in\{0,1,\dots,2L+1\}.$ The equation $K(\lambda_{j})=0$ def\/ines a co-dimension one mani\-fold~$\L_j$ in the~$\C^{2L+2}$ space of all branchpoints~$\lambda_j$. Then Theorem~\ref{equiv-mod} implies the following corollary.

\begin{Corollary}\label{dhdxt=0}
Let potential $f=f(z;\vec t)$ be given by \eqref{potential_n}. If $\vec\l\in\L$, where $\L=\cap_{j=0}^{2L+1}\L_j$, then
\begin{gather*}
\frac{d}{dt_j}h(z;\vec t)\equiv \frac{\part }{\part t_j}h(z;\vec t),\qquad
\frac{d}{dt_j}g(z;\vec t)\equiv \frac{\part }{\part t_j}g(z;\vec t),\\
\frac{d}{dt_j}\varpi_k(\vec t)= \frac{\part }{\part t_j}\varpi_k(\vec t),\qquad
\frac{d}{dt_j}\eta_k(\vec t)= \frac{\part }{\part t_j}\eta_k(\vec t).
\end{gather*}
Similar formulae are valid for the derivatives in $\bar {t}_j$.
\end{Corollary}

According to Corollary \ref{dhdxt=0} and \eqref{potential_n}, modulation equation from~\eqref{Kmod} can be written in the so-called {\it hodograph form} as
\begin{gather*}
\sum_{m=1}^nt_mK_{t_m}(\lambda_j)+K_0(\lambda_j)=0,\qquad j=0,1,\dots,2L+1,
\end{gather*}
where $K_{t_m}(z)$, $K_0(z)$ can be obtained from $K(z)$ by replacing $\widetilde f$ in \eqref{K} by $\frac{\z^m}{m}$ and by $-2u(\z)$
respectively.

Let us now rewrite the equations \eqref{om_eta} in terms of the data of the Riemann surface $\Rscr=\Rscr(\vec t)$. If $\vec\o$ denotes the vector of normalized holomorphic (f\/irst kind) dif\/ferentials and $\t$ of~$\Rscr$, then $A^{-1}\vec r=\frac{1}{2\pi i}\oint_\mathfrak{M}(f-2u)\vec\o$, $A^{-1}B= \t^t(-H^t)$, so that \eqref{om_eta} become
\begin{gather}
\vec\eta=J\big(\Im \t^t\big)^{-1}\Im\le[\frac{1}{2\pi i}\oint_\mathfrak{M}(f-2u)\vec\o\ri],\nonumber\\
\vec\varpi=\Re\le[\frac{1}{2\pi i}\oint_\mathfrak{M}(f-2u)\vec\o\ri]-\Re \big(\t^t\big({-}H^t\big)\big) \vec\eta,\label{eta-nu}
\end{gather}
where $J=-(H^t)^{-1}$ is the $L\times L$ Jordan block with the eigenvalue $-1$. Since $f$ is linear in the entries $t_m$ of $\vec t$,
 Corollary \ref{dhdxt=0} implies another set of hodograph equations
\begin{gather}\label{hodocon}
\sum_{m=1}^nt_m \part_m\vec\eta+\vec\eta_0=0,\qquad \sum_{m=1}^nt_m \part_m\vec\varpi+\vec\varpi_0=0,
\end{gather}
where $\part_m=\frac{\part}{\part t_m}$,
\begin{gather*}
\part_m\vec\eta=J\big(\Im \t^t\big)^{-1}\Im\res{\z=\infty}\frac{\z^m\vec\o}{m},\qquad \part_m\vec\varpi=\Re\res{\z=\infty}\frac{\z^m\vec\o}{m}+\Re\t^tH^t\part_m\vec\eta,
\qquad m=1,\dots,L,
\end{gather*}
and $\vec\eta_0, \vec\varpi_0$ can be obtain from \eqref{eta-nu} by setting $f=0$, $\vec\eta=\vec\eta_0$, $\vec\varpi=\vec\varpi_0$ there. The following set of conservation equations
\begin{gather*}
\Im\res{\z=\infty}\le[\part_l\frac{\z^m}{m}-\part_m\frac{\z^l}{l} \ri]\vec\o=
\Im\big(\part_m \t^t\big)H^t\part_l\vec\eta-\Im\big(\part_l \t^t\big)H^t\part_m\vec\eta ,\nonumber\\
\Re\res{\z=\infty}\le[\part_l\frac{\z^m}{m}-\part_m\frac{\z^l}{l} \ri]\vec\o=
\Re\big(\part_m \t^t\big)H^t\part_l\vec\eta-\Re\big(\part_l \t^t\big)H^t\part_m\vec\eta
\end{gather*}
or
\begin{gather*}\label{conserv-fin}
\res{\z=\infty}\le[\part_l\frac{\z^m}{m}-\part_m\frac{\z^l}{l} \ri]\vec\o=
\part_l \t^t\big(\Im \t^t\big)^{-1}\Im\res{\z=\infty}\frac{\z^m\vec\o}{m}-
\part_m \t^t\big(\Im \t^t\big)^{-1}\Im\res{\z=\infty}\frac{\z^l\vec\o}{l},\\
 \qquad m,l=1,\dots,L,
\end{gather*}
follows from \eqref{hodocon}.

\begin{Lemma}\label{lem-dKda}
If $\l=\lambda_j$, $j=0,1,\dots,2L+1$ and the modulation equations~\eqref{Kmod} hold, then
\begin{gather}\label{dKdagen}
\frac{\part K(z)}{\part \l}= \le[\frac{\part \ln |D|}{\part \l}+\frac{1}{2(z-\l)}\ri]K(z).
\end{gather}
In particular, for any $j,m=0,1,\dots,2L+1$
\begin{gather}\label{dKda}
\frac{\part}{\part \lambda_m}K(\lambda_j)=0\quad \text{if} \ \ j\neq m \qquad \text{\rm and} \qquad
\frac{\part}{\part \lambda_j}K(\lambda_j)=\hf K'(z)|_{z=\lambda_j}.
\end{gather}
\end{Lemma}
\begin{proof}
Formula \eqref{dKdagen} is a direct consequence of Theorem \ref{equiv-mod} and \eqref{hK}, whereas \eqref{dKda} follows from the analyticity of $K(z)$ at $z=\lambda_j$ \eqref{K} and \eqref{Kmod}. The analyticity of $K(z)$ at $z=\lambda_0$ follows from the fact that $\frac{u(z)}{R(z)}$ does not have a jump across $\g_{m,0}$ and, thus,
\begin{gather*}
\frac{R(z)}{2\pi i}\oint_{\hat{\mathfrak M}}\frac{u(\z)d\z}{(\z-z)R(\z)}=u(z)+ \frac{R(z)}{2\pi i}\oint_{\widetilde { \mathfrak{M}}}\frac{u(\z)d\z}{(\z-z)R(\z)}.
\end{gather*}
Here $\widetilde {\mathfrak{M}}$ is a small deformation of $\hat{\mathfrak M}$ near $\lambda_0$, such that $\widetilde {\mathfrak{M}}$ crosses $\g_{m,0}$ and does not contain a neighborhood of $\lambda_0$.
\end{proof}

\begin{Corollary}
For any $j=0,1,\dots,2L+1$ and any $m=1,\dots,L$ the motion of $\lambda_j$ with respect to $t_m$ is given by
\begin{gather*}
\part_m (\lambda_j)=-\frac{2\part_mK(\lambda_j)}{K'(\lambda_j)}.
\end{gather*}
Moreover, the system of branchpoints satisfies the so-called Whitham equations
\begin{gather}\label{Whitham}
\part_mK(\lambda_j)\part_l\lambda_j=\part_lK(\lambda_j)\part_m\lambda_j,
\end{gather}
for any $l=1,\dots,L$.
\end{Corollary}

The proof follows directly from Lemma~\ref{lem-dKda}. Equations~\eqref{Whitham} are, in fact, an instance of the well-known Whitham equations for the evolution of branchpoints of the Riemann surface~$\Rscr(\vec t)$, see, for example,~\cite{ET}.

\section[Existence of genus zero $g$ and $h$ functions, breaking curves and continuation principle]{Existence of genus zero $\boldsymbol{g}$ and $\boldsymbol{h}$ functions,\\ breaking curves and continuation principle}\label{sec-high-gen}

The $L=0$ $g$-function constructed in Section \ref{modeqsubsec} for our quartic potential $f(z;t)$ satisf\/ies the equality requirements
(jump conditions) of Section~\ref{eqreqs}, but so far there is no information about the sign requirements of Section~\ref{signreqs}. We f\/irst discuss the sign requirements in the genus zero region (they were established in~\cite{ArnoDu} for $t\in (- 1/12, 0 )$ and later in~\cite{BT3} for any~$t$) Then we introduce the notion of the breaking curve, derive equations for the breaking curve and f\/ind their exact locations. Finally, we show that the sign requirements of Section~\ref{signreqs} are satisf\/ied in the region of genera one and two, that is, at any regular point of $\C^*=\C\setminus\{0\}$.

\subsection{Sign requirements in the genus zero region}\label{sect-exist}

The explicit expression for $h$, given by \eqref{h}, satisf\/ies the RHP in Remark~\ref{rhphremark} and, thus, $g=\hf(h+f+\ell)$ satisf\/ies the equality requirements of Section~\ref{eqreqs}. It will be convenient to make some deformations, see Remark~\ref{rem-analjump}, of the contour $\O$ that are similar to that of~\cite{ArnoDu}. Consider f\/irst $t\in(-\frac{1}{12}, 0)$, that is, $\arg t =-\pi$. The contour $\O$ can f\/irst be deformed into the union of the segment $[-b,b]$ and four rays connecting the endpoints $\pm b$ with inf\/inity, where $b=b(t)$ is given by~(\ref{lambdasym}), see Fig.~\ref{comppic}, left. The two rays emanating from $b$ have directions $\pm \frac{\pi}{4}$ and the two rays emanating
from $-b$ have directions $\pm \frac{3\pi}{4}$. The rays in the left half-plane are oriented towards $z=-b$ whereas the rays in the right half-plane are oriented towards inf\/inity; the segment $[-b,b]$ has the standard (left to right) orientation. As we consider Generic traf\/f\/ic conf\/iguration, see~\eqref{gener-eq},~\eqref{normal-traff}, the traf\/f\/ic on $[-b,b]$ equals $1$, while all the traf\/f\/ics $\varrho_j$ on the four
rays are dif\/ferent from zero. The following Lemma~\ref{lemma-signs}~\cite{BT3} shows that the segment $[-b,b]$ represents the main arc, whereas the rays can be deformed into unbounded complementary arcs.

\begin{Lemma}\label{lemma-signs}
If $t\in(-\frac{1}{12}, 0)$, the four rays of the contour $\O$ can be deformed so that the func\-tion~$\Re h(z)$, where $h(z)$ is given by~\eqref{h}, satisfies the sign requirements of Section~{\rm \ref{signreqs}} along the contour~$\Omega$.
\end{Lemma}

Lemma~\ref{lemma-signs} implies the existence of the genus zero $h$ function (and, thus, the existence of the genus zero $g$-function) for all $t\in(-\frac{1}{12}, 0)$. The following arguments shows that the genus zero region extends from $(-\frac{1}{12}, 0)$ into a region in~$\C^*$.

First observe that if a point $t\in\C^*=\C\setminus\{0\}$ is a regular point of genus $L$, see Remark~\ref{rem-except}, that is, the branchpoints in $\vec\l(t)$ are distinct and the sign requirements of Section~\ref{signreqs} for $h(z;t)$ are satisf\/ied in every interior point of~$\mathfrak M$,~$\mathfrak C$, then there is a neighborhood~$U$ of~$t$ consisting of regular point of genus $L$. The proof of this well known fact is given in the lemma below (see also \cite{BertoBoutroux,TVZ1} etc.).

\begin{Lemma}\label{lem-cont}
The set of regular points is open. This statement is true for a~general polynomial potential $f=f(z;\vec t)$, see~\eqref{potential_n}.
\end{Lemma}

\begin{proof}
Writing modulation equations \eqref{Kmod} as $F(\vec \l;\vec t)=0$, we obtain from Lemma~\ref{lem-dKda} that the Jacobian matrix
\begin{gather}\label{dFd_lam}
\frac{\part F}{\part \vec \l}=\hf \operatorname{diag}\big(K'(\lambda_0),\dots,K'(\lambda_{2L+1})\big).
\end{gather}
If $\vec t_0$ is a regular point then all the branchpoints are f\/inite and distinct. In this case the determinant $K(z)$ has simple zeroes at every branchpoint $\lambda_j$, see~\eqref{modeq} and~\eqref{hK}, so that the Jacobian matrix~\eqref{dFd_lam} is invertible. By the implicit function theorem, we obtain a unique continuation of the branchpoints $\vec\l=\vec\l(\vec t)$ into a neighborhood of the original parameter vector $\vec t_0$. Since $\vec\l(\vec t)$ is a continuous function of $\vec t$ in this neighborhood, then, according to~\eqref{hK}, so is $h(z;\vec t)$. Thus, if this neighborhood is suf\/f\/iciently small, then the sign requirements of Section~\ref{signreqs} for $h$ will be preserved, and we complete the proof.
\end{proof}

\subsection{Breaking curves. Regular and critical breaking points. Symmetries }\label{sect-symm}

Let the $h$-function $h=h(z;t)$ satisfy the sign requirements of Section~\ref{signreqs} except, possibly, at f\/initely many points. Def\/ine $z_0$ as a saddle point of $h(z;t)$ if $h'(z;t)$ has zero at~$z_0$ of order at least one, that is, $h'(z;t)=O(z-z_0)$ as $z \ra z_0$. A point $t\in\C^*$ is called a breaking (non-regular) point, if there is a saddle point $z_0$ of $h(z;t)$ that is either on a main arc or on a complementary arc that cannot be deformed away from $z_0$. In the latter case, $\Re h(z_0,t)=0$ and we say that the complementary arc is pinched at $z_0$. So, if $t_b$ is
a~breaking point, then there exists a~saddle point $z_0\in\O$ satisfying the equations
\begin{gather}\label{br-pt-eq}
h'(z_0;t_b)=0\qquad \text{and} \qquad \Re h(z_0;t_b)=0.
\end{gather}
(In the case $z_0$ is a branchpoint, the f\/irst equation~\eqref{br-pt-eq} should be replaced by the condition: $h'(z;t)$ is of order at least $O(z-z_0)^{\frac 32}$.)

A breaking point $t_b\in\C^*$ is called {\it critical} breaking point if one of the following applies:
\begin{enumerate}\itemsep=0pt
\item[1)] a saddle point $z_0$ of $h(z;t)$ coincides with a branchpoint; in this case we have $h'(z;t_b) = O(z-z_0)^{m}$ with $m\geq \frac 3 2$;
\item[2)] there are at least two (counted with multiplicity) saddle points of $h(z;t)$ on $\O$ that, in the case when a saddle point is on $\mathfrak C$, cannot be deformed away.
\end{enumerate}

For the case $\O=\R$ considered in \cite{DKMVZ}, critical points correspond to special cases of {\it irregular} potentials. A breaking point $t_b$ that is not a critical point is called a regular breaking point. In this case the saddle point $z_0$, corresponding to $t_b$, is simple and does not coincide with any branchpoint. It is called a {\it double point}. This name ref\/lects the fact that a double point can be considered as a double branchpoint, that can move apart and form an extra main or complementary arc of $\O$ in a higher genus region, that is, an extra branchcut of~$\Rscr$.

Considering \eqref{br-pt-eq} as a system $G(u,v,t_j)=0$, $j=1,2$ of three real equations with four real variables $z=u+iv$ and $t=t_1+it_2$, we calculate the Jacobian
\begin{gather}\label{Jac-break}
 \le. \det\le( \frac{\pa G}{\pa(u,v,t_j)}\ri)=i^{j-1}|h''(z;t)| \Re h_{t_j}(z;t)\ri |_{(z,t)=(z_0,t_b)},
\end{gather}
where we can choose $j$ to be either $1$ or $2$. If $t_b$ is a regular breaking point and if $h_t(z_0;t_b)\neq 0$, then, by the implicit function theorem, there is a smooth curve passing through $t_b\in\C$ that consists of breaking points. Such curves are called breaking curves. Thus, we have proved the following lemma.

\begin{Lemma}\label{lem-br-curve-loc}
Any regular breaking point $t_b$ belongs to a locally smooth branch of the breaking curve provided that at least one of the derivatives $h_{t_j}(z_0;t_b)$, $j=1,2$, is not zero.
\end{Lemma}

Lemma \ref{lem-br-curve-loc} can be extended for a critical breaking point $t_b$ corresponding to a situation where there are several double points not coinciding with the branchpoints. In this case each double point, according to Lemma \ref{lem-br-curve-loc}, will def\/ine its own local piece of a breaking curve, passing through~$t_b$. Thus, generically, there would be several breaking curves intersecting each other at~$t_b$. However, under certain symmetries, all these breaking curves could coincide and form just one breaking curve that is smooth at~$t_b$. Examples of such symmetries include the case of the semiclassical focusing NLS~\cite{TVZ1}, where the Schwarz symmetry of $h(z)$ implies that the saddle points of~$h(z)$ on~$\O$ appears in complex conjugated pairs. As it will be shown below, they also include the case of Generic traf\/f\/ic conf\/iguration.

In particular, we will 	prove that $h'(z;t)$ is an odd function (in $z$) and, thus, the saddle points $\pm z_0\neq 0$ appear simultaneously on the contour $\O$ at a critical point $t_b$. We will also show that at least one of the derivatives $h_{t_j}(z_0;t_b)\neq 0$ if $t_b$ is a breaking point and $z_0$ is a~corresponding double point, so that the implicit function theorem guarantees existence of a~smooth local breaking curve for each double point~$\pm z_0$. Because of the symmetry, these (local) breaking curves coincide with each other. Therefore, in the case of symmetric~$h'$ {\it we extend the notion of a regular breaking point~$t_b$ to the cases when there is a pair of symmetric double points} on~$\O$ that generate the same smooth breaking curve passing through~$t_b$.

\begin{Lemma}\label{lem-sym-bp} Under the assumption of Generic traffic configuration, the branchpoints $\vec \l$ are symmetrical for any $t\in\C^*$, that is, $\l$ is a branchpoint if and only if~$-\l$ is a branchpoint.
\end{Lemma}
\begin{proof}
Let $\varrho_0=\varrho_2$ and $\varrho_1=\varrho_3$ in~\eqref{orthog}. Then, since the weight is even, the orthogonal polynomials $\pi_n(z)$ are also symmetrical, and so, the curves of accumulation of zeroes will also be even. It is well known that these curves coincide with the main arcs (see for example~\cite{BertolaMo}). But, within the Generic traf\/f\/ic conf\/iguration, the RHP contour $\O$ does not depend on the traf\/f\/ics~$\vec \varrho$. Thus, in the Generic case with any admissible traf\/f\/ics, the branchpoints are symmetrical.
 \end{proof}

\begin{Corollary}\label{lem-symm}
If the branchpoints $\vec \l=\vec\l(t)$ are symmetrical, the function $h'(z;t)$ is odd and the functions $h_{t_j}(z;t)$ are even, $j=1,2$, $t=t_1+it_2$.
\end{Corollary}
\begin{proof}
If the branchpoints are symmetrical, the corresponding radical $R(z)$, see \eqref{R}, is even for odd $L$ and odd for even $L$. It follows then that the solution $g'(z;t)$ to the scalar RHP~\eqref{jumpassg'}, given by~\eqref{g'gen}, is odd, and so, $h'(z;t)$ is odd, regardless of the genus. Because of the modulation equations, solution of the RHP \eqref{jumph}, \eqref{asshgen} for $h$ commutes with the operator~$\frac d{dt_j}$. Writing solution to the dif\/ferentiated RHP for $h_{t_j}$ in the same form as \eqref{hform}, we can show, in a similar way to~$h'$, that~$h_t$ is even.
\end{proof}

Now, according to \eqref{br-pt-eq}, a breaking point $t_b$ is critical if either
\begin{gather}\label{crit-pt}
\left. \frac{h'(z;t_b)}{R(z)} \right|_{z=\pm z_0}=0\qquad \text{or} \qquad h'(z_0;t_b)=h''(z_0;t_b)=0,
\end{gather}
where $z_0$ is a branchpoint in the former case and an interior point on a main or a complementary arc in the latter case. In the genus zero case, according to~\eqref{lambdasym}, the branchpoints are~$\pm b$. Using equation~\eqref{h'} for~$h'$, we easily calculate that the f\/irst equation~\eqref{crit-pt} yields $t=t_0=-\frac 1{12}$, whereas the second set of equations yields $z_0=0$, $\frac{tb^2}{2}=-1$ and, thus, $t=t_2 = \frac 14$. As we will show below, there are no other critical points for Generic traf\/f\/ic conf\/igurations (the critical point $t_1=\frac{1}{15}$
occurs for dif\/ferent traf\/f\/ic conf\/igurations, see~\cite{BT3}).

It follows then that the arcs of breaking curves begin and end at a critical point ($t_0=-\frac 1{12}$, $t_2=\frac 14$) or at $t=0$ or at inf\/inity. They separate either regions of dif\/ferent genera or regions of the same genus but with dif\/ferent topology of main arcs. In Section~\ref{sect-br-cr-symm} we shall exclude the possibility of a breaking curve forming a smooth closed curve (loop) in the complex $t$-plane while not passing through any critical point or through $t=0$.

\subsection{Continuation principle for Boutroux deformations}\label{sect-cont-princ}

In Section \ref{sec-genusL} we have derived the explicit formula \eqref{hK}, \eqref{K} for the $h$ function in the genus~$L$ region, where $L=1,2,\dots$. For quartic exponential weight considered here ($L\leq 2$) there is no guarantee yet that the function~$h$, given by~\eqref{hK},~\eqref{K}, satisf\/ies the sign requirements of Section~\ref{signreqs}. In Section~\ref{sect-exist}, we showed that the sign requirements are satisf\/ied in the genus zero region. To prove that the correct signs of $\Re h$ persist beyond the genus zero region, we use the continuation principle for Boutroux deformations. Let $\vartheta$ be a smooth bounded curve in the external parameters space (complex $t$-plane in our case). According to \eqref{modeq-g'}, the deformation of the vector of (all the) branchpoints $\vec \l(t)$ along $\vartheta$, governed by the modulation equations, preserves the Boutroux condition (Boutroux deformation). The idea of the continuation principle is that moving along $\vartheta$, we can show that the sign conditions for $\Re h(z;t)$ will be satisf\/ied
for all $t$ along the curve, provided that:
\begin{itemize}\itemsep=0pt
\item the sign conditions for $\Re h(z;t)$ were satisf\/ied at the beginning point of $\vartheta$;
\item the curve $\vartheta$ does not contain any critical point or any singular point ($t=0$), and;
\item every time a breaking curve is crossed (at a regular breaking point), the genus is properly
adjusted.
\end{itemize}

Let $t$ be a regular point in a genus $L$ region, that is, all the branchpoints in $\vec\l(t)$ are distinct and $\Re h(z;t)$ satisf\/ies the strict inequalities from Section~\ref{signreqs} at every point of the contour $\O=\O(t)$ except the branchpoints. Take another regular point $t^*$ in the same open connected component $\mathfrak{G}$ of the genus $L$ region, that is, $t,t^*\in\mathfrak{G}$, and let $\nu\subset \mathfrak{G}$ be a simple smooth contour, connecting~$t$ and~$t^*$. Using Lemma~\ref{lem-cont}, we can continue $\vec\l(t)$ from $t$ to $t^*$ along $\nu$ preserving the condition $\lambda_j(t)\neq \lambda_k(t)$, $j\neq k$. Therefore, we can also continue $h(z;t)$ from $t$ to $t^*$ along~$\nu$. The only possibility for the inequalities for $h(z;t)$ to fail at some point $t\in \nu$ is if $t$ is a breaking point, which is excluded by construction of $\nu$. Thus, we proved the existence of $h$-function satisfying the strict inequalities at any point of~$\mathfrak{G}$.

In view of the above results, it remains only to prove that the strict inequalities can be preserved after a transversal crossing (along $\vartheta$) of the breaking curve at a regular breaking point~$t_b$. Since $L\leq 2$, at any regular breaking point~$t_b$, there can be either one or two symmetrical double points $z_0=0$ or $\pm z_0$ respectively. Because of the symmetry of $h_z(z;t)$, it is suf\/f\/icient to consider only one double point~$z_0$ in any case.

At a regular breaking point $t_b$ the topology of zero level curves of $\Re h(z;t)$ undergoes a change according to one of the following four scenarios:
\begin{itemize}\itemsep=0pt
\item two endpoints of a main arc collide and the arc becomes a double point $z_0$;
\item two endpoints of neighboring main arcs collide and the corresponding bounded complementary arc turns into a double point $z_0$;
\item a complementary arc is pinched by zero level curves of $\Re h(z;t)$ and a new double point $z_0$ appears;
\item a main arc collides with (another) zero level curve of $\Re h(z;t)$ and a new double point $z_0$ appears.
\end{itemize}

The f\/irst two scenarios lead to decrease of the genus whereas the last two lead to its increase.

Let $h^{(n)}(z;t)$ denote the $h$ function (that means that all the strict inequalities of Section~\ref{signreqs} are satisf\/ied) in the genus $n$ region.

\begin{Theorem}\label{theo-h-on-break}
Let $t_b$ be a regular breaking point on a breaking curve, separating regions of genus~$n$ and~$m$, where $n,m\in \N\cup\{0\}$. Then $h_z^{(m)}(z;t_b)\equiv h_z^{(n)}(z;t_b)$ and $h_{zt}^{(m)}(z;t_b)\equiv h_{zt}^{(n)}(z;t_b)$ for all~$z\in\C$.
\end{Theorem}

For the RHP, describing the NLS evolution, and with $|m-n|=2$, the f\/irst identity of the theorem was proved in \cite[Theorem 3.1]{TVZ1}, and the second identity~-- in \cite[Theorem~3.1]{TV1}. In the context of orthogonal polynomials, this theorem was proved in~\cite{BertoBoutroux}. Finally, in the most general case which include both regular breaking points and the critical points, the theorem follows from the continuity of Boutroux-normalized meromorphic dif\/ferentials, proven in~\cite{BT5}.

\begin{Corollary}\label{cor-cont}
In the conditions of Theorem~{\rm \ref{theo-h-on-break}}, $h^{(m)}(z;t_b)\equiv h^{(n)}(z;t_b)$ and $h_{t}^{(m)}(z;t_b)\equiv h_{t}^{(n)}(z;t_b)$ for all $z\in\C$.
\end{Corollary}

We can now return to the continuation principle for Boutroux deformations. Let $t(s)$, $s\in\R$, be a parametrization of the smooth curve $\vartheta$ that is transversal to the breaking curve at a regular breaking point $t_b=t(0)$. Then, by Lemma~\ref{lem-ht-neq-0}, $h_t(z_0,t_b)\neq 0$, where $z(s)$ is the corresponding saddle point and $z_0=z(0)$. Since $\Re[h_t(z_0;t_b)\D t]=0$ when $\D t$ is tangential to the breaking curve at $t_b$, we conclude that $\Re\frac{dh}{ds}=\Re[h_t \frac{dt}{ds}]\neq 0$ at $z_0$, $t_b$.

Let us assume, for example, that $\Re\frac{dh}{ds}(z_0;t_b)<0$, and that the sign conditions for $h(z,t)=h^{(L)}(z,t)$ hold when $t=t(s)$ with $s<0$. Then, as $s\ra 0^-$, either a main arc $\g_{mj}$ collapsed into the double point $z=z_0$, or a main arc collides with another zero level curve of $\Re h(z;t)$ at $z_0$. We now need to prove the sign requirements on the other side of the breaking curve in a small vicinity of $z_0$, that is, when $t=t(s)$ with a small $s>0$.

Consider f\/irst the case of a collapsing simple (without whiskers) main arc $\g_{mj}$. This case can be illustrated by Fig.~\ref{comppic} as transition from the second to the f\/irst (counting from the left) panel, caused by collapse of the two symmetrical small main arcs (on the right and on the left). We will now consider only a vicinity of~$z_0$ and disregard the symmetrical to~$z_0$ double point. If $s<0$ is suf\/f\/iciently close to zero then $\g_{mj}$ is the only main arc in a vicinity of~$z_0$. Moreover, $z_0$ is a saddle point of $\Re h^{(L)}(z,t_b)$ with $\Re h^{(L)}(z_0,t_b)=0$. Let us desingularize the hyperelliptic Riemann surface $\mathfrak R(t_b)=\mathfrak R^{(L)}(t_b)$ of the genus $L$ by removing a pair of branchpoints that collided into $z=z_0$. The genus of the desingularized Riemann surface $\mathfrak R^{(L-1)}(t_b)$ is $L-1$. The corresponding $h^{(L-1)}(z;t_b)\equiv h^{(L)}(z;t_b)$ satisf\/ies the strict sign inequalities everywhere except $z=z_0$. However, since $h_t^{(L-1)}(z;t_b)\equiv h_t^{(L)}(z;t_b)$ and $\Re\frac{dh^{(L)}}{ds}(z_0;t_b)<0$, we conclude that for any small $s>0$ the function $\Re h^{(L-1)}(z;t(s))<0$ in some deleted neighborhood of~$z_0$. Then the strict sign inequalities for $h^{(L-1)}(z;t(s))$, where $s>0$ and small, will be valid everywhere in this neighborhood. So, we proved that $h(z,t(s))=h^{(L)}(z,t(s))$, $s<0$, can be continued by $h(z,t(s))=h^{(L-1)}(z,t(s))$, $s>0$, across the breaking curve while preserving the sign inequalities.

Consider the case when $\g_{mj}$ is not a simple main arc, that is, $\g_{mj}$ is a disappearing ``whisker'', see the center top panel in Fig.~\ref{Generic} and the other panel right below it (see also transition from the third to the second panel (counting from the left) panel on Fig.~\ref{comppic}). In this case we do not change the genus after crossing the breaking curve, so that $h(z,t(s))=h^{(L)}(z,t(s))$ for all $s$ close to zero. By repeating the previous arguments we see that a new complementary arc is formed near $z=z_0$, so that strict inequalities are valid everywhere for small $s>0$. Here we would like to mention one more issue: at the breaking point $t_b$ solution of the modulation equations $F(\vec\l,t)=0$ does not produce a set of distinct branchpoints; in fact, two of them collided at $z=z_0$. Under this circumstances we cannot use Theorem~3.3 from~\cite{TVZ1} to prove existence and continuity of the branchpoints $\vec\l(t(s))$ for $s>0$, since the Jacobian $\frac{\partial F}{\partial \vec\l}=0$ at $t=t_b$. However, after a proper rescaling of $t-t_b$, it was shown that $\vec\l(t(s))$ can be continuously extended for $s>0$ in~\cite[Theorem~6.4]{TVZ1}. This fact also follows from Remark~\ref{rem-cont}.

The remaining case of a main arc colliding with another zero level curve can be considered similarly. Finally, the remaining two cases of dif\/ferent topological change at the breaking curve correspond to $\Re\frac{dh}{ds}(z_0;t_b)>0$. They can be considered in the similar way by just reversing the orientation of the curve $t(s)$, i.e., by replacing $s\mapsto -s$.

Thus, we have essentially proved the following continuation principle: if $t(s)$, $0\leq s\leq 1$, is a smooth curve in the complex $t$-plane that does not contain critical points, and if $h(z, t(0))$ satisf\/ies all the conditions (equations and inequalities) for $ h$-functions, then $h(z, t(0))$ can be continuously (in~$s$) deformed into $h(z, t(1))$, so that for any $s\in [0,1]$, the function $h(z, t(s))$ satisf\/ies all the conditions for $ h$-functions. Since the curve $t(s)$ and the breaking curve are smooth and their intersections are transversal, there can be no more than a f\/inite values of $s\in [0,1]$ for which the inequalities are not strict; for these values of $s$, $t(s)$ are regular breaking points.

The next step in the proof of the sign requirements for $h(z;t)$ for any $t\in\C^*\setminus\{t_0,t_2\}$ is to show that for all $s\in [0,1]$ the set of branchpoints $\l(t(s))$ is bounded. (Note that the dimension of the vector $\l(t(s))$ may vary with $s$). This will be done in Section~\ref{sect-bound}.

\section{Boundedness of the branchpoints}\label{sect-bound}

Continuation principle allows us to prove that if the $g$-function (or $h$-function) exists for some $t=\hat t_b$, then it can be continued to a given nonsingular $t=\hat t_1$ along any simple, compact, piece-wise smooth contour $t(s)$ in~$\C$, where $s\in [0,1]$ and $t(0)=\hat t_b$, $t_1=\hat t(1)$, provided that $t(s)$ does not contain any critical points. However, the arguments of Section~\ref{sect-cont-princ}, as well as general theorem about continuation principle from~\cite{BertoBoutroux}, does not exclude the possibility that some of the branchpoints $\lambda_j(t(s))$ may approach inf\/inity as $s\ra s_*$ for some $s_*\in (0,1)$. In this section we prove the boundedness of the branchpoints and nodes of a nodal hyperelliptic curve $\Rscr$ under the Boutroux evolution, which implies boundedness of the branchpoints for general polynomial potentials.

We start our analysis by considering polynomials $P(z)$ of the form
\begin{gather} \label{11}
P(z) = M^2(z) S(z),\qquad M(z):= t^2 \prod_{j=1}^m (z-\mu_j),\qquad S(z):= R^2(z)=\prod_{j=0}^{2L+1} (z-\lambda_j),
\end{gather}
where all $\lambda_j$ are distinct. We are motivated by the fact that \eqref{11} coincides with~\eqref{y-curve}, where $y(z)=\sqrt{P(z)}$. Let us assume that
\begin{gather}
y(z) = f'(z) + \frac T z + \mathcal O\big(z^{-2}\big), \label{y-ass}
\end{gather}
where $T\in\R$ is a residue of $y(z)$ at inf\/inity and the polynomial potential $f$ has the form
\begin{gather}
f(z) = f(z,\vec t)=\sum_{j=1}^{m+L+2} \frac {t_j}{j} z^j \label{potential}
\end{gather}
with $\vec t=(t_1,\dots, t_{m+L+2})$ being a vector of (complex) coef\/f\/icients of $f$ and $t=t_{m+L+2}$. (The constant term in $f(z)$ does not play any signif\/icant role and hence is conventionally set to zero.)

Note that $y^2(z) =P(z)$ def\/ines a~nodal hyperelliptic curve $\Rscr$ of the (geometric) genus $L$. We further assume that the meromorphic dif\/ferential $y(z) d z$ satisf\/ies the Boutroux condition:
\begin{gather}
\Re \oint_{\gamma} y(z) d z = 0 \label{Boutrouxcond}
\end{gather}
along any loop $\g$ on $\Rscr$.
\begin{Remark}
A simple counting of the parameters shows that the $m+L+2$ complex coef\/f\/icients from $\vec t$, the $T\in \R$ and the $2L+1$ real constraints~\eqref{Boutrouxcond} account for all the $m+2L+3$ complex parameters ($\vec \l,\vec \m$ and $t$) def\/ining a polynomial $P(z)$ of the form~\eqref{11}.
\end{Remark}

\begin{Remark}\label{rem-Bout}
It is clear that $y(z)$ satisf\/ies the jump conditions \eqref{jumpassh'}, where $\g_{m,j}$ are the branchcuts of $\Rscr$, and the asymptotics of~$y$ at inf\/inity is given by~\eqref{y-ass}. Then the Boutroux condition implies that $h(z)=\int_{\lambda_0}^z y(u)du$ satisf\/ies the jump conditions~\eqref{jumph} with some {\it real} constants $\eta_j$, $\o_j$. Note that the jump across~$\g_0$ in~\eqref{jumph} has to be replaced by $2\pi i T$. Here $\g_{c,j}$ is a collection of complementary loops on $\Rscr$ and $\g_0$ is a contour connecting~$\lambda_0$ and $\infty$. The asymptotics of $h(z)$ at inf\/inity is given by
 \begin{gather*}
h(z) = f(z)+\ell +T\ln z + \mathcal O\big(z^{-1}\big) ,\qquad z\ra\infty,
\end{gather*}
where $\ell\in\C$ and $f$ is given by \eqref{potential}. Then various form of modulation equations, obtained in Section~\ref{solRHPsect}, as well as determinantal formulae for $h$, obtained in Section~\ref{sec-high-gen}, are valid for $h(z)=\int_{\lambda_0}^z y(u)du$. In particular, condition~\eqref{Boutrouxcond} together with $T\in\R$ is, of course, equivalent to the integral conditions of~\eqref{modeq-g'}. Then, as was shown in Section~\ref{subsect-modeq}, conditions~\eqref{Boutrouxcond}, \eqref{y-ass} def\/ine the branchpoints $\vec \l=(\lambda_1,\dots,\lambda_{m+L+2})$, the nodes $\vec \m=(\m_1,\dots,\m_m)$ and $T\in\R$.
\end{Remark}

Our goal is to study f\/inite deformations of the parameters $\vec t$ (we consider $T=const$) that preserve the Boutroux condition. Such deformations will be called Boutroux deformations. Let~$\vec t(s)$ be a smooth f\/inite curve in $\C^{m+L+2}$, where $s\in [0,1]$ is the arclength parameter. We assume the residue $T$ to be constant. Consider a Boutroux deformation $P(z,s):=P(z,\vec t(s))$ of the polynomial $P(z)=P(z;\vec t)$. The simple roots $\vec \l(s)$ and the double roots $\vec \m(s)$ now depend continuously on $s$ (as long as they are distinct). As we move along $\vec t$ the genus of the nodal curve $\Rscr=\Rscr(s)$ can change as some simple roots may become double roots and vice versa. We want to prove that the roots $\vec \l(s)$, $\vec \m(s)$ are bounded provided that $\vec t(s)$ is separated from the hyperplane $t=t_{m+L+2} =0$.

The smooth deformation contour $\vec t(s)$ can be then interpreted as a parametrized integral curve of the corresponding vector f\/ield $\vec v(\vec t)$, that is, $\vec {\dot t}(s)= \vec v(\vec t(s))$ and $\|\vec v(\vec t)\|=1$. Thus, we need to show how an arbitrary unitary vector f\/ield on the $\vec t$ induces an inf\/initesimal deformation of the coef\/f\/icients of~$P(z)$, namely, a vector f\/ield on the space of polynomials~$P(z)$ of the form~\eqref{11} that is tangent to the manifold of constraints~\eqref{Boutrouxcond}.

\subsection{A general theorem on ODE's for polynomials}

 Let
 \begin{gather}\label{Pzs}
 P(z,s) = p_0(s) z^n + p_1(s) z^{n-1} +\cdots+ p_{n-1}(s)z +p_n(s),
 \end{gather}
 where the coef\/f\/icients $\vec p(s)=(p_1(s),\dots,p_{n}(s))$ of $P$ satisfy {an} {\it autonomous} system of dif\/ferential equations of the form
 \begin{gather}\label{ODEp}
{\dot {\vec p}(s)} = \vec g(\vec p(s)), \qquad s\in[0,1],
\end{gather}
where $ \vec g(\vec p)=(g_1(\vec p),\dots,g_{n}(\vec p))$. To prove that the roots of the polynomial~$P$ are bounded (in~$s$), it would be suf\/f\/icient to show that $p_0(s)$ is separated from zero and all the coef\/f\/icients $\vec p(s)$ are bounded on $[0,1]$.

\begin{Theorem}\label{noblow}
Let $\vec g(\vec p)$ in \eqref{ODEp} be continuous and suppose that there exists $K>0$ such that the functions $g_j(\vec p)$ satisfy the bound
\begin{gather}\label{g_j-est}
|g_j(\vec p )| \leq K \r^{ j },\qquad j=1,\dots,n,
\end{gather}
 where
\begin{gather}\label{def-rho}
\r = \r(\vec p):= \max_{\ell\geq 1} |p_{\ell}|^{\frac 1 \ell}.
\end{gather}
Then for any solution of the system~\eqref{ODEp} the function $\r(s) := \r (\vec p(s))$ remains bounded for all $s\in [0,1]$.
\end{Theorem}

To prove Theorem~\ref{noblow}, we need to establish the following statements from nonsmooth analysis.

\begin{Lemma}
Let $\vec p(s)\in C^1[0,1]$. Then $\r(s)$ , defined by \eqref{def-rho}, admits Dini's derivative
\begin{gather}\label{Dini}
D^+ \r (s):= \limsup _{\wt s \to s_+} \frac {\r(\wt s )-\r(s)}{\wt s -s}, \qquad s\in [0,1].
\end{gather}
Moreover,
\begin{gather}
D^+\r(s) = \max_{1\leq j\leq n} \le\{\frac {d}{d s} | p_j (s)|^{\frac 1 j}\colon \r(s) = |p_j(s)|^{\frac 1 j}\ri\}.\label{2.6}
\end{gather}
In particular, the derivative is finite at all points where $\r(s)>0$.
\end{Lemma}

\begin{proof} Denote $f_j(s): =|p_j(s)|^\frac 1 j$. Let $s$ be such that $\r(s)>0$. Fix a sequence $s_k\ra s$ from above and let $j_k$ be such that $f_{j_k}(s_k) = \r(s_k)$. Since the set of indices is f\/inite, there is a subsequence of $j_k$ that eventually becomes constant, say $j_0$. Then, along any such subsequence~$\{j_{k_\ell}\}$,
\begin{gather*}
f_{j_0}(s) = \lim_\ell f_{j_{k_\ell}}(s) =\lim_\ell f_{j_{k_\ell}} (s_{k_\ell}) = \lim_\ell \r (s_{k_\ell}) = \r(s).
\end{gather*}
Then
\begin{gather}\label{dpr}
D^+\r(s) \geq \lim_\ell \frac{ \r (s_{k_\ell}) - \r(s)}{s_{k_\ell} -s} = \lim_\ell \frac{ f_{j_{k_\ell}} (s_{k_\ell}) - \r(s)}{s_{k_\ell} -s}
= \lim_\ell \frac{ f_{j_{0}} (s_{k_\ell}) - f_{j_0} (s)}{s_{k_\ell} -s} = f_{j_0}'(s),
\end{gather}
where we used the fact that the subsequence $j_{k_\ell}$ is eventually constant with limiting value $j_0$. Since $f_{j_0}'(s)$ belongs to the set in~\eqref{2.6} we prove the inequality.

Now let us prove ($\leq$). Let the limsup in \eqref{Dini} is attained along the sequence $\{s_k\}$, introduced above. We have
\begin{gather*}
D^+\r(s) = \lim_k \frac{ \r (s_{k}) - \r(s)}{s_{k} -s} = \lim_k \frac{ f_{j_k} (s_{k}) - \r(s)}{s_{k} -s}\leq \lim_k \frac{ f_{j_k} (s_{k}) - f_{j_k}(s)}{s_{k} -s}.
\end{gather*}
The limit set of the sequence $j_k$ necessarily belongs to the set $\{\wt j\colon f_{\wt j}(s)=\r(s)\}$ and hence
\begin{gather*}
 \lim_k \frac{ f_{j_k} (s_{k}) - f_{j_k}(s)}{s_{k} -s} \in \le\{ \frac {d}{d s} f_j(s) \colon f_j(s) = \r(s)\ri\}.
\end{gather*}
Thus the inequality $\leq$ holds. The latter equation and \eqref{dpr} prove \eqref{2.6}.
\end{proof}

Next we state the following well known fact that will be used in the proof of Theorem \ref{noblow}.

\begin{Proposition}[\cite{HagoodThompsonMonthly}]
If $D^+M(s)$ exists everywhere $($finite$)$ in the interval $[a,b]$ then
\begin{gather*}
M(b)-M(a) \leq \ov{\int_a^b} D^+ M(t) \d t \label{2.7}
\end{gather*}
where $\ov {\int_a^b}$ denotes the upper Riemann integral.
\end{Proposition}

\begin{proof}[Proof of Theorem \ref{noblow}] Since ${(\ell\geq 1)}$, we have
\begin{gather}
\frac {d}{d s} |p_{\ell}(s)|^{\frac 1\ell} = \frac 1 \ell |p_{\ell}(s)|^{\frac 1\ell-1} \frac {d}{d s} |p_{\ell}(s)| {\leq}|p_{\ell}(s)|^{\frac 1\ell-1} \frac {d}{d s} |p_{\ell}(s)|\leq |p_{\ell}(s)|^{\frac 1\ell-1} \le|\dot p_{\ell}(s) \ri| \nonumber\\
\hphantom{\frac {d}{d s} |p_{\ell}(s)|^{\frac 1\ell} }{}
= |p_{\ell}(s)|^{\frac 1\ell-1} \le|g _{\ell}(\vec p(s)) \ri|{\leq} K \r^{\ell}|p_{\ell}(s)|^{\frac 1\ell-1},\label{2.14}
\end{gather}
where the assumption of the theorem was used in the latter inequality. By formula~\eqref{2.6} we have that~$D^+\r(s)$ is the maximum of all derivatives amongst the functions $f_\ell=|p_\ell(s)|^{\frac 1 \ell}$ that realize the maximum at~$s$. Thus in~\eqref{2.14}
\begin{gather*}
\frac {d}{d s} |p_{\ell}(s)|^{\frac 1\ell} \leq K \r^{\ell} \r^{1-\ell} = K \r 
\end{gather*}
and, taking into the account \eqref{2.6}, we obtain
\begin{gather*}
D^+ \r(s) \leq K \r(s).
\end{gather*}
We would like to use Gronwall's theorem to complete the proof. In order to use this theorem, we introduce $F(s) = \max \{1, \r(s)\}$. Since $\r(s)$ is the maximum of locally Lipshitz functions $f_j$, $j=1,\dots,n$, (when $\r(s)>1$) then~$F(s)$ is locally Lipshitz as well. It is straightforward to verify that
\begin{gather*}
D^+F(s) = \begin{cases}
0, & \r(s)<1,\\
D^+\r(s), & \r(s)>1,\\
\max\{0, D^+\r(s)\}, & \r(s)=1.
\end{cases}
\end{gather*}
Then, if $\r(s)>1$ then $D^+F(s) = D^+\r(s)\leq K \r(s) \leq K F(s)$; if $\r(s) =1$ then also $D^+F(s) = \max \{0, D^+\r(s)\} \leq \max \{0, K \} = KF(s)$. If $\r(s)<1$ then $D^+F(s)=0\leq K \r(s)\leq K F(s)$. Thus in all cases
\begin{gather*}
D^+F(s) \leq K F(s).
\end{gather*}
Integrating this inequality from $0$ to $s$ and using \eqref{2.7}, we obtain
\begin{gather*}
F(s)-F(0) \leq \ov{\int_0^s} D^+ F(t)d t \leq \ov{\int_0^s} K F(t) d t = {\int_0^s} K F(t) d t ,
\end{gather*}
where in the last equality we have replaced the upper Riemann integral by a regular Riemann integral because~$M(t)$ is continuous. Now Gronwall's theorem yields $F(s) \leq F(0){\rm e}^{Ks}$ for any $s\in[0,1]$. The theorem is proved.
\end{proof}

\begin{Corollary}\label{cor-rootbounded}
If $p_0(s)$ remains separated from $0$ and the remaining coefficients $p_j(s)$, $j=1,\dots,n$, of \eqref{Pzs} satisfy the assumptions of Theorem~{\rm \ref{noblow}}, on~$[0,1]$ then the roots of $P(z;s)$ are uniformly bounded on~$[0,1]$.
\end{Corollary}

\begin{proof}
According to Rouch\'e's theorem, the roots of the polynomial~\eqref{Pzs} are within a circle of the radius
\begin{gather}\label{rat-roots}
N\max_{n\geq \ell\geq 1} \le| \frac {p_\ell(s)}{p_0(s)} \ri|^\frac 1 \ell \leq \frac {N\r(s)}{\min \{|p_0(s)|,1\}},
\end{gather}
where the constant $N$ depends only on the order of the polynomial $P(z;s)$. Now the corollary follows from Theorem~\ref{noblow}.
\end{proof}

 \subsection{Properties of the polynomial \eqref{11}}

 In this section we prove that the ODEs \eqref{ODEp} for the coef\/f\/icients $p_j(s)$, $j=1,\dots,n$ of the polynomial
\begin{gather} \label{22}
P(z,s) = t^2 \prod_{j=1}^m (z-\mu_j(s))^2\prod_{j=0}^{2L+1} (z-\lambda_j(s)),
\end{gather}
see \eqref{11}, where $n=2m+2L+2$, satisfy the conditions~\eqref{g_j-est}. Under the additional assumption that $p_0(s)=t^2(s)$ is separated from zero on $[0,1]$, it would then follow from Corollary~\ref{cor-rootbounded} that the roots $\m_j(s)$, $\lambda_j(s)$ of~\eqref{22} are uniformly bounded on~$[0,1]$.

Let $\part_j=\frac{\part}{\part t_j}$, $j=0,1,\dots,n$. Then $\dot P(z,s)=\nabla_{\vec t}P\cdot \vec v(\vec t)$ and, taking into account $\vec{\dot t}(s)=\vec v(\vec t)$ and \eqref{ODEp}, we obtain
\begin{gather*}
g_\ell(\vec p)=\dot p_\ell(s)=\nabla_{\vec t}p_\ell\cdot \vec v(\vec t).
\end{gather*}
Since $\| \vec v\|=1$, it is clear that if there exist some $ K_{j,\ell}>0$ such that
\begin{gather} \label{p_j}
|\part_jp_\ell |<K_{j,\ell}\r^\ell \qquad \text{for all} \ \ j=1,\dots, m+L+2, \quad \ell=1,\dots,n,
\end{gather}
then the condition \eqref{g_j-est} in Theorem \ref{noblow} will hold for all~$\ell$. Therefore, we need to prove~\eqref{p_j} for some f\/ixed~$j$ and all $\ell=1,\dots,n$. According to~\eqref{11}, we have
\begin{gather}\label{defQ}
y_{j}=\frac {P_{j}}{2\sqrt{P}}=\frac QR,
\end{gather}
where the subindex $j$ is used to denote $\frac{\part}{\part t_j}$, so that
\begin{gather}\label{Pt}
P_{j}=2MQ.
\end{gather}
Expressions \eqref{defQ}, \eqref{Pt} together with \eqref{potential}, \eqref{y-ass} imply that $Q$ is a polynomial of degree $L+j$, and the condition~\eqref{Boutrouxcond} in dif\/ferentiated form implies that the dif\/ferential $y_t {\rm d} z = \frac{Q(z) {\rm d } z}{R(z)}$~\eqref{defQ} is a second kind dif\/ferential on the Riemann surface of $R(z)$ uniquely determined by the condition that $\frac Q R= z^j+ \mathcal O(z^{-2})$ for $|z|\to \infty$ and by the requirement that it is an {\it imaginary normalized}, namely, all integrals on closed paths are purely imaginary. Then the results of~\cite{BT5} apply and we have the following theorem.

\begin{Theorem}[\cite{BT5}]\label{theo-cont-Q} The polynomial $Q(z;\vec \l)$ is continuous in $\vec \l$.
\end{Theorem}

\begin{Lemma}\label{lemmaQ0}
The polynomial $Q(z)=Q(z;\vec \l)$ satisfy the homogeneity
\begin{gather*}
 Q (z;\vec \l) =\nu^{L+j} Q \big(\nu^{-1}z;\nu^{-1}\vec \l\big) ,\qquad \nu\in \R_+.
\end{gather*}
\end{Lemma}

\begin{proof} The Boutroux condition \eqref{Boutrouxcond} guarantees the validity of Corollary~\ref{dhdxt=0}. Using this corollary, \eqref{hK}, \eqref{K} and
the fact that $y(z)=h'(z)$, see \eqref{y-curve}, we obtain
\begin{gather*}
y_{ j }(z)=h'_{ j }(z)=\frac{1}{R(z)}\le[\frac{\big(R^2(z)\big)'K_{ j }(z)+R^2(z)K'_{ j }(z)}{2D}\ri],
\end{gather*}
so that, according to \eqref{defQ}
\begin{gather}\label{QK}
Q(z;\vec \l)=\frac{\big(R^2(z;\vec\l)\big)'K_{ j }(z;\vec\l)+R^2(z;\vec\l)K'_{ j } (z;\vec\l)}{2D(\vec\l)},
\end{gather}
where $()'$ denotes $\frac d{dz}$. Taking into account \eqref{M}, a straightforward calculation shows that
\begin{gather}
D\left(\frac{\vec \l}{\nu}\right)=\nu^{(L+1)L}D(\vec\l), \qquad R\left(\frac z\nu;\frac{\vec \l}{\nu}\right)=\nu^{-(L+1)}R(z;\vec\l),\nonumber\\
\int\frac{\z^kd\z}{R(\z, \frac {\vec \l}{\nu})}=\nu^{L-k}\int\frac{\z^kd\z}{R(\z,\vec \l)}.\label{calc-hom}
\end{gather}

Note that, according to Corollary \ref{dhdxt=0},
\begin{gather}\label{K_t}
K_{{ j}}(z)= \frac{1}{2\pi i j}\!\left| \begin{array}{@{}c@{\,}c@{\,}c@{\,}c@{\,}c@{\,}c@{\,}c@{\,}c@{\,}c@{\,}c@{\,}c@{\,}c@{\,}}
 \oint_{\gt_{m,1}}\frac{ d\z}{R(\z)} &
\cdots & \oint_{\gt_{m,1}}\frac{\z^{L-1} d\z}{R(\z)} & \overline{ \oint_{\gt_{m,1}}\frac{ d\z}{R(\z)}}
& \cdots & \overline{ \oint_{\gt_{m,1}}\frac{\z^{L-1} d\z}{R(\z)}}& \oint_{\gt_{m,1}}\frac{d\z}{(\z-z)R(\z)}\\
\cdots &\cdots & \cdots& \cdots & \cdots & \cdots & \cdots\\
 \oint_{\gt_{m,L}}\frac{ d\z}{R(\z)} &
\cdots & \oint_{\gt_{m,L}}\frac{\z^{L-1} d\z}{R(\z)} & \overline{ \oint_{\gt_{m,L}}\frac{ d\z}{R(\z)}}
& \cdots & \overline{ \oint_{\gt_{m,L}}\frac{\z^{L-1} d\z}{R(\z)}}
 &\oint_{\gt_{m,L}}\frac{d\z}{(\z-z)R(\z)}\\
 \oint_{\gt_{c,1}}\frac{ d\z}{R(\z)} &
\cdots & \oint_{\gt_{c,1}}\frac{\z^{L-1} d\z}{R(\z)} & \overline{ \oint_{\gt_{c,1}}\frac{ d\z}{R(\z)}}
& \cdots & \overline{ \oint_{\gt_{c,1}}\frac{\z^{L-1} d\z}{R(\z)}}
 &\oint_{\gt_{c,1}}\frac{d\z}{(\z-z)R(\z)}\\
\cdots &\cdots & \cdots& \cdots & \cdots & \cdots & \cdots \\
 \oint_{\gt_{c,L}}\frac{ d\z}{R(\z)} &
\cdots & \oint_{\gt_{c,L}}\frac{\z^{L-1} d\z}{R(\z)} & \overline{ \oint_{\gt_{c,L}}\frac{ d\z}{R(\z)}}
& \cdots & \overline{ \oint_{\gt_{c,L}}\frac{\z^{L-1} d\z}{R(\z)}}
&\oint_{\gt_{c,L}}\frac{d\z}{(\z-z)R(\z)}\\
\oint_{\hat{\mathfrak M}}\frac{\z^j d\z}{R(\z)} & \cdots
& \oint_{\hat{\mathfrak M}}\frac{\z^{j+L-1} d\z}{R(\z)} &
 \overline{ \oint_{\hat{\mathfrak M}}\frac{\z^j d\z}{R(\z)}} & \cdots
& \overline{ \oint_{\hat{\mathfrak M}}\frac{\z^{j+L-1} d\z}{R(\z)}} &
\oint_{\hat{\mathfrak M}}\frac{\z^j d\z}{(\z-z)R(\z)}
\end{array}\right|.\!\!\!
\end{gather}
In light of \eqref{calc-hom} and \eqref{K_t}, we obtain
\begin{gather}\label{calc-Kt}
K_{ j}\left(\frac z\nu;\frac{\vec \l}{\nu}\right)=\nu^{(L+1)^2-j}K_{ j}(z,\vec \l), \qquad K'_{ j}\left(\frac z\nu;\frac{\vec \alpha}{\nu}\right)=\nu^{(L+1)^2-j+1}K'_{j}(z,\vec \l).
\end{gather}
 Now the statement of the lemma follows from \eqref{QK}, \eqref{calc-hom} and \eqref{calc-Kt}.
\end{proof}

Let $\vec\a=(\vec\l,\vec\m)$. According to \eqref{Pt}, (see \eqref{11} for the def\/inition of~$M$)
\begin{gather}\label{PG}
\part_jP(z;\vec\a)=2M(z;\vec\mu)Q(z;\vec\l)=G(z;\vec\a).
\end{gather}
Then, as an immediate consequence of Lemma~\ref{lemmaQ0}, we have
\begin{gather*}
 G (z;\vec \a) =\nu^{m+L+j} G \big (\nu^{-1}z;\nu^{-1}\vec \a\big) ,\qquad \nu\in \R_+.
\end{gather*}
Equation \eqref{PG} implies $(p_\ell)_t=G_\ell$, where
\begin{gather}\label{Gexp}
 G \le( z; \vec \a \ri) = \sum_{k=0}^{L+m+j} G_{k}(\vec \a) z^{L+m+j-k}.
\end{gather}
Thus,
\begin{gather}\label{Gellhom}
 G_\ell \big(\nu^{-1}\vec \a\big) =\nu^{-\ell} G_\ell \le (\vec \a\ri).
\end{gather}

\begin{Remark}\label{rem-cont} According to Theorem~\ref{theo-cont-Q}, the polynomial $Q(z;\vec \l)$ is continuous even when any number of branchpoints in $\vec \l$ coincide. Therefore $G_\ell(\vec \a)$ def\/ined in \eqref{PG} is continuous with respect to $\vec \alpha$. That also mean that $G_\ell(\vec \a)$ is continuous in the symmetric polynomials $p_\ell$ of the entries of vector $\vec \a$. Thus, the right hand side of equation~\eqref{PG}, considered as a~nonlinear autonomous system of ODEs, is continuous in the $p_\ell$'s. This right hand side is Lipshitz when the branchpoints in~$\vec\l$ are distinct, so the $t$-evolution of~$\vec \a$ is def\/ined uniquely. On the discriminant locus (where two or more~$\alpha_j$'s coincide) the right hand side of equation \eqref{PG} is still continuous but not Lipshitz-continuous; in fact, the uniqueness of solutions of~\eqref{PG} fails there (but the existence still holds).
This non-uniqueness of the solutions for initial data with coinciding $\alpha$'s (i.e., on the breaking curves) is important because solutions of the continuation equations~\eqref{PG} of dif\/ferent genera coincide at those points and hence it allows us to select the appropriate solution that also satisf\/ies the sign inequality requirements.
\end{Remark}

\begin{Lemma}\label{monicnoblow}
Conditions \eqref{p_j} are satisfied.
\end{Lemma}
\begin{proof}
We prove this statement for an arbitrary f\/ixed $j=1,\dots,m+L+2$. According to \eqref{PG}, \eqref{Gexp}, $\part_j p_\ell=G_\ell(\vec \a)$ for any $\ell=0,1,\dots,n$. Let $\nu=\r$, where $\r$ is def\/ined by \eqref{def-rho}. Then, all the roots of the polynomial $P$, that is, all the components of~$\vec \a=(\vec\l,\vec\m)$, are within the a circle $\wh K\r$, where $\wh K=\frac N{\min\{|p_0|, 1\}}$, see~\eqref{rat-roots}. Since $G_\ell(\vec \a)$ is a continuous function of $\vec\a$, it is bounded on the set $\max_k|\a_k|\leq \wh K$ by some $K_{j\ell}>0$. Then, according to~\eqref{Gellhom},
\begin{gather*}
|\part_jp_\ell|=|G_\ell(\vec\a)|\leq K_{j,\ell}\r^\ell,
\end{gather*}
so that conditions \eqref{p_j} are satisf\/ied.
\end{proof}

Thus, we have proved the following theorem.

\begin{Theorem}\label{theo-bound}\looseness= 1 Suppose that the polynomial $P(z)P=P(z;\vec t)$, given by~\eqref{11},~\eqref{y-ass}, is such that the Boutroux condition~\eqref{Boutrouxcond} is satisfied and $t_{M+L+2}\neq 0$. Let $\vec t(s)$, $s\in[0,1]$, be a smooth deformation of $\vec t$, with $\vec t=\vec t(0)$ preserving the Boutroux condition~\eqref{Boutrouxcond} and such that $t_{M+L+2}(s)$ remains separated from zero. Then the roots $\vec \l(s)$, $\vec\m(s)$ are uniformly bounded for $ s \in [0,1]$.
\end{Theorem}

As an immediate consequence of Theorem~\ref{theo-bound} we obtain the following two corollaries.
\begin{Corollary}
Let $t(s)$, $s\in[0,1]$, be simple, compact, separated from zero piece-wise smooth contour, connecting points $t(0)=\hat t_0$ and $t_1=\hat t(1)$ of $\C$. If $f(z,t)$ is given by~\eqref{orthog} and $T=2$, then the roots of the polynomial $P(z,t(s))$, given by \eqref{11}, \eqref{y-ass}, are uniformly bounded on $s\in[0,1]$, provided that the Boutroux condition~\eqref{Boutrouxcond} holds for all $s\in[0,1]$.
\end{Corollary}

\begin{Corollary}\label{cor-g-exist}
For any $t\in\C^*$ there exists the $g$ function $g(z;t)$, satisfying all the sign requirements of Section~{\rm \ref{signreqs}}. The corresponding hyperelliptic surface $\Rscr(t)$ has genus $L=0,1,2$. The function $h$ is given by \eqref{h} for $L=0$ and by \eqref{hK} $L>0$.
\end{Corollary}

The explicit expression for $h$ with $L=1$ is given in \eqref{h1}.

\section[Breaking curves for quartic potential $f(z;t)$]{Breaking curves for quartic potential $\boldsymbol{f(z;t)}$}\label{sect-br-cr-symm}

In this section, we prove existence of the breaking curve and conf\/irm their topology, shown on Fig.~\ref{Generic}. We start with the following lemma.

\begin{Lemma}\label{lem-ht-neq-0}
If $t_b$ is a regular breaking point and $\pm z_0$ are the corresponding saddle points on $\O$ then $h_t(\pm z_0;t_b)\neq 0$.
\end{Lemma}

\begin{proof} Since the maximal genus for a quartic exponential potential cannot exceed two, see equation \eqref{max-gen},
the genus can be only zero or one along
any breaking curve. Comparing \eqref{h'} with \eqref{hthx}, we see that in the case of genus zero,
the conditions
\begin{gather}\label{bcrv-gen}
h_t(z;t)=0\qquad {\rm and} \qquad h_z(z;t)=0\qquad {\rm while}\quad z\neq \pm b
\end{gather}
imply that $t=\frac 14$, which is the critical breaking point $t_2$.

In the remaining case of genus one with symmetrical branchpoints, the functions $h_z$, $h_t$ are given by
\begin{gather}\label{brcrv-1}
h_z(z;t)=-tz\sqrt{\big(z^2-\lambda_1^2\big)\big(z^2-\lambda_2^2\big)},\qquad\!\! h_t(z;t)=-\frac 14 \big(z^2+K\big)\sqrt{\big(z^2-\lambda_1^2\big)\big(z^2-\lambda_2^2\big)},\!\!\!
\end{gather}
where $\pm\lambda_1$, $\pm\lambda_2$ are the branchpoints and $K$ is a constant in $z$. The $h_z$ expression follows from~\eqref{y-curve}. To derive the second expression, we observe that, according to Corollary~\ref{dhdxt=0}, $h_t(z;t)=\frac{R(z)}{|D|}K_t(z)$. In the case of genus one with symmetrical branchpoints,
$R(z)$ is an even function. Since the residue at~$\z=\infty$ of any even function is zero, the last row of the determinant~$K_t$ (see~\eqref{K_t}) contains all zero entries except of the diagonal (the last) entry. Thus, $h_t(z)=-\frac{R(z)}{8\pi i}\oint_{\hat{\mathfrak M}}\frac{\z^4d\z}{(\z-z)R(\z)}$, which implies~\eqref{brcrv-1}. Using the large $z$ asymptotics~\eqref{jumpassh'} and~\eqref{assht} for $h_z$ and~$h_t$ respectively, we obtain $\lambda_1^2+\lambda_2^2= -\frac 2t$ and $K=-\frac 1t$. Then, conditions~\eqref{bcrv-gen}, where $z\neq \pm b$ is replaced by $z\neq \pm \lambda_j$, $j=1,2$, are not compatible. Thus, we have shown that for the Generic traf\/f\/ic conf\/iguration zeroes of $h_t$ cannot coincide with any regular breaking point.
\end{proof}

Let $t_b$ be an arbitrary regular breaking point. Then, according to Lemmas~\ref{lem-symm} and~\ref{lem-ht-neq-0}, the breaking curve passing through $t_b$ is a locally smooth curve. Thus, we obtain the following corollary.

\begin{Corollary}\label{cor-bc} Breaking curves are smooth simple curves consisting of regular breaking points $($except, possibly, the endpoints$)$. They do not intersect each other except, possibly, at the critical points $t_0=-\frac 1{12}$, $t_2=\frac 14$, at the singularity $t=0$ and at infinity. They can originate and end only at the above mentioned points.
\end{Corollary}

\subsection{Existence of breaking curves}\label{sect-ex-br-crv}

The continuation principle shows that the existence of the $g$-function (or $h$-function) on the interval $(-\frac 1{12},0)$ implies the existence of the $g$ function for any $t\in\C^*$. However, it does not describe the breaking curves or even prove their existence. The collection of the breaking for the Generic traf\/f\/ic conf\/iguration, the so-called ``asymptotic phase portrait'', shown on Fig.~\ref{Generic}, was calculated numerically in~\cite{BT3}.

Let us brief\/ly characterize the breaking curves on Fig.~\ref{Generic}. The curves connecting $t_0=-\frac 1{12}$ and $t=0$ are the boundaries of the genus zero region. They are described by \eqref{br-pt-eq} with $h$ given by \eqref{h}, where $b=b(t)$ is given by~\eqref{branchpeq} with the choice of the negative sign. The next breaking curves, connecting $t_0$ with $t_2=\frac 14$ are given by the same equations, except that now $b=b(t)$ is given by~\eqref{branchpeq} with the choice of the positive sign. Two dif\/ferent values of $b(t)$ given by \eqref{branchpeq} coincide at the critical point~$t_0$. Moreover, with the positive choice of sign in~\eqref{branchpeq} both equations~\eqref{br-pt-eq} with $h$ given by~\eqref{h} are satisf\/ied at $t=\frac 14$ and $z=0$.

To derive equations of the genus zero breaking curves we note that, according to \eqref{h'}, the two saddle points $\pm z_0$ are def\/ined by
\begin{gather}\label{sad_0}
z^2_0=-\frac 1t-\frac{b^2}{2}.
\end{gather}
Thus, the values of $h$, $h_t$, see \eqref{h}, \eqref{hthx} respectively, at either of the two saddle points $\pm z_0=\pm z_0(t)$ are
\begin{gather}\label{hz_0}
h(z_0;t)=\ln\frac{2z_0^2-b^2+2z_0\sqrt{z_0^2-b^2}}{b^2}-\frac {z_0}4 \sqrt{z_0^2-b^2}, \qquad h_t(z_0;t)=\frac {z_0}{4t} \sqrt{z_0^2-b^2}.
\end{gather}
Introducing $u=\sqrt{1+12t}=2\sqrt 3\sqrt{t-t_0}$, where $\arg u=0$ when $t\in(-\frac 1{12},0)$, and using \eqref{branchpeq} and \eqref{sad_0},
we obtain
\begin{gather}\label{hz_0u}
h(z_0(t);t)=\ln\frac{1\pm 2u+\sqrt 3\sqrt{u^2\pm 2u}}{1\mp u}-\frac{\sqrt{u^2\pm 2u}}{4\sqrt 3 t}, \qquad h_t(z_0(t);t)=\frac {\sqrt{u^2\pm 2u}}{4\sqrt 3 t^2}.
\end{gather}
Thus, the equations of genus zero breaking curves are
\begin{gather}\label{brc0}
\ln\le|\frac{1\pm 2\sqrt{1+12t}+\sqrt 3\sqrt{1+12t \pm 2\sqrt{1+12t}}}{1\mp \sqrt{1+12t}}\ri|=\Re \frac{\sqrt{1+12t\pm 2\sqrt{1+12t}}}{4\sqrt 3 t}.
\end{gather}

Let us now consider the local structure of the breaking curves near $t=t_0$. It would be convenient to introduce a new ``hyperbolic'' variable~$p$ by
\begin{gather*}
\cosh p =\frac {z_0}b.
\end{gather*}
Then
\begin{gather*}
\sinh p =\sqrt{\frac{z_0^2}{b^2}-1},
\end{gather*}
and, according to \eqref{hz_0}, equation $\Re h(z_0;t)=0$ for genus zero breaking curves becomes
\begin{gather*}\label{brc0-p0}
\Re \le[2p- \frac{b^2 \sinh 2p}{8}\ri]=0.
\end{gather*}
Direct calculations yields
\begin{gather}\label{bu-p}
\sinh^2p=\frac{\pm 3u}{2(1\mp u)},\qquad b^2=\frac 8{1\pm u}=8\frac{2\cosh^2p+1}{4\cosh^2p-1}=8\frac{\cosh 2p+2}{2\cosh 2p+1},
\end{gather}
which transforms \eqref{brc0-p0} into
\begin{gather}\label{brc0-p}
\Re [2p]=\Re \frac{ \sinh 2p(\cosh 2p+2)}{2\cosh 2p+1}.
\end{gather}
The right hand side of \eqref{brc0-p} has Taylor expansion $2p-\frac{8}{45}p^5+\cdots$ at $p=0$. Thus, for small~$p$, equation~\eqref{brc0-p} becomes
\begin{gather*}
\Re p^5 \approx 0.
\end{gather*}
Taking into account \eqref{bu-p}, we obtain that the breaking curves~\eqref{brc0} with upper signs approach $t_0=-\frac 1{12}$ at angles $\pm \frac{2\pi}{5}$, and the breaking curves~\eqref{brc0} with lower signs approach $t_0=-\frac 1{12}$ at angles $\pm \frac{4\pi}{5}$.

In order to prove the transition from genus 0 to genus 2 along the breaking curve \eqref{brc0} with the upper sign, it is suf\/f\/icient to prove that
\begin{gather}\label{tran_0-2}
\Re \le[h_t(z_0(t);t)\frac{dt}{ds}\ri]>0,
\end{gather}
where $\frac{dt}{ds}$ is the direction of the crossing from the genus~0 to genus~2 region at a regular breaking point $t=t_b$. Indeed, the positive sign of the inequality ref\/lects the fact that a complementary arc is pinched at a double point and, thus, a new main arc is born.

Let us, for simplicity, take $t=t_b$ suf\/f\/iciently close to $t_0=-\frac 1{12}$. Then $|\arg (t_b-t_0)|\leq \pi$ and we know $t_b\not\in\R$. So, $u$ is small and $\Re \sqrt u >0, \Im \sqrt u \neq 0$. Then the logarithmic term in~\eqref{brc0}, and, thus, the right hand side of~\eqref{brc0} are positive. But, according to~\eqref{hz_0u}, $\Re h_t(z_0(t);t)= \Re \le[t^{-1} \frac{\sqrt{u^2\pm 2u}}{4\sqrt 3 t}\ri] <0$ since $\arg t \approx \pi$. We cross the breaking curve moving along the line $\Im t=\Im t_b$ in the negative direction, so $\frac{dt}{ds}<0$. Thus, we have proved inequality~\eqref{tran_0-2} for~$t$ close to $t_0=-\frac 1{12}$. Since, according to Lemma~\ref{lem-ht-neq-0}, $h_t(z_0(t);t)\neq 0$ on the breaking curve, the inequality~\eqref{tran_0-2} will hold on the breaking curve until we reach the singularity $t=0$. Hence we have proved the transition from genus zero to genus two across the breaking curve~\eqref{brc0} with the upper choice of the sign.

The second pair of breaking curves, emanating from $t_0=-\frac 1{12}$, is given by~\eqref{brc0} with the lower sign. They end at $t_2=\frac 14$. They separate two regions of genus 2 with dif\/ferent topology of the main arcs, see Fig.~\ref{Generic}, two middle panels at the bottom and the panel in the middle (lower half plane). As we see, bounded complementary arcs shrink to points as we cross this breaking curve, so that the inequality~\eqref{tran_0-2} should hold in this case as well. The proof is identical to the previous arguments once we notice that $\Re \sqrt{-u}>0$. We have f\/inally proved the following lemma.

\begin{Lemma}\label{lem-brc0}
Four breaking curves, emanating from $t_0=-\frac 1{12}$, separate regions of genera zero and two, or two regions of genus two with different topology of main arcs, as shown on Fig.~{\rm \ref{Generic}}. They are given by equation~\eqref{brc0}.
\end{Lemma}

The genus one breaking curves consist of the ray $\Re z>\frac 14$ of $\R$ as well as two bounded curves, symmetrical with respect to~$\R$, that connect $t_2=\frac 14$ with $t=0$. Let us f\/irst consider this case. Expanding the expression~\eqref{brcrv-1} for~$h'$ at $z=\infty$ and taking into the account the asympto\-tics~\eqref{jumpassh'} of~$h'$, we obtain
\begin{gather*}
h'(z)=-tz^3-z+\frac 2z +O(z_2)=-tz^3\left[1-\frac{\lambda_1^2+\lambda_2^2}{2z^2}-\frac{\big(\lambda_1^2-\lambda_2^2\big)^2}{8z^4}+O\big(z^{-6}\big)\ri],
\end{gather*}
which yields
\begin{gather}\label{bpseq_g=1}
\lambda_1^2+\lambda_2^2=-\frac 2t,\qquad \lambda_1^2-\lambda_2^2=\frac{4}{\sqrt{t}}\qquad {\rm or}\qquad \lambda_1^2=\frac{2}{\sqrt{t}}-\frac 1t,\qquad \lambda_2^2=-\frac{2}{\sqrt{t}}-\frac 1t.
\end{gather}
It follows that the branchpoints are situated on the imaginary axis when $t\in(0,\frac 14)$ and on the cross ($\pm \lambda_1\in\R,~\pm\lambda_2\in i\R$)
when $t>\frac 14$, see Fig.~\ref{Generic}, the panel in lower left corner and the panels in the middle on the right side).

To obtain the equations of the genus one breaking curve we need to compute $h$. Substituting~\eqref{bpseq_g=1} into~\eqref{brcrv-1}, we obtain
\begin{gather*}
h'(z;t)=-tz\sqrt{\le(z^2+\frac 1t\ri)^2-\frac 4t},
\end{gather*}
so that
\begin{gather}\label{h1}
h(z;t)= \ln\frac{tz^2+1+\sqrt{(tz^2+1)^2-4t}}{-2\sqrt{t}}- \frac{tz^2+1}{4t}\sqrt{\big(tz^2+1\big)^2- 4t},
\end{gather}
where we integrated from the branchpoint $-\lambda_2$ (because of the Boutroux condition, $\Re h$ does not depend on the choice of a particular branchpoint). Then
\begin{gather}\label{h1(0)}
h(0;t)=\ln\frac{1+\sqrt{1-4t}}{-2\sqrt t}-\frac{\sqrt{1-4t}}{4t},
\end{gather}
so that the equation of the breaking curves, connecting $t_2=\frac 14$ with $t=0$ is
\begin{gather}\label{brc1}
\ln\le|\frac{1+\sqrt{1-4t}}{-2\sqrt t}\ri|=\Re\frac{\sqrt{1-4t}}{4t}.
\end{gather}
Moreover, it is clear that
\begin{gather}\label{on0,qt}
\Re h(0;t)<0\qquad {\rm when} \quad t\in\left(0,\frac 14\right)
\end{gather}
and $\Re h(0;t)=0$ when $t>\frac 14$. Thus, the ray $t>\frac 14$ is a breaking curve whereas the interval $(0,\frac 14)$ is not.

As one crosses the breaking curve \eqref{brc1} from the genus two region (Fig.~\ref{Generic}, the upper panel in the center), to the genus one region (Fig.~\ref{Generic}, the lower left panel), the central main arc in the f\/irst panel shrinks into a double point (on the breaking curve), which disappear as one gets into the region zero region. To prove this scenario, it is suf\/f\/icient to prove the inequality along the breaking curve~\eqref{brc1},
\begin{gather}\label{tran_2-1}
\Re \le[h_t(0;t)\frac{dt}{ds}\ri]<0.
\end{gather}
As it was observed above, it is suf\/f\/icient to prove \eqref{tran_2-1} at one point of the breaking curve. Let us take this point, $t_b=\r e^{i\phi}$,
suf\/f\/iciently close to $t=0$. Then equation \eqref{brc1} implies
\begin{gather}\label{phi-ineq}
\cos\phi\approx -\r\ln\r \qquad {\rm or} \qquad \phi \ra \frac{\pi}{2}-0
\end{gather}
as $\r\ra 0$. So, according to the second equation \eqref{brcrv-1} and \eqref{phi-ineq}, we have
\begin{gather*}
\Re \le[h_t(0;t_b)\ri]=\Re \frac{\sqrt{1-4t_b}}{4t_b^2} \approx \frac{\cos 2\phi}{4\r^2}<0.
\end{gather*}
Finally, $\frac{dt}{ds}>0$ when we cross the breaking curve~\eqref{brc1} along the line $\Im t=\Im t_b$ with a suf\/f\/iciently small $\Im t_b>0$. Thus, we have proved the following lemma.

\begin{Lemma}\label{lem-brc1} The breaking curves connecting $t_2=\frac 1{4}$ with $t=0$, separate regions of genera one $($inside the curves$)$ and two $($outside the curves$)$, see Fig.~{\rm \ref{Generic}}. They are given by equation~\eqref{brc1}. The breaking curve $\Im t=0$, $\Re t>\qt$ is surrounded by the genus two region.
\end{Lemma}

All the breaking curves, shown on Fig.~\ref{Generic}, emanate from either critical point $t_0=-\frac 1{12}$ or $t_2=\frac 14$. The former are genus zero breaking curves whereas the latter are genus one breaking curves (except those, that connect $t_0$ and $t_2$, which are genus zero). These breaking curves were described in Lemmas~\ref{lem-brc0} and~\ref{lem-brc1}. To complete the proof of the asymptotic phase portrait of Fig.~\ref{Generic}, we still need to show that there are no more breaking curves, that is, there are no breaking curves that do not pass through a critical point, or $t=0$ or~$\infty$. According to Corollary~\ref{cor-bc}, any such curve must be a simple smooth closed loop,~$\nu$, that does not intersect any other breaking curve. We rule out existence of such $\nu$ by observing that, according to~\eqref{br-pt-eq}, each breaking curve is zero level curve of $\Re h(z_0(t),t)=0$, where $z_0(t)$ is the double point. Then, it remains to show that $\Re h(z_0(t),t)$ is a harmonic function of~$t$.

We start with a genus one level curve, where $z_0(t)\equiv 0$ and so $h(z_0(t),t)=h(0,t)$ is given by~\eqref{h1(0)}. This is a~harmonic function with a~singularity at $t=0$ and the branchcut~$\R^+$. Thus, if $\Re h(0,t)=0$ along bounded $\nu$, then $\nu\cap\R^+\neq\varnothing$. But $\nu \cap (\frac 14,+\infty)=\varnothing$ since $(\frac 14,+\infty)$ is a~breaking curve. Thus, $\nu\cap (0,\frac 14)\neq\varnothing$, which contradicts~\eqref{on0,qt}. Finally, if $\nu$ is an unbounded loop that does not intersect $\R^+$, then $\nu$ originates and ends at inf\/inity and $\Re h(0,t)$ is harmonic inside $\nu$. Since $\Re h(0,t)=O(t^{-1})$ as $t\ra\infty$, we conclude that such loop cannot exist.

Consider now a genus zero level curve given by the f\/irst equation \eqref{hz_0u}. The function $\psi(t)=\Re h^0(z_0(t),t)$, given by~\eqref{hz_0} (see also \eqref{hz_0u}) in the complex $t$ plane, has singularities at $\{t_0,t_2, 0,\infty\}$, with branchcut $\R\setminus (-\frac{1}{12},0)$. Thus, if $\nu$ is bounded, it must intersect $(-\infty, -\frac{1}{12})$ or $(0,\frac 14)$. In order to prove that this is not the case, we calculate
\begin{gather*}
\psi'(t)=\Re h_t(z_0(t),t)+ \Re [h'(z_0(t),t)z_0'(t)]= \Re h_t(z_0(t),t),
\end{gather*}
since $h'(z_0(t),t)= 0$ for every $t$ and $ z_0'(t)$ is bounded for $t\not\in \{t_0,t_2, 0,\infty\}$.

Let us f\/irst consider the interval $(-\infty, -\frac{1}{12})$. According to \eqref{h}, $\psi(t_0)=0$. Thus, it is suf\/f\/icient to prove that $\psi'(t)\neq 0$ on $(-\infty, -\frac{1}{12})$. Direct calculations (for $t\in\R$) yield
\begin{gather}\label{psi'}
\psi'(t)=\frac{1}{8t}\Re\sqrt{z^2_0(t)\big[z_0^2(t)-b^2\big]}=\frac{\Re\sqrt{(1+12t)\pm 2\sqrt{1+12t}}}{8\sqrt 3 t^2},
\end{gather}
where we used
\begin{gather*}
z_0^2(t)=\frac{-2\mp\sqrt{1+12t}}{3t},\qquad z_0^2(t)-b^2=\frac{\mp\sqrt{1+12t}}{t}.
\end{gather*}
When $t\in(-\infty, -\frac{1}{12})$, we have $\sqrt{1+12t}\in i\R$, so that the radical is in the right half plane. Thus, $\psi'(t)>0$, so that $\psi(t)<0$ on this interval. Therefore, $\nu$ cannot cross $t\in(-\infty, -\frac{1}{12})$. This estimate also shows that genus zero breaking curves are bounded.

Consider the interval $t\in (0,\frac 14)$, we distinguish the functions $\psi_l(t)$ and $\psi_u(t)$, corresponding to the lower/upper choice of signs in~\eqref{hz_0}, \eqref{hz_0u} (and also~\eqref{psi'}). It follows directly from \eqref{hz_0u} that $h^0(z_0(t),t)$ with the lower choice of the signs does not have branching at $t=0$ (corresponds to $u=1$ on the main sheet of~$t$). Thus, we have to study only $\psi=\psi_u$. A straightforward computation yields
\begin{gather*}
\psi_u\le(\frac 14\ri)= \ln \le(2\sqrt \frac 23 + \frac 13\ri)+ 2\sqrt \frac 23 >0\qquad {\rm and} \qquad \psi_u'(t)<0,
\end{gather*}
since both $z^2_0(t)$, $z_0^2(t)-b^2$ have argument $\pi$ on $ (0,\frac 14)$. Thus, $\nu$ cannot cross $(0,\frac 14)$, and we completed the proof of the global asymptotic portrait from Fig.~\ref{Generic}.

\section[Asymptotics of recurrence coef\/f\/icients in the higher genera regions]{Asymptotics of recurrence coef\/f\/icients\\ in the higher genera regions}\label{sect-hi-gen-sol}

The leading order behavior of the recurrence coef\/f\/icients and of the pseudonorms in the genus zero region and in full vicinities of the critical points $t_0$, $t_2$ was obtained in \cite{BT3, ArnoDu} and \cite{BT4} respectively. In this section we use Theorem~\ref{the0-mod-alt} to f\/ind the leading order behaviour of $\a_n(t)$ and ${\bf h}_n(t)$ in
regions of genus one and two.

\subsection{Genus 1 region}\label{sect-g1}
Let us start with the genus one region $\mathfrak{O}_1$ that contains the segment $(0,\qt)$. Explicit formulae for the corresponding branchpoints are given by
\eqref{bpseq_g=1}, where, thanks to the symmetry, we use notations $-\lambda_2$, $-\lambda_1$, $\lambda_1$, $\lambda_2$ for consecutive branchpoints along the Riemann--Hilbert contour instead of $\lambda_0$, $\lambda_1$, $\lambda_2$, $\lambda_3$.

The negatively oriented loop $\gt_{m,1}$ around the branchcut $\lambda_1$, $\lambda_2$ is the A cycle, whereas the loop $-\gt_{c,1}$ is the B cycle. The period~$\t$ and the normalized holomorphic dif\/ferential $\o$ are elliptic integrals
\begin{gather*}
\t=\frac{-\int_{-\lambda_1}^{\lambda_1}\frac{d\z}{R(\z)}}{2\int_{\lambda_1}^{\lambda_2}\frac{d\z}{R(\z)}}, \qquad
\o(z)=\frac{{dz}}{2\int_{\lambda_1}^{\lambda_2}\frac{d\z}{R(\z)}{R(z)}},
\end{gather*}
where
\begin{gather*}
R(\z)=\sqrt{\big(\z^2-\lambda_1^2\big)\big(\z^2-\lambda_2^2\big)}.
\end{gather*}
Direct calculations yield
\begin{gather}\label{taom-ex}
{2\int_{\lambda_1}^{\lambda_2}\frac{d\z}{R(\z)}}=-\frac{2\sqrt{t}}{\sqrt{1+2\sqrt{t}}}{\bf K}\le(\frac{2t^\qt}{\sqrt{1+2\sqrt{t}}}\ri),\qquad
\t=i\frac{{\bf K}\le(\sqrt{\frac{1-2\sqrt{t}}{1+2\sqrt{t}}}\ri)} {{\bf K}\le(\frac{2t^\qt}{\sqrt{1+2\sqrt{t}}}\ri)},
\end{gather}
where ${\bf K}$ denotes complete elliptic integral. Using the fact that $R(z)$ is an even function and~\eqref{t-varp}, we calculate
\begin{gather}\label{g1Abvar}
2\mathfrak{u}(\infty)=-\frac \t 2,\qquad \widetilde \varpi=2\mathfrak{u}(\infty)=-\frac \t 2.
\end{gather}

Let us normalize the traf\/f\/ics $\varrho_j$, see \eqref{orthog}, so that the traf\/f\/ic on $\g_{m,0}=(-\lambda_2,-\lambda_1)$ is~1. We then denote the traf\/f\/ic on $\g_{m,0}=(\lambda_1,\lambda_2)$ by $\rho_1$. Note that in the Generic case, considered in the paper, $\rho_1\neq 0$.

\begin{Lemma}\label{lem-g1-reg-nonex}
For any $t\in\mathfrak{O_1}$ the model RHP~\eqref{RHPPsi_0-hi} has a solution $($for all sufficiently large $N\in\N)$ if and only if
\begin{gather}\label{g1-reg-nonex}
\frac{\ln\rho_1 -i\pi}{2\pi i} \neq 0 \ \ \mod \ \frac\t 2\Z\oplus\Z,
\end{gather}
where $\t$ is given by~\eqref{taom-ex}.
\end{Lemma}

Lemma \ref{lem-g1-reg-nonex} follow from Corollary~\ref{the0-mod-alt} and~\eqref{g1Abvar}. Combining~\eqref{g1Abvar} with Theorem~\ref{the0-mod-alhgen} and using the properties~\eqref{propth} of the Riemann theta function, we obtain the theorem below. The constant
\begin{gather*}
\ell =-\hf + \frac 1{4t}-\hf\ln(t)+i\pi
\end{gather*}
can be easily obtained from \eqref{h1}.

\begin{Theorem}\label{theo-g1reg-sol}
If $t\in\mathfrak{O_1}$ is such that~\eqref{g1-reg-nonex} is satisfied then
\begin{gather}
 \a_n(t)=\frac{\sqrt{1-4t}-1}{2t} \nonumber\\
 \hphantom{\a_n(t)=}{} \times
 \begin{cases}
\displaystyle \rho_1e^{-2\pi i \frac \t 2}\frac{\Th^2\big(\frac{\ln\rho_1 }{2\pi i}\big)\Th^2(0)}{\Th^2\big(\frac{\ln\rho_1 }{2\pi i}-\frac \t 2\big)\Th^2\big(\frac \t 2\big)} +O\big(n^{-1}\big)& \text{if $n$ is odd}, \\
\displaystyle \rho_1^{-1}\frac{\Th^2\big(\frac{\ln\rho_1 }{2\pi i}-\frac \t 2\big)
\Th^2(0)}{\Th^2\big(\frac{\ln\rho_1 }{2\pi i}\big)\Th^2\big(\frac \t 2\big)}+O\big(n^{-1}\big)&
\text{if $n$ is even},
\end{cases}\label{g1regalph}\\
 {\bf h}_n(t)=\pi it^{-\frac n2}\frac{\sqrt{1+2\sqrt t}-\sqrt{1-2\sqrt t}}{ (-1)^{(n+1)}e^{\frac n2(1-\frac{1}{2t}+i\pi\t)}\sqrt t}\nonumber\\
 \hphantom{{\bf h}_n(t)=}{}\times
 \begin{cases}
\displaystyle e^{-\frac{\pi i \t}2 }\frac{\Th\big(\frac{\ln\rho_1 }{2\pi i}\big)\Th(0)}{\Th\big(\frac{\ln\rho_1 }{2\pi i}+\frac \t 2\big)\Th\big({-}\frac \t 2\big)} +O\big(n^{-1}\big)& \text{if $n$ is odd}, \\
\displaystyle \frac{\Th\big(\frac{\ln\rho_1 }{2\pi i}-\frac \t 2\big)
\Th(0)}{\Th\big(\frac{\ln\rho_1 }{2\pi i}\big)\Th\big({-}\frac \t 2\big)}+O\big(n^{-1}\big)&
\text{if $n$ is even}.
\end{cases}\nonumber
\end{gather}
\end{Theorem}

According to Theorem \ref{theo-g1reg-sol}, the even and odd subsequences of $\a_n(t)$ attain, in general, dif\/ferent limits. The special values of $\rho_1$, when these limits are equal, are provided in the following corollary.

\begin{Corollary}\label{cor-alp-lim} Under the same conditions as Theorem~{\rm \ref{theo-g1reg-sol}}, the requirement
\begin{gather}\label{equal-lim-traf}
\frac{\ln\rho_1}{2\pi i}=\left(\qt +k\right)\t +m, \qquad {\rm where} \quad k\in \Z, \quad m\in \hf\Z,
\end{gather}
implies that the even and odd subsequences of $\a_n(t)$ attain a single limit
\begin{gather}\label{ant_single}
 \a_n(t)=\frac{\sqrt{1-4t}-1}{2t} \frac{e^{2\pi i\le(m-\frac \t 4\ri)}\Th^2(0)}{\Th^2(\frac \t 2)}+O\big(n^{-1}\big).
\end{gather}
\end{Corollary}

\begin{proof}
 Denote $x=\frac{\ln\rho_1}{2\pi i}$. According to \eqref{g1regalph}, the limits of the even and odd subsequences of~$\a_n(t)$ coincide when
 \begin{gather}\label{eq-coinc}
 e^{2\pi i(2x- \frac \t 2)}\Th^4(x)=\Th^4\left(x-\frac \t 2\right).
 \end{gather}
 Substituting \eqref{equal-lim-traf} into \eqref{eq-coinc} and using properties \eqref{propth}, we obtain the f\/irst statement. Direct calculations also lead to \eqref{ant_single}.
\end{proof}

\subsection{Genus 2 solutions} \label{sect-whisk}

In this section we adjust the results of Theorem \ref{the0-mod-alhgen} to the case of $L=2$ and symmetric branchpoints with respect to the origin. Because of the symmetry, we denote the main arc containing zero as $\g_{m,0}$, its endpoints (branchpoints) as $\pm \lambda_0$, the remaining main arcs as $\g_{m,\pm 1}$ and their branchpoints as $\pm \lambda_1$, $\pm\lambda_2$ respectively. In the genus two region we distinguish two dif\/ferent conf\/igurations of the main arcs: the standard conf\/iguration shown on two panels in the middle of Fig.~\ref{Generic} and conf\/iguration with whiskers shown on the central and right upper panels, the second and third (from below) left panels, etc. These two regions are separated by two symmetrical genus zero breaking curves, connecting the critical points $t_0=-\frac 1 {12}$ and $t_2= \frac 14$; the whisker region is outside and the standard genus two region (together with genera zero and one regions) is inside these breaking curves.

We start with calculating the Riemann period matrix $\t$. We take negatively oriented loops around $\g_{m,1}$, $\g_{m,-1}$ as A cycles $\rm A_1$, $\rm A_2$ respectively, and oriented loops connecting $\lambda_{1}$ and $\lambda_0$ and $-\lambda_{1}$ and $-\lambda_0$ (on the main sheet and going back on the second sheet) as $\rm B$ cycles~$\rm B_1$, $\rm B_2$ respectively.

The normalized holomorphic dif\/ferentials $\o_j=\frac{P_j(z)dz}{R(z)}$, def\/ined by~\eqref{ao}, in the case of genus two and symmetric branchpoints have the form
\begin{gather} \label{g=2_vect_omega}
\o_1=\frac{\k_1z+\k_2}{R(z)}dz,\qquad \o_2=\frac{-\k_1z+\k_2}{R(z)}dz,
\end{gather}
where
\begin{gather*}
\k_1=\le(2\int_{\g_{m,1}}\frac{zdz}{R(z)}\ri)^{-1},\qquad \k_2=\le(2\int_{\g_{m,1}}\frac{dz}{R(z)}\ri)^{-1},
\end{gather*}
with
\begin{gather*}
R(z)=\sqrt{\prod_{j=0}^2\big(z^2-\lambda_j^2\big)}.
\end{gather*}
Here we used the fact that $R(z)$ is an odd function in $\C\setminus \mathfrak M$. Further, the symmetry of $\rm B$ cycles implies that $\t_{11}+\t_{21}=\t_{12}+\t_{22}$, where $\t_{j,k}$ are the entries of the Riemann period matrix~$\t$, see~\eqref{rpm}. Since $\t$ is symmetric, we obtain $\t_{11}=\t_{22}$, i.e., the two by two matrix $\t$ has equal diagonal entries.

Let us now consider the vector $\vec {\tilde\varpi}$ from the RHP \eqref{RHPtPsi_0}. Using~\eqref{t-varp} and Cauchy's residue theorem, we obtain
\begin{gather}\label{hW_g=2}
\vec {\tilde\varpi}=2\mathfrak u_\infty-\frac{1}{2\pi i}\int_{\hat{\mathfrak M}}\frac{\le(\frac{tz^4}{4}+\frac{z^2}{2}\ri )(\pm \k_1 z+\k_2)dz}{R(z)}=
2\mathfrak u_\infty+ \frac{1+\frac t4 \sum\limits_{j=0}^2\lambda_j^2}{\int_{\g_{m,1}}\frac{dz}{R(z)}}(1,1)^t,
\end{gather}
where, in view of \eqref{g=2_vect_omega}, the plus sign corresponds to $\tilde\varpi_1$ and the minus sign to~$\tilde\varpi_2$. According to \eqref{jumpassh'}, \eqref{y-curve}, in the genus two region
\begin{gather*}
h'(z)=-t\sqrt{\prod_{j=0}^2\big(z^2-\lambda_j^2\big)}=-tz^3\sqrt{\prod_{j=0}^2\le(1-\frac{\lambda_j^2}{z^2}\ri)}=-tz^3-z + \frac 2z +O\big(z^{-2}\big)
\end{gather*}
so that, similarly to \eqref{bpseq_g=1}, we obtain
\begin{gather}\label{lam-eq_g=2}
\sum_{j=0}^2\lambda_j^2=-\frac{2}{t}, \qquad \sum_{j<k}\lambda_j^2\lambda_k^2=\frac{1}{t^2}-\frac{4}{t}.
\end{gather}
The remaining third condition for the branchpoints $\pm\lambda_j$, $j=0,1,2$, is the Boutroux condition
\begin{gather*}
\int_{\lambda_0}^{\lambda_1}h'(z)dz=-t\int_{\lambda_0}^{\lambda_1}\sqrt{\prod_{j=0}^2\big(z^2-\lambda_j^2\big)}\in\R.
\end{gather*}

Equations \eqref{hW_g=2} and \eqref{lam-eq_g=2} yield
\begin{gather}\label{t_varpi_g2}
\vec {\tilde\varpi}=2\mathfrak u_\infty+ \frac{1}{2\int_{\g_{m,1}}\frac{dz}{R(z)}}(1,1)^t=2\mathfrak u_\infty+\vec \k_2,
\end{gather}
where $\vec \k_2=\k_2(1,1)^t$. According to \eqref{normal-traff}, the traf\/f\/ic on $\g_{m,0}$ is~$1$. Let $\vec {\r}=\frac 1{2\pi i} (\ln \r_1, \ln \r_2)^t $, where $\ln\r_j$ is the logarithmic term in the exponential of the jump matrix on $\g_{m,j}$ in the RHP~\eqref{RHPPsi_0-hi}.

Substituting \eqref{g=2_vect_omega}, \eqref{t_varpi_g2} into \eqref{alp_n}, \eqref{h_n} we obtain the following theorem.

\begin{Theorem}\label{theo-gen2-gen}
Let $\mathcal T$ be a compact subset in the genus two region. The recurrence coefficients and the pseudonorms satisfy
\begin{gather}
\a_n = (\lambda_0+\lambda_2-\lambda_1)^2\frac{\Th(2(N-1)\mathfrak{u}_\infty+N\vec\k_2+\vec {\r})\Th(2(N+1)\mathfrak{u}_\infty+N\vec\k_2+
\vec {\r})\Th^2(0)}{4\Th^2(2N\mathfrak{u}_\infty+N\vec\k_2+\vec{\r})\Th^2(2\mathfrak{u}_\infty)}\nonumber\\
\hphantom{\a_n =}{} +O\big(N^{-1}\big), \nonumber\\
\h_n = \pi(\lambda_0+\lambda_2-\lambda_1)\frac{\Th(2(N+1)\mathfrak{u}_\infty+N\vec\k_2+\vec {\r})\Th(0)}{\Th(2N\mathfrak{u}_\infty+N\vec\k_2+\vec{\r})\Th(2\mathfrak{u}_\infty)} e^{N\ell}+O\big(N^{-1}\big),\label{g2regalph}
\end{gather}
as $N\ra\infty$ uniformly in $t\in \mathcal T$ provided that
\begin{gather*}
\Th\big(2N\mathfrak{u}_\infty+N\vec\k_2+\vec{\r}\big) \qquad \text{is separated from $0$}.
\end{gather*}
\end{Theorem}

\begin{Remark} \label{rem-orient}
If we change the orientation of the contour $\g_{m,2}$ in the model RHP \eqref{RHPPsi_0-hi}, then, according to \eqref{alp_n}, \eqref{h_n}, we will have to replace in \eqref{g2regalph} the prefactors $(\lambda_0+\lambda_2-\lambda_1)$ by~$\lambda_0$, the vector $\vec\rho$ by $\vec\rho+(0,1)^t$ and the real part $\Re \tau=\hf \1$ of the period matrix for all the theta functions in~\eqref{g2regalph} by $\Re \tau =\hf \s_3$. For justif\/ication and more details see \cite[Section~3]{BT6}.
\end{Remark}

\subsection*{Acknowledgements}
Work supported in part by the Natural Sciences and Engineering Research Council of Canada (NSERC).

\LastPageEnding

\end{document}